\documentclass[sn-mathphys]{sn-jnl_arxiv}

\usepackage{graphicx}
\graphicspath{{./}}
\usepackage[utf8]{inputenc}
\usepackage{dsfont}
\usepackage{subfigure}
\usepackage{color}
\usepackage{url}
\usepackage{enumerate} 
\usepackage{booktabs}
\usepackage{tabularx}
\usepackage{ltablex}
\usepackage{mathptmx}
\usepackage{amssymb,amsmath,amsfonts,latexsym} 
\usepackage{longtable} 
\usepackage{textcomp}
\usepackage{lscape}
\usepackage[12pt]{extsizes}
\usepackage{geometry}
\usepackage{ifthen}
\usepackage{xcolor}
\usepackage{mathtools}
\mathtoolsset{showonlyrefs}

\jyear{2023}%

\newtheorem{theorem}{Theorem}{\bfseries}{\itshape}
{\bfseries}{\itshape}
\newtheorem{lemma}[theorem]{Lemma}{\bfseries}{\itshape}
\newtheorem{proposition}[theorem]{Proposition}{\bfseries}{\itshape}
{\bfseries}{\itshape}
\newtheorem{remark}[theorem]{Remark}{\bfseries}{\upshape}
\newtheorem{assumption}[theorem]{Assumption}{\bfseries}{\itshape}
\DeclareMathAlphabet{\pazocal}{OMS}{zplm}{m}{n}
\let\mathcal\pazocal
\counterwithin{theorem}{section}

\numberwithin{figure}{section}

\renewcommand{\paragraph}[1]{{\smallskip\noindent\textbf{#1.}}}
\definecolor{myred}{rgb}{0.8,0,0}  
\definecolor{orange}{rgb}{1,0.53,0}
\definecolor{mygreen}{rgb}{0,0.5,0}
\definecolor{myviolet}{rgb}{0.58,0,0.83}

{\vskip\baselineskip\noindent\textbf{Proof of {#1}:}}%
{\hspace*{.1pt}\hspace*{\fill}\BOX\vskip\baselineskip}

\def \R{\mathbb{R}}               
\def \N{\mathbb{N}}               
\def \P{\mathbb{P}}             
\def \1{{\bf 1}}                
\def \0{{\bf 0}}
\def\utility{\mathcal U}

\def\revlevel{\overline{\mu}}
\def\revspeed{\kappa}
\def\zustand{m}
\def\driftinitial{\overline{\zustand}_0}
\def\filterinitial{\zustand_0}
\def\covinitial{\overline{\variance}_0}
\def\condcovinitial{\variance_0}
\def\Sigmak{\Delta\Qpro_{t_k}^{\HC}}
\def\minusSigmak{{+}}
\def\plusSigmak{{-}}

\def\Abound{A^F}
\def\Bbound{B^F}
\def\Cbound{C^F}

\def\cpsi{\gamma}    

\def\variance{q}
\def\pointp{p}

\def\HF{F}
\def\HR{R}
\def\HD{J}
\def\contexp{J}
\def\HC{Z}
\def\nAktien{d}
\def\nExperten{n}
\def\nWienerRendite{d_1}
\def\nWienerDrift{d_2}
\def\nWienerExperten{d_3}

\def\varianceexp{\Gamma}

\def\volR{\sigma_R}
\def\voldrift{\sigma_{\mu}}
\def\volexp{\sigma_{\HD}}
\def\alphamm{\alpha}
\def\betam{\beta}

\def\drift{\mu}

\def\Radon{\Lambda}

\def\generator{\mathcal L}
\def\reward{D}
\def\valuefkt{V}
\def\valueorigin{\mathcal V}
\def\vterm{v}
\def\rewardorigin{\mathcal D}
\def\rewardconstant{C_{\valueorigin}}
\def\wealth{X}

\def\Jpi{\overline \rewardorigin_0}

\def\Mpro{M}
\def\Qpro{Q}
\def\Mstate{\Mpro}   

\def\qed{\hfill$\Box$}

\newcommand\norm[1]{\left\lVert#1\right\rVert}

\def \trace{\operatorname{tr}}
\def \E{{\mathbb{E}}}

\def \Ebar{ \overline{\E}^H}  
\def \EbarZ{\overline{\E}^\HC}  

\def\Ashort{{\widetilde A}}
\def\Bshort{{\widetilde B}}
\def\Cshort{{\widetilde C}}
\def\diffop{D}
\def\Update{{\Psi}}

\begin{document}
	
	\title[Portfolio Optimization with Expert Opinions at  Fixed Arrival Times]{Power Utility Maximization with Expert Opinions at Fixed Arrival Times in a Market with Hidden Gaussian	Drift}


\author[1]{\fnm{Abdelali} \sur{Gabih}}\email{a.gabih@uca.ma}

\author[2]{\fnm{Hakam} \sur{Kondakji}}\email{kondakji@hsu-hh.de}

\author*[3]{\fnm{Ralf} \sur{Wunderlich}}\email{ralf.wunderlich@b-tu.de}

\affil[1]{ \orgname{Equipe de Modélisation et Contrôle des Systèmes Stochastiques et Déterministes, 
		Faculty of Sciences,  Chouaib Doukkali University}, \orgaddress{ \postcode{24000}, \state{El Jadida}, \country{Morocco}}}

\affil[2]{\orgdiv{Faculty of Economics und Social Sciences}, \orgname{Helmut Schmidt University}, \\\orgaddress{\street{P.O. Box 700822}, \city{Hamburg}, \postcode{22008},  \country{Germany}}}

\affil*[3]{\orgdiv{Institute of Mathematics}, \orgname{Brandenburg University of Technology Cottbus-Senftenberg}, \orgaddress{\street{P.O. Box 101344}, \city{Cottbus}, \postcode{03013}, \country{Germany}}}


\abstract{
	In this paper we study optimal trading strategies in a financial market in which   stock returns depend on a hidden
	Gaussian mean reverting drift process.  Investors obtain information on that drift  by observing stock returns. Moreover, expert opinions in the form
	of signals about the current state of the drift arriving at fixed and known dates are included in the analysis.		
	Drift estimates are based on Kalman filter techniques. They are used to transform a power utility maximization problem under partial information into an optimization problem under full
	information where the state variable is the filter of the drift. The dynamic programming equation for this problem is studied and closed-form solutions for the value function and the optimal trading strategy of an investor are derived. They allow to quantify the monetary value of information delivered by the expert opinions. We illustrate our theoretical findings by  results of extensive numerical experiments.}




\maketitle

\section{Introduction}
In dynamic portfolio selection problems optimal trading strategies depend crucially on
the drift of the underlying asset return processes. That drift describes the expected asset returns, varies over time and is driven by certain stochastic  factors such as dividend yields, the  firm’s return on
equity, interest rates and macroeconomic  indicators. This was already addressed in the seminal paper of Merton \cite[Sec.~9]{Merton (1971)}. We also refer to other  early articles  by  Bielecki and Pliska \cite{Bielecki_Pilska (1999)},  Brennan, Schwartz, Lagnado \cite{Brennan et al (1997)}  and Xia \cite{Xia (2001)}. The dependence of the drift process on these factors is usually  not perfectly known and some of the factors may be not directly observable. Therefore, it is reasonable to model the drift as an unobservable stochastic process for which only statistical estimates are available. Then  solving the associated portfolio problems has to be based on such estimates of the drift process.

For the one-period Markowitz model the surprisingly large  impact of statistical errors in the estimation of model parameters on mean-variance optimal portfolios is often reported in the literature, e.g.~by Broadie \cite{Broadie (1993)}.  Estimating the drift  with a reasonable degree of precision based only on historical asset prices is known to be almost impossible. This is nicely shown  in Rogers \cite[Chapter 4.2]{Rogers (2013)} for a model in which the drift is even constant.  For a reliable estimate extremely long time series of data are required which are usually not available. Further, assuming a constant drift over longer periods of time is quite unrealistic. Drifts tend to fluctuate 	randomly over time and drift effects are often overshadowed by volatility.

For	these reasons, portfolio
managers and traders try to diversify their  observation sources and also rely on external sources of information such as news, company reports, ratings and benchmark values. Further, they increasingly turn to  data outside of the traditional sources that companies and financial markets provide. Examples are social media posts, internet search, satellite imagery, sentiment indices, pandemic data,  product review trends and are often related to Big Data analytics.
Finally,  they use views of  financial analysts or just their own intuitive views on the future asset performance.

In the literature these outside sources of information are coined  \textit{expert opinions} or more generally \textit{alternative data}, see Chen and Wong \cite {Chen Wong (2022)},  Davis and Lleo \cite{Davis and Lleo (2022)}. In this paper we will use the first term. After an appropriate mathematical modeling as additional noisy observations they are included  in the drift estimation and the construction of optimal portfolio strategies.	That approach goes back to the celebrated Black-Litterman model which is an	extension of the classical one-period Markowitz model, see Black and
Litterman~\cite{Black_Litterman (1992)}. The idea is to improve return predictions using expert opinions by means of a  Bayesian updating of the drift estimate.

Instead of a static one-period model  we consider in this paper a continuous-time model for asset prices. Additional information in the form
of expert opinions arrive repeatedly over time. Davis and Lleo  \cite{Davis and Lleo (2013_1)} coined this approach ``Black--Litterman in Continuous Time'' (BLCT). More precisely,   we study a  hidden Gaussian model   where asset returns are driven by an unobservable mean-reverting Gaussian process.
Information on the drift is of mixed type. First investors observe stock prices or equivalently the return process. Moreover, investors may have access to expert opinions arriving at already known discrete time points in a form of unbiased drift estimates.  Since the investors' ability to construct good trading strategies depends on the quality of the hidden drift estimation we study a filtering problem. There the aim is to find the conditional distribution of the drift given the available information drawn from the return observations and expert opinions.

For investors who observe only the return process that filter known as the classical  Kalman--Bucy filter, see for example Liptser and Shiryaev~\cite{Liptser-Shiryaev}. Based on this one can derive the filter for investors who also observe expert opinions by a Bayesian update of the  current drift estimate at each information date. This constitutes  the above mentioned   continuous-time version of the static Black--Litterman approach.

\smallskip
Utility maximization problems for partially informed investors have been intensively studied in the last years. For models with Gaussian drift we refer to to Lakner \cite{Lakner (1998)} and  Brendle \cite{Brendle2006}. Results for models in which the drift is described by a continuous-time hidden Markov chain  are given in Rieder and Bäuerle \cite{Rieder_Baeuerle2005}, Sass and Haussmann \cite{Sass and Haussmann (2004)}  and more recently by  Chen and Wong \cite {Chen Wong (2022)}. A good overview with further references and generalization  can be found  in Björk et al.~\cite{Bjoerk et al (2010)}. 

For the literature on BLCT in which expert opinions are included we refer to a series of papers \cite{Gabih et al (2014),Gabih et al (2019) FullInfo,Sass et al (2017),Sass et al (2021),Sass et al (2022)} of the present authors and of Sass and Westphal. They investigate utility maximization problems for investors with logarithmic preferences in market models with a hidden Gaussian drift process  and discrete-time expert opinions. The case of continuous-time expert opinions and power utility maximization is treated in a series of papers by Davis and Lleo, see \cite{Davis and Lleo (2013_1),Davis and Lleo (2020),Davis and Lleo (2022)} and the references therein.  Power utility maximization problems for expert opinions arriving randomly at the jump times of a Poisson process  are treated in the recent work \cite{Gabih et al (2022) PowerRandom}. Similar portfolio problems for drift processes described by  continuous-time hidden Markov chains have been studied in  Frey et al.~\cite{Frey et al. (2012),Frey-Wunderlich-2014}. Finally,  we refer to our companion paper \cite{Gabih et al (2022) Nirvana} where we investigate the well posedness of power utility maximization problems which are addressed in the present paper.

\smallskip
 \paragraph{Our contribution}
	The new contribution to the literature in the present paper is the detailed analysis of the case of power utility maximization for market models with a hidden Gaussian drift and discrete-time expert opinions arriving at fixed and known time points. That case  was not yet treated in the literature and is only considered in the PhD thesis of Kondakji \cite{Kondkaji (2019)} on which this paper is based. Only recently, a few results of  \cite{Kondkaji (2019)} were briefly mentioned  in  Sass et al.~\cite{Sass et al (2022)} and applied in numerical experiments.   Note that the approach in  \cite{Gabih et al (2022) PowerRandom} for the case of randomly arriving expert opinions cannot be adopted to the case of fixed arrival times considered in this paper. In  \cite{Gabih et al (2022) PowerRandom} the dynamic programming equation appears as a partial integro-differential equation which requires a numerical solution. Instead, in this paper we can derive closed-form solutions to the derived dynamic programming equation. 

Our main results are presented in Theorems \ref{backward_recursion_I} and \ref{backward_recursion_II}. Here we present a backward recursion for the value function and the optimal strategy of the partially informed investors observing also discrete-time expert opinions,  and discuss verification issues.	Another contribution are results of extensive numerical experiments which we present in Sec.~\ref{numeric_result}. There we study in particular the asymptotic properties of filters, value functions and optimal strategies for high-frequency experts and the monetary value of expert opinions.   These studies on high-frequency experts and their limiting behavior provide a link to the results for continuous-time expert opinions which are available from the works of Davis and Lleo, see \cite{Davis and Lleo (2013_1),Davis and Lleo (2020),Davis and Lleo (2022)}.

\smallskip	
The paper is organized as follows. In Sec.~\ref{market_model} we
introduce the model for our financial market including the expert
opinions and define information regimes for investors with different
sources of information. Further,  we formulate the portfolio optimization problem. Sec.~\ref{Filtering} states for the different information regimes the filter equations for the corresponding conditional mean and conditional covariance process. Then it reviews properties of the filter, in particular the asymptotic filter behavior for high-frequency expert opinions. Sec.~\ref{Utility_Max}
is devoted to the solution of the power utility 	maximization problem. That problem   is reformulated as an equivalent stochastic optimal control problem which can be solved by dynamic programming techniques. Solutions are presented for the fully informed investor.  Sec.~\ref{UtilityMaxDiffusion}  presents the solution  for partially informed investors  combining return observations with diffusion type expert opinions and Sec.~\ref{UtilityMaxDiscreteExperts} studies the case of investors observing  discrete-time experts.
Sec.~\ref{numeric_result} illustrates the theoretical findings by numerical results.  

\medskip
\paragraph{Notation} Throughout this paper, we use the notation $I_d$ for the identity matrix in $\R^{d\times d}$,  $0_{d}$
denotes the null vector in $\R^d$, $0_{d\times m}$ the null matrix in $\R^{d\times m}$. For a symmetric and positive-semidefinite matrix $A\in\R^{d\times d}$ we call a symmetric and positive-semidefinite matrix $B\in\R^{d\times d}$ the \emph{square root} of $A$ if $B^2=A$. The square root is unique and will be denoted by $A^{1/2}$. For a vector $X$ we denote by $\norm{X}$ the Euclidean norm.  Unless stated otherwise, whenever $A$ is a matrix, $\lVert A\rVert$ denotes the spectral norm of $A$.

\section{Financial Market and Optimization Problem}
\label{market_model}
\subsection{Price Dynamics}
\label{PriceDynamics}  We model a financial market with one
risk-free and multiple risky assets. The  setting is based on Gabih
et al.~\cite{Gabih et al (2014),Gabih et al (2019) FullInfo} and
Sass et al.~\cite{Sass et al (2017),Sass et al (2022),Sass et
	al (2021)}.
For a  fixed date $T>0$ representing the investment horizon, we work
on a filtered probability space $(\Omega,\mathcal{G},\mathbb{G},\P)$,
with filtration $\mathbb{G}=(\mathcal {G}_t)_{t \in [0,T]}$
satisfying the usual conditions. All processes are assumed to be
$\mathbb{G}$-adapted.

 We consider a market model  for one risk-free asset 
	and $\nAktien$ risky securities.  We follow an  approach frequently used in the literature on optimal portfolio selection and consider  discounted asset prices with  the risk-free asset as numéraire. Then the  risk-free asset has a 	constant price $S^0_t=1$.
	The excess  log returns or risk premiums $R=(R^{1},\ldots,R^{\nAktien})$ of the risky securities  are described by stochastic processes defined by the SDE
	\begin{align}
		dR_t=\mu_t\, dt+\volR\, dW^{R}_t, \label{ReturnPro}
	\end{align}
	driven by  a $\nWienerRendite$-dimensional $\mathbb{G}$-adapted	Brownian motion $W^{\HR}$ with  $\nWienerRendite\geq\nAktien$. In the remainder of this paper we will call $R$ simply \textit{returns}.    $\mu=(\mu_t)_{t\in[0,T]}$ denotes the stochastic drift process which is described in detail below. The
volatility matrix $\volR\in\mathbb
R^{\nAktien\times\nWienerRendite}$ is  assumed to be constant over
time such that $\Sigma_{R}:=\volR\volR^{\top}$ is positive definite. 
In this setting the  discounted price process $S=(S^1,\ldots,S^{\nAktien})$ of
the risky securities reads as
\begin{align}
	dS_t&=diag(S_t)\, dR_t,~~ S_0=s_0, \label{stockmodel}
\end{align}
with some fixed initial value $s_0=(s_0^1,\ldots,s_0^d)$. Note
that for the solution to the above SDE it holds
\begin{align*}
	\log S_t^{i}-\log s_0^{i} &= \int\limits_0^t \drift_s^{i}ds
	+\sum\limits_{j=1}^{\nWienerRendite}\Big(
	\sigma_R^{ij}W_t^{R,j}-\frac{1}{2} (\sigma_R^{ij})^2 t\Big)
	=R_t^{i}-\frac{1}{2}\sum\limits_{j=1}^{\nWienerRendite}
	(\sigma_R^{ij})^2 t ,\quad i=1,\ldots,\nAktien.
\end{align*}
So we have the equality $\mathbb{G}^R = \mathbb{G}^{\log S} =
\mathbb{G}^S$, where  for a generic process $X$ we denote by
$\mathbb{G}^X$ the filtration generated by $X$. This is useful since
it allows to work with $R$ instead of $S$ in the filtering part. 

The dynamics of the drift process $\mu=(\mu_t)_{t\in[0,T]}$ in \eqref{ReturnPro}
are
given by the stochastic differential equation (SDE)
\begin{eqnarray}
	\label{drift} d\mu_t=\revspeed(\revlevel-\mu_t)\, dt+\voldrift\,
	dW^{\mu}_t,
\end{eqnarray}
where $\revspeed\in\mathbb R^{\nAktien\times\nAktien}$,
$\voldrift\in\mathbb R^{\nAktien\times\nWienerDrift}$ and
$\revlevel\in\mathbb R^{\nAktien} $ are constants such that 
all eigenvalues  $\revspeed$ have a positive real part (that is, $-\kappa$ is a stable matrix)  and $\Sigma_{\mu}:=\voldrift\voldrift^{\top}$
is positive definite. Further, $W^{\mu}$ is a
$\nWienerDrift$-dimensional Brownian motion such that
$\nWienerDrift\geq\nAktien$.   For the sake of simplification and
shorter notation we assume  that the Wiener processes $W^{R}$ and
$W^{\mu}$ driving the return and drift process, respectively, are
independent. We refer to Brendle \cite{Brendle2006},  Colaneri
et al. \cite{Colaneri et al (2021)} and  Fouque et al.~\cite{Fouque et al. (2015)} for
the general case. Here, $\revlevel$ is the mean-reversion level,
$\revspeed$ the mean-reversion speed and $\voldrift$ describes the
volatility of $\mu$. The initial value $\drift_0$ is assumed to be a
normally distributed random variable independent of $W^{\mu}$ and
$W^{R}$ with mean $\driftinitial\in \R^{\nAktien}$ and covariance
matrix $\covinitial\in\mathbb R^{\nAktien\times\nAktien}$ assumed to
be symmetric and  positive semi-definite.  It is well-known 
that the solution to  SDE \eqref{drift}  is known  as
Ornstein-Uhlenbeck
process which is  a  Gaussian process  given by  
\begin{align}
	\label{mu_explicit}
	\mu_t &= \revlevel +e^{-\revspeed t}\Big[(\mu_0 -\revlevel) +
	\int_0^t e^{\revspeed s} \voldrift dW^{\mu}_s\Big],\quad t\ge 0.
\end{align}

\subsection{Expert Opinions}
\label{Expert_Opinions}  We assume that investors observe the return
process $R$ but they neither observe the factor process $\mu$ nor
the Brownian motion $W^{R}$. They do however know the model
parameters such as $\volR,\revspeed, \revlevel, \voldrift $  and
the distribution $\mathcal{N}(\driftinitial,\covinitial)$ of  the
initial value $\drift_0$. Information about the drift $\mu$ can be
drawn from observing the returns $R$. A special feature of our model
is that investors may also have access to additional information
about the drift in form of \textit{expert opinions} such as news,
company reports, ratings or their own intuitive views on the future
asset performance. The expert opinions provide  noisy signals about
the current state of the drift arriving at known deterministic time points
$0=t_0<t_1<\ldots<t_{n-1}<T$. For the sake of convenience we also write $t_n =
T$ although no expert opinion arrives at time $t_n$.
The signals or ``the expert views'' at time $t_k$ are modelled by
$\R^\nAktien$-valued  Gaussian random vectors
$Z_k=(Z_k^1,\ldots,Z_k^{\nAktien})^{\top}$ with
\begin{align}
	\label{Expertenmeinungen_fest}
	Z_k=\drift_{t_k}+{\varianceexp}^{\frac{1}{2}}\varepsilon_k,\quad
	k=0,\ldots,n-1,
\end{align}
where the matrix  $\varianceexp\in\R^{\nAktien\times\nAktien}$ is
symmetric and positive definite.
Further,  $(\varepsilon_k)_{k=0,...,n-1}$ is a sequence of independent
standard normally distributed random vectors, i.e.,
$\varepsilon_k\sim~\mathcal{N}(0,I_d)$. It is  also independent of
both the Brownian motions $W^R, W^\mu$ and the initial value $\mu_0$
of the drift process. That means that, given $\mu_{t_k}$, the expert
opinion $Z_k$ is $\mathcal{N}(\mu_{t_k},\varianceexp)$-distributed.
So, $Z_k$ can be considered as an unbiased estimate of the unknown
state of the drift at time $t_k$. 

Modeling expert opinions as normally distributed random variables corresponds well  to a variety of additional information on average stock returns available in real-world markets.  We refer to  Davis and Lleo \cite{Davis and Lleo (2020)} for more details about an appropriate preprocessing, debiasing and approximation of such extra information by Gaussian models. Let us  briefly sketch the mathematical modeling of analyst views in terms of confidence intervals. Inspired by \cite{Davis and Lleo (2020)}  we consider the following example of a view at time $t=t_k$:
%
\textit{
	``My research leads me to believe that the average stock return lies within a	range of $6\%$ to $10\%$, and I'm $90\%$ confident about this''.}
%
This view can be treated as a $90\%$-confidence interval for the unknown mean $\mu_{t_k}$ of a Gaussian distribution centered at $0.08$,  which is the observed $Z_k$. The corresponding variance $\Gamma$ is chosen such that the boundaries of the interval are $0.06$ and $0.10$.
We also want to emphasize that the Gaussian expert opinions allow to work with  Kalman filtering techniques. For other distributions, in general no closed-form filters are available. 

 The matrix $\varianceexp$ is a
	measure of the expert's reliability.  The diagonal entries  $\varianceexp_{ii}$ are just the variances of the expert's
	estimates of the drift for the $i$-th asset at time $t_k$: the larger $\varianceexp_{ii}$ the less
	reliable is the expert's view about $\mu_{t_k}^i, i=1,\ldots,d.$  The off-diagonal entries describe the correlation between the experts’
	views.  An example for a  diagonal matrix $\varianceexp$, i.e., uncorrelated views $Z_k^i$ about the  $i$-th asset's drift, is obtained if the random vector  $Z_k$ contains  the views $Z_k^i$ of  $d$ independent analysts estimating the drift of a single asset only. That case is hard to justify in reality where one can observe so-called ``groupthink'' leading to positive correlations between the views of the analysts. 	However correlations between the views are hard to calibrate. We refer to Davis and Lleo \cite{Davis and Lleo (2020)} for more details. 

The above model of discrete-time expert opinions can be modified such that expert opinions arrive not at fixed and known dates but at random times $(T_k)_{k\in\N}$. That approach together with results for filtering  and maximization of log-utility was studied in detail in Sass et al.~\cite{Sass et al (2021)}. There the arrival times are modeled as the jump times of a Poisson process. 	The  maximization of power utility   is considered in \cite{Gabih et al (2022) PowerRandom}.

One may also allow for relative expert views where experts
give an estimate for the difference in the drift of two stocks
instead of absolute views. This extension  can be studied in
Sch\"ottle et al.~\cite{Schoettle et al. (2010)} where the authors
show how to switch between these two models for expert opinions by
means of a pick matrix.


In addition to expert opinions arriving  at discrete time points  we
will also consider expert opinions arriving continuously over time  as in
Davis and Lleo \cite{Davis and Lleo (2013_1),Davis and Lleo (2020)} who called this approach ``Black--Litterman in Continuous Time''. This is motivated by the results of Sass et al.~\cite{Sass et al (2021),Sass et al (2022)} who show that asymptotically as the arrival frequency tends to infinity
and the  expert's variance $\Gamma$ grows linearly in that frequency   the information drawn from these expert
opinions is  essentially the same as the information one gets from
observing yet another diffusion process. This diffusion process can
then be interpreted as an expert who gives a continuous-time
estimation about the current state of the drift.  Another interpretation is that the diffusion process models returns of assets which are not traded in the portfolio but depend on the same stochastic factors and are observable by the investor. Let these continuous-time expert opinions 
be given by the diffusion process
\begin{align}\label{continuous-expert}
	d\contexp_t = \mu_t dt +\volexp dW_t^{\HD}
\end{align}
where $W_t^{\HD}$ is a $\nWienerExperten$-dimensional Brownian
motion independent of $W_t^R$ and $W^{\mu}$ and such that with
$\nWienerExperten\geq\nAktien$. The volatility   matrix
$\volexp\in\mathbb R^{\nAktien\times\nWienerExperten}$ is assumed to
be constant over time such that the matrix
$\Sigma_{\HD}:=\volexp\volexp^{\top}$ is positive definite. In
Subsecs.~\ref{LogUtility} and \ref{Numerik_Konvergenz_Wertfunktion} we show that  
based on this model and on the diffusion approximations provided in
\cite{Sass et al (2022)}   one can find efficient approximative solutions to
utility maximization problems for partially informed investors
observing high-frequency discrete-time expert opinions.

\subsection{Investor Filtration}
\label{Investor_Filtration}  We consider various types of investors
with different levels of information. The information available to
an investor is described by the \textit{investor filtration}
$\mathbb{F}^H=(\mathcal{F}^H_t)_{t\in[0,T]}$. Here, $H$ denotes the
information regime for which we
consider the cases $H=\HR,\HC,\HD,\HF$  and the investor with filtration
$\mathbb{F}^H=(\mathcal{F}^H_t)_{t\in[0,T]}$ is called the $H$-investor.
The $\HR$-investor observes only the return process $R$, the $\HC$-investor combines
return observations with the discrete-time expert opinions
$Z_k$ while  the $\HD$-investor observes the return process together with the continuous-time expert $\HD$. Finally,  the $\HF$-investor has full information  and can observe the drift process $\mu$  directly and of course the return process.
For stochastic drift this case is not realistic, but we use it as a
benchmark.  It will serve as a limiting case for high-frequency
expert opinions with fixed covariance matrix $\varianceexp $,  see Subsec.~\ref{filter_high_frequency}, Thm.~\ref{theorem_Konv_Fn}.

The $\sigma$-algebras $\mathcal{F}^H_t$ representing the  $H$-investor's information  at time $t\in[0,T]$ are defined at initial time  $t=0$ by $\mathcal{F}^\HF_0=\sigma\{\drift_0\}$ for the fully informed investor and by  $\mathcal{F}^H_0=\mathcal{F}^I_0\subset \mathcal F_0^{\HF}$ for $H=\HR,\HC,\HD$, i.e., for the partially informed investors. Here, $\mathcal{F}_0^I$ denotes the $\sigma$-algebra representing  prior information about the initial drift $\mu_0$. More details on $\mathcal{F}_0^I$ are given below. For $t\in (0,T]$ we define 
\[\begin{array}{rl}
	\mathcal {F}_t^{\HR}&=\sigma(R_s,~ s\le t)  \vee \mathcal{F}_0^I, \\[0.5ex]
	\mathcal {F}_t^{\HC}&=\sigma(R_s, s\le t,\,~(t_k,Z_k),~ t_k\le t) \vee \mathcal{F}_0^I, \\[0.5ex]
	\mathcal {F}_t^{\HD}&=\sigma(R_s,\contexp_s,~ s\le t) \vee \mathcal{F}_0^I, \\[0.5ex]
	\mathcal {F}_t^{F}&=\sigma(R_s, \mu_s,~ s\le t).
\end{array}
\]
We assume that the above $\sigma$-algebras $\mathcal{F}_t^H$
are augmented by the null sets 
of $\P$.

Note that all partially informed investors ($H=\HR,\HD, \HC$) start  at $t=0$ with the same initial information given by
$\mathcal{F}_0^I$. The latter models
prior knowledge about the drift process at time $t=0$, e.g., from
observing  returns  or expert opinions in the past, before the
trading period $[0,T]$.  The expert opinion $Z_0$ arriving at time $t=0$ does not belong to this prior information and is therefore excluded from $\mathcal {F}_0^{\HC}$ and only contained in $\mathcal {F}_t^{\HC}$ for $t>0$. At first glance this may appear not consistent but it will facilitate below in Subsec.~\ref{MonetaryValue} the formal definition of the monetary value of the expert opinions. 

We assume that the conditional distribution
of the initial value drift $\mu_0$ given $\mathcal{F}_0^{ I}$ is the
normal distribution $\mathcal{N}(m_0,q_0)$ with mean
$\filterinitial\in \R^{\nAktien}$ and covariance matrix
$\condcovinitial\in\mathbb R^{\nAktien\times\nAktien}$ assumed to be
symmetric and  positive semi-definite.
In this setting typical examples are:
\begin{enumerate}
	\renewcommand{\labelenumi}{\alph{enumi})}
	\item  The investor has no  information about the initial value of the drift $\mu_0$. However, he knows  the model parameters,
	in particular the distribution
	$\mathcal{N}(\driftinitial,\covinitial)$ of $\mu_0$ with  given
	parameters $\driftinitial$ and $\covinitial$. This corresponds to
	$\mathcal{F}_0^I=
	\{\varnothing,\Omega\}$ and
	$\filterinitial=\driftinitial$, $\condcovinitial=\covinitial$.
	
	\item  The investor can fully observe the initial value of the drift $\mu_0$, which corresponds
	to $\mathcal{F}_0^I=\mathcal{F}_0^F =\sigma\{\mu_0\}$ and
	$\filterinitial=\drift_0(\omega)$ and $\condcovinitial=0$.
	
	\item  Between the above limiting cases we consider an investor who has some prior but no complete information about
	$\mu_0$ leading to $\{\varnothing,\Omega\}\subset \mathcal
	{F}_0^I\subset\mathcal F_0^F $.
\end{enumerate}

\subsection{Portfolio and Optimization Problem}
We describe the self-financing trading of an investor by the initial
capital $x_0>0$ and the  $\mathbb{F}^H$-adapted trading strategy
$\pi=(\pi_t)_{t\in[0,T]} $ where $\pi_t\in\R^{\nAktien}$. Here
$\pi_t^{i}$ represents the proportion of wealth invested in the
$i$-th stock at time $t$.  The assumption that $\pi$ is
$\mathbb{F}^H$-adapted models that investment decisions have to be
based on information available to the $H$-investor which he obtains
from from observing assets prices ($H=R$) combined with expert
opinions ($H=\HC,\HD$) or with the drift process ($H=F$). Following
the strategy $\pi$   the investor generates a wealth process
$(X_t^{\pi})_{t\in [0,T]}$ whose dynamics  reads as
\begin{align} \label{wealth_phys}
	\frac{dX_t^{\pi}}{X_t^{\pi}}= \pi_t^{\top}dR_t = 
	\pi_t^{\top}\mu_t\; dt+\pi_t^{\top}\volR\; dW_t^{R},\quad
	X_0^{\pi}=x_0.
\end{align}
We denote by 
\begin{align}
	\begin{aligned}
		\label{set_admiss_0} \mathcal{A}_0^H =\Big\{&\pi= (\pi_t)_{t}  \colon
		\pi_t\in\mathbb R^{\nAktien}, \text{ $\pi$ is $\mathbb{F}^H$-adapted
		}, X^\pi_t > 0, \E\Big[ \int\nolimits_0^T \|\pi_t\|^2\, dt
		\Big]<\infty, \\
		&  \text{the process } \Radon^H \text{ defined below in \eqref{Radon} satisfies } \E\big[\Radon^H_T\big]=1  \Big\}
	\end{aligned}
\end{align}
the  class of {\em admissible trading strategies}.  The last condition in the definition of $\mathcal{A}_0^H$ is  needed to apply  a change of measure technique for the solution of the optimization problem. More  details  are provided below in Subsec.~\ref{PowerUtility}. 

We assume that
the investor wants to maximize the expected utility of terminal
wealth for a given utility function $U : \R_+\rightarrow\R$
modelling the risk aversion of the investor. In our approach we will
use  the function
\begin{align}
	\label{util_def}
	\utility_{\theta}(x):=\frac{x^{\theta}}{\theta},\quad
	\theta\in(-\infty,0)\cup(0,1).
\end{align}
The limiting case for $\theta\rightarrow 0$ for the power utility
function leads to the logarithmic utility $\utility_0(x):=\ln x$,
since we have
$\utility_{\theta}(x)-\frac{1}{\theta}=\frac{x^{\theta}-1}{\theta}
\xrightarrow[~\theta\rightarrow 0~]{\text{}} \log x.$
The optimization problem thus reads as 
\begin{align} 
	\mathcal V_0^H(x_0):=\sup\limits_{\pi\in\mathcal{A}_0^{H}}
	\rewardorigin_0^H(x_0;\pi) \quad \text{where}\quad
	\rewardorigin_0^H(x_0;\pi) = \E\left[\utility_{\theta}(X_T^{\pi})\mid \mathcal{F}^H_0\right],~\pi\in\mathcal A^{H}_0,
	\label{opti_org}
\end{align}
where we call $\rewardorigin_0^H(x_0;\pi)$ \textit{reward function}
or \textit{performance criterion} of the strategy $\pi$ and $
\mathcal V_0^H(x_0)$ \textit{value function} to given initial
capital $x_0$.
This is  for $H\neq F$ a maximization problem under partial
information since we have required that the strategy $\pi$ is
adapted to the investor filtration $\mathbb F^{H}$.  However, the drift
coefficient of the wealth equation \eqref{wealth_phys} is not
$\mathbb F^{H}$-adapted, it depends on the non-observable drift
$\mu$. Note that for $x_0 > 0$ the solution of the SDE
\eqref{wealth_phys} is strictly positive. This guarantees that
$X_T^{\pi}$ is in the domain of logarithmic and power utility.

\subsection{Well Posedness of the Optimization Problem}
\label{Subsec:WellPosedness}
The analysis of   utility maximization problem
\eqref{opti_org} requires conditions under which the problem
is well-posed.   Problem \eqref{opti_org}  is said to be  \textit{well-posed} for the
$H$-investor, if there exists a constant
$\rewardconstant^{H}<\infty$  such that	$\valueorigin_0^{H}(x_0)\leq \rewardconstant^{H} $. Then the maximum expected
utility  of terminal wealth cannot explode in finite time as it is
the case for  so-called nirvana strategies described e.g. in Kim and
Omberg \cite{Kim and Omberg (1996)} and Angoshtari
\cite{Angoshtari2013}. Nirvana strategies generate in finite time  a
terminal wealth with a distribution leading to infinite expected
utility although the realizations of terminal wealth  may be finite.	

In general, well posedness will depend not only on the initial capital $x_0$ but on the complete set of model parameters which are $T,\theta,d,\volR,\sigma_{\mu},\revspeed,\revlevel,x_0,\driftinitial,\covinitial,
\filterinitial,\condcovinitial,\Gamma,n,\volexp$.

For power  utility  with parameter $\theta<0$ we have $\utility_\theta(x)<0$ for all $x>0$. Hence, in that case we can simply  choose
$\rewardconstant^{H}=0$ and the optimization problem is well-posed
for all model parameters with negative $\theta$. The logarithmic utility function ($\theta=0$) is no longer bounded
from above but we show below in Subsec.~\ref{LogUtility} that the value
function $\valueorigin_0^H(x_0)$ is bounded from above by some
positive constant $\rewardconstant^{H}$ for any selection of the
remaining model parameters.  More delicate  is the case of power
utility with positive parameter $\theta\in (0,1)$ which is also not
bounded from above. 	
 Here, well posedness is only guaranteed under certain restrictions on the choice of model parameters  such as the investment horizon and parameters controlling the variance of the asset price and drift processes.  For a market with a fully observed drift rate modeled by an Ornstein-Uhlenbeck process	this phenomenon was already described in Kim and
	Omberg \cite{Kim and Omberg (1996)}.    Further, it was also observed in Korn and Kraft \cite[Sec. 3]{Korn and Kraft (2004)} who coined it ``I-unstability'', in  Angoshtari \cite{Angoshtari2013,Angoshtari2016},  and Lee and Papanicolaou  \cite{Lee Papanicolaou (2016)}  who studied  power utility maximization problems and their well posedness for financial market models with  cointegrated  asset price processes, and in Battauz et al.~\cite{Battauz et al (2017)} for markets with defaultbale assets. 
	
	For detailed investigation also for models with not directly observable drift and expert opinions we refer to our  paper \cite{Gabih et al (2022) Nirvana} where we find  sufficient conditions on the model parameters ensuring well posedness.  They are given below in \eqref{wellposed_full} and \eqref{wellposed_partial}. Some results for markets with a single risky asset ($d=1$) are also contained in Colaneri et al.~\cite{Colaneri et al (2021)}.
	
One of the findings  is that depending on the chosen parameters well posedness can be guaranteed only if the trading horizon $T$ is smaller than some certain ``explosion time''.
In the following we always assume that   \eqref{opti_org} constitutes a well-posed optimization problem.

\subsection{Monetary Value of Information}
\label{MonetaryValue}  
 In this subsection we want to express the  value of information	available to the $H$-investors  in monetary terms. It is expected that the fully informed $F$-investor which can directly observe the drift has an advantage over the partially informed investors. In fact, an easy calculation as in Lee and Papanicolaou \cite[Subsec.~3.1]{Lee Papanicolaou (2016)} and \cite[Subsec.~3.3]{Gabih et al (2022) Nirvana} shows that for $H=\HR,\HC,\HD$ it holds  $\mathcal V_0^{H}(x_0)\le \E [ \mathcal V_0^{\HF}(x_0)\mid \mathcal F_0^{H}]$. 
	This inequality expresses the above advantage in mathematically.  
	The difference between the right and left hand side of the inequality is termed in Lee and Papanicolaou \cite[Subsec.~3.1]{Lee Papanicolaou (2016)} loss of utility and constitutes a first measure for the the value of information. However, utility functions and the derived value function to the
	utility maximization problems \eqref{opti_org}  do not carry a meaningful unit and therefore it is difficult to compare results for different utility functions. In order to derive  quantities with a clear economic interpretation which allow to  express the  value of information  in monetary terms  we follow a utility indifference approach	as in Brendle \cite{Brendle2006}, Lee and Papanicolaou \cite[Subsec.~3.1]{Lee Papanicolaou (2016)}, \cite[Sec.~6]{Gabih et al (2014)}.

First, we compare the fully informed
$\HF$-investor with the other partially informed $H$-investors,
$H=\HR,\HC,\HD$.  Recall that the fully informed
$F$-investor can observe the drift. The $R$-investor can only
observe stock returns while the $\HC$- and $\HD$- investors have access
to additional information and combine observations of stock return
with expert opinions. Now we consider for $H=\HR,\HC, \HD$
the initial capital $x_0^{H/\HF}$ which the $\HF$-investor needs to
obtain the same maximized expected utility at time $T$ as the
partially informed $H$-investor who started at time $0$ with
wealth $x_0^H>0$  which according to \eqref{opti_org} is given by
$\mathcal V_0^{H}(x_0)$.  Following this utility indifference
approach  $x_0^{H/ \HF}$ is obtained as solution of the following
equation   
\begin{align}
	\mathcal V_0^{H}(x_0^H)=\E \Big[ \mathcal
	V_0^{\HF}\big(\text{$x$}_0^{H/ \HF}\big)\mid \mathcal F_0^{H}\Big].
	\label{monetaer_Wert}
\end{align}
The difference $x_0^H-x_0^{H/ \HF}>0$ can be interpreted as \textit{loss of
	information} for the (non fully informed) $H$-investor measured in
monetary units, while the ratio
\begin{align}
	\varepsilon^{H}:=\frac{x_0^{H/ \HF}}{x_0^H} \in (0,1]
	\label{effekt_formel}
\end{align}
is a measure for the \textit{efficiency} of the $H$-investor
relative to the $\HF-$investor.

The above utility indifference approach can also be used to
quantify the monetary value of the additional information delivered by the experts.
We now compare  the maximum expected utility of an  $\HR$-investor who only observes returns
with that utility of the $H$-investor for $H=\HC,\HD$ who can combine
return observations with information from the experts.  Given that
the $\HR$-investor is equipped with initial capital
$x_0^R>0$ we determine the initial capital $x_0^{\HR/H}$
for the $H$-investor which leads to the same maximal expected
utility, i.e $x_0^{\HR/H}$ is the solution of the equation
\begin{align*}
	\mathcal V_0^{\HR}(x_0^R)=\E \Big[ \mathcal V_0^{H}(x_0^{\HR/H})\mid
	\mathcal F_0^{\HR}\Big].
\end{align*}
Since we assume  that at time $t=0$ all partially informed
investors have access to the same information about the drift, it
holds $\mathcal F_0^{\HR}=\mathcal F_0^{H}=\mathcal F_0^{I}$ (see
Subsec.~\ref{Investor_Filtration}) the above equation reads as
\begin{align}\label{initial_cap_H}
	\mathcal V_0^{\HR}(x_0^R)= \mathcal V_0^{H}(x_0^{\HR/H}).
\end{align}
From the initial capital $x_0^R$  the $R$-investor can
put aside the amount $P_{Exp}^{H}:=x_0^R-x_0^{\HR/H}$ to buy the
information from the expert. The remaining capital  $x_0^{\HR/H}$
can be invested in  an $H-$optimal portfolio and leads to the same
expected utility of terminal wealth as the $\HR$-optimal portfolio
with initial capital $x_0^R$. Hence,  $P_{Exp}^{H}$ describes the
\textit{monetary value of the expert opinions} for the $\HR$-investor.

\section{Partial Information and Filtering}
\label{Filtering} The trading decisions of investors are based on
their knowledge about the drift process $\mu$. While the
$F$-investor observes the drift directly, the $H$-investor for
$H=R,\HC,\HD$ has to estimate the drift. This leads us to a filtering
problem with hidden signal process $\mu$ and observations given by
the returns $R$ and the sequence of  expert opinions
$(t_k,Z_k)$. The \textit{filter} for the drift  $\mu_t$ is its
projection on the $\mathcal{F}_t^H$-measurable random variables
described by the conditional distribution of the drift given
$\mathcal{F}_t^H$. The mean-square optimal estimator for the drift
at time $t$, given the available information is  the
\textit{conditional mean}
$$\Mpro_t^{H}:=\E[\mu_t\mid\mathcal{F}_t^H].$$
The accuracy of that estimator  can be described by the
\textit{conditional covariance matrix}
\begin{align}\nonumber
	\Qpro_t^{H}:=\E[(\mu_t-\Mpro^{H}_t)(\mu_t-\Mpro^{H}_t)^{\top}\mid\mathcal{F}^{H}_t].
\end{align}
Since in our filtering problem  the signal   $\mu$, the observations
and the initial value of the filter  are jointly Gaussian also the
the  conditional distribution of the drift is Gaussian and completely characterized by the
conditional mean $\Mpro_t^{H}$ and the conditional covariance
$\Qpro_t^{H}$.
\subsection{Dynamics of the Filter}
\label{filter_dynamics}  We now  give the dynamics of the
filter  processes $\Mpro^{H}$ and    $\Qpro^{H}$  for $H=R,\HD,\HC$.

\paragraph{$R$-Investor}
The $R$-investor  only observes returns and has no access to
additional expert opinions,  the information is given by
$\mathbb{F}^R$. Then, we are in the classical case of the   Kalman-Bucy
filter, see e.g.~Liptser and Shiryaev \cite{Liptser-Shiryaev},
Theorem $10.3$, leading to the following  dynamics of $\Mpro^R$ and
$\Qpro^R$.

\begin{lemma}
	\label{Kalmann_Filter_R_lemma} 
	The conditional distribution of the drift given the $\HR$-investor's observations is Gaussian.
	The conditional mean $\Mpro^{R}$ follows the dynamics
	\begin{align}
		\label{Filter_R}
		d\Mpro_t^{R}&=\revspeed(\revlevel-\Mpro_t^{R})\;dt+\Qpro_t^{R}\,
		\Sigma_R^{-1/2}\;d\widetilde{W}_t^\HR.
	\end{align}
	The  innovation process
	~$\widetilde{W}^R=(\widetilde{W}_t^R)_{t\in[0,T]}$   given by
	$d\widetilde{W}_t^\HR=\Sigma_R^{-1/2}\left(dR_t-\Mpro_t^\HR
	dt\right)$, $\widetilde{W}_0^\HR=0$,  is a $\nAktien$-dimensional standard  Brownian motion adapted to $\mathbb{F}^R$.\\
	The dynamics of the conditional covariance  $\Qpro^{R}$ is given by
	the   Riccati differential equation
	\begin{align}
		\label{Riccati_R} d\Qpro_t^{R}&=(\Sigma_{\mu}-\revspeed
		\Qpro_t^{R}-\Qpro_t^{R} \revspeed^{\top}-\Qpro_t^{R} \Sigma_{R}^{-1}
		\Qpro_t^{R})\; dt.
	\end{align}
	The initial values are $\Mpro_0^{R}=\filterinitial $ and
	$\Qpro_0^{R}=\condcovinitial$.
\end{lemma}

Note that the conditional covariance matrix $\Qpro_t^{R}$ satisfies
an ordinary differential equation and is hence deterministic,
whereas the conditional mean $\Mpro_t^{R}$ is a stochastic process
defined by an SDE driven by the innovation process $\widetilde{W}^R$.

\paragraph{$\HD$-Investor}
The $\HD$-investor  observes a $2d$-dimensional diffusion process with components $R$ and $\HD$. That observation process is driven by a $(\nWienerRendite+\nWienerExperten)$-dimensional Brownian motion with components $W^\HR$ and $W^\HD$.
Again, we can apply classical  Kalman-Bucy filter theory as in Liptser and Shiryaev \cite{Liptser-Shiryaev} to deduce the dynamics of $\Mpro^\HD$ and $\Qpro^{\HD}$. We also refer to Lemma 2.2 in Sass et al.~\cite{Sass et al (2022)}.

\begin{lemma}
	\label{Kalmann_Filter_D_lemma}  The conditional distribution of the drift given the $\HD$-investor's observations is Gaussian.
	The conditional mean $\Mpro^{\HD}$ follows the dynamics
	\begin{align}
		\label{Filter_D}
		d\Mpro_t^{\HD}&=\revspeed(\revlevel-\Mpro_t^{\HD})\;dt+\Qpro_t^{\HD}\,
		(\Sigma_R^{-1/2}, \Sigma_\HD^{-1/2}) d\widetilde{W}_t^\HD.
	\end{align}
	The  innovation process
	$\widetilde{W}^\HD=(\widetilde{W}_t^\HD)_{t\in[0,T]}$   given by
	$$d\widetilde{W}_t^\HD=
	\begin{pmatrix}
		\Sigma_R^{-1/2}\left(dR_t-\Mpro_t^\HD   dt\right) \\
		\Sigma_\HD^{-1/2}\left(d\contexp_t-\Mpro_t^\HD   dt\right)
	\end{pmatrix},~~~ \widetilde{W}_0^\HD=0,$$  is a $2d$-dimensional  standard  Brownian motion adapted to $\mathbb{F}^\HD$.\\
	The dynamics of the conditional covariance  $\Qpro^{\HD}$ is given by
	the  Riccati differential equation
	\begin{align}
		\label{Riccati_D} d\Qpro_t^{\HD}&=(\Sigma_{\mu}-\revspeed
		\Qpro_t^{\HD}-\Qpro_t^{\HD} \revspeed^{\top}-\Qpro_t^{\HD}  (\Sigma_{R}^{-1} + \Sigma_{\HD}^{-1})
		\Qpro_t^{\HD})\; dt.
	\end{align}
	The initial values are $\Mpro_0^{\HD}=\filterinitial $ and  $\Qpro_0^{\HD}=\condcovinitial$.
\end{lemma}

\paragraph{$\HC$-Investor}
The next lemma provides the filter for the $\HC$-investor who
combines continuous-time observations of  stock
returns and expert opinions received at discrete points in time. For a detailed proof we refer to  Lemma 2.3 in \cite{Sass et al
	(2017)} and  Lemma 2.3 in \cite{Sass et al (2022)}.
\begin{lemma}
	\label{filter_C}
	The conditional distribution of the drift given the $\HC$-investor's observations is Gaussian. The dynamics of the conditional mean and conditional covariance matrix are given as follows:
	\ \\[-2ex]
	\begin{enumerate}
		\item[(i)]
		Between two information dates $t_k$ and $t_{k+1}$, $ k=0,\ldots,n-1$,   the
		conditional mean  $\Mpro_t^\HC$  satisfies SDE \eqref{Filter_R}, i.e.,
		\begin{align*}
			d\Mpro_t^{\HC}&=\revspeed(\revlevel-\Mpro_t^{\HC})\;dt+\Qpro_t^{\HC}\,
			\Sigma_R^{-1/2}\;d\widetilde{W}_t^\HC \quad\text{for}~~ t\in
			[t_k,t_{k+1}).
		\end{align*}
		The  innovation process
		$\widetilde{W}^\HC=(\widetilde{W}_t^\HC)_{t\in[0,T]}$   given by
		$$d\widetilde{W}_t^\HC=\Sigma_R^{-1/2}\left(dR_t-\Mpro_t^\HC
		dt\right),~ \widetilde{W}_0^\HC=0,$$  is a $\nAktien$-dimensional standard  Brownian motion adapted to $\mathbb{F}^\HC$.\\
		The conditional covariance $\Qpro^{\HC}$ satisfies the ordinary Riccati
		differential equation \eqref{Riccati_R}, i.e.,
		\begin{align*}
			d\Qpro_t^{\HC}&=(\Sigma_{\mu}-\revspeed
			\Qpro_t^{\HC}-\Qpro_t^{\HC} \revspeed^{\top}-\Qpro_t^{\HC} \Sigma_{R}^{-1}
			\Qpro_t^{\HC})\; dt.
		\end{align*}
		The initial values are $\Mpro_{t_k}^\HC $ and $\Qpro^\HC_{t_k}$,
		respectively, with $\Mpro_{0}^{\HC}=\filterinitial$ and~
		$\Qpro_{0}^{\HC}=\condcovinitial$.
		\item[(ii)]
		At the information dates $t_k$,  $ k= 1,\ldots,n-1$, the conditional mean  and
		covariance $\Mpro_{t_k}^\HC$ and $\Qpro_{t_k}^\HC$ are obtained from the
		corresponding values at time $t_{k^-}$ (before the arrival of the
		view) using the update formulas
		\begin{align}
			\label{updateformel_Filter_CN}
			\Mpro_{t_k}^{\HC}&=\rho_k\Mpro_{t_k-}^{\HC}+(I_d-\rho_k)Z_k,\\
			\label{updateformel_Variance_CN} \Qpro_{t_k}^{\HC}&=\rho_k
			\Qpro_{t_k-}^{\HC},
		\end{align}
		with the update factor
		$\rho_k=\varianceexp(\Qpro_{t_k-}^{\HC}+\varianceexp)^{-1}$.  At initial time  $t=0$  the  above update formulas give $\Mpro_{0+}^{\HC}$ and $\Qpro_{0+}^{\HC}$ based on the initial values  $\Mpro_{0}^{\HC}=\filterinitial$ and  $\Qpro_{0}^{\HC}=\condcovinitial$.
	\end{enumerate}
\end{lemma}
Note that the dynamics of $\Mpro^\HC$ and $\Qpro^\HC$ between
information dates are the same as for the $R$-investor, see Lemma
\ref{Kalmann_Filter_R_lemma}. The values at an information date
$t_k$ are obtained from a Bayesian update. Further, we recall that  for the $R$- and $\HD $-investor the conditional mean $\Mpro^H$ is
a diffusion process and the conditional covariance $\Qpro^H$ is a continuous and deterministic function.  Contrary to that,  for the  $\HC$-investor  the conditional mean $\Mpro^\HC$ is a diffusion process between the information dates but shows  jumps of random jump size at those dates. The conditional covariance $\Qpro^\HC$ is piecewise continuous with deterministic jumps at the arrival dates  $t_k$ of the expert  opinions.

\subsection{Properties of the Filter}
\label{propertiesfiltersection} 
In this and the next subsection we  collect some properties of the filter processes. We start with a proposition  stating in
mathematical terms  the intuitive property that additional
information from the expert opinions improves drift estimates. Since
the accuracy of the filter is measured by the conditional variance
it is expected that this quantity for the $\HC$- and $\HD$-investor
who combine observations of returns and  expert opinions  is
``smaller'' than for the $R$-investor who observes returns only.
Mathematically, this can be expressed by the   partial ordering of
symmetric matrices. For symmetric matrices $A,B\in\R^{d\times d}$ we
write $A \preceq B$ if $B-A$ is positive semi-definite. Note that $A
\preceq B$   implies that $\norm{A}\le \norm{B}$.

\begin{proposition}[Sass et al. \cite{Sass et al (2021)}, Lemma 2.4]~    	
	\label{properties_filter} 
	
	\noindent        For $H=\HC,\HD$ it holds $\Qpro^{H}_t \preceq \Qpro^R_t$
	and there exists a constant $C_{\Qpro}>0 $ such that
	$\norm{\Qpro^{H}_t} \le \norm{\Qpro^R_t}\le C_{\Qpro} $ for all
	$t\in[0,T]$.
\end{proposition}

At the arrival dates $t_k$ of the expert opinions the expert's view $Z_k$ lead to an update of the conditional mean $\Mpro^Z$ given by \eqref{updateformel_Filter_CN}. That update can be considered as a weighted mean of the filter estimate $\Mpro^Z_{t_k-}$ before the arrival and the expert opinion $Z_k$. The following proposition shows that the update improves  the accuracy both of the estimate  $\Mpro^Z_{t_k-}$ before the arrival as well as of the expert's estimate $Z_k$.
\begin{proposition}[Sass et al. \cite{Sass et al (2017)}, Proposition 2.2]~
	\label{properties_filter_2} 
	
	\noindent
	For $k=0,\ldots,n-1$ it holds ~$\Qpro^{Z}_{t_k} \preceq \varianceexp$ ~ and ~ $\Qpro^{Z}_{t_k} \preceq \Qpro^{Z}_{t_k-}$.
\end{proposition}
The following lemma provides 
the conditional distribution of the expert opinions $Z_k$ given the available information of the $Z$-investor before the arrival of the expert's view.  
\begin{lemma}[Kondkaji \cite{Kondkaji (2019)}, Lemma 3.1.6]~
	\label{bedingte_Verteilung_Z_lem} 
	
	\noindent
	The conditional distribution of
	the expert opinions $Z_k$ given $\mathcal F_{t_k-}^{\HC}$ is  the
	multivariate normal distribution $\mathcal
	N\left(\Mpro_{t_k-}^{\HC},\;\varianceexp+\Qpro_{t_k-}^{\HC}\;\right)$,
	$k=0,\ldots,n-1$. 
\end{lemma}
According to this lemma we can choose a sequence of i.i.d.~$\mathcal
F_{t_k-}^{\HC}$-measurable random vectors
$U_k\sim\mathcal{N}(0,I_d)$, $k= 0,\ldots, n-1$ such that under
$\mathbb{F}^{\HC}$ it holds $
Z_k-\Mpro_{t_k-}^{\HC}=\left(\varianceexp+\Qpro_{t_k-}^{\HC}\right)^{\frac{1}{2}}
U_k$. From the update formula \eqref{updateformel_Filter_CN} we
deduce that the increments of the filter process $\Mpro^{\HC}$ at
the information dates $t_k$ can be expressed as
\begin{align}
	\Mpro_{t_k}^{\HC}-\Mpro_{t_k-}^{\HC}=
	\Qpro_{t_k-}^{\HC}\left(\varianceexp+\Qpro_{t_k-}^{\HC}\right)^{-\frac{1}{2}}
	U_k. \label{updateformel_Filter_CN_bedingt}
\end{align}
Further the update formula \eqref{updateformel_Variance_CN} implies
that  the (deterministic) increments of the filter process
$\Qpro^{\HC}$ at the information dates $t_k$ can be expressed as
\begin{align}
	\Delta \Qpro_{t_k}^{\HC} =~  \Qpro_{t_k}^{\HC}-\Qpro_{t_k-}^{\HC}=
	-\Qpro_{t_k-}^{\HC}\left(\varianceexp+\Qpro_{t_k-}^{\HC}\right)^{-1}
	\Qpro_{t_k-}^{\HC}. \label{updateformel_Variance_CN_bedingt}
\end{align}
\begin{remark}\label{remark_asymptotic_t}
	We mention  some asymptotic properties of the conditional variances $\Qpro^H_t$ for $t\to \infty$. Sass et al. \cite[Theorem 4.1.]{Sass et al (2017)} show that the conditional variances $\Qpro^R$ and $\Qpro^J$ for diffusion type observations stabilize for increasing $t$ and tend to a finite limit. For the $Z$-investor receiving expert opinions at equidistant time points they show in Prop.~4.1 that the conditional variances $\Qpro^Z_{t_k-}$ and $\Qpro^Z_{t_k}$ before and after the arrival, respectively, stabilize and tend to  (different) finite limits. We also refer to our numerical results presented in Subsec.~\ref{Num_Filter}.
\end{remark}

\subsection{Asymptotic Filter Behavior for High-Frequency Expert Opinions}
\label{filter_high_frequency}
In this subsection we provide results for the asymptotic behavior of the filters for a $\HC$-investor when the number of
expert opinions goes to infinity. This will be helpful for deriving approximations not only of the filters but also of solutions to the utility maximization  problem \eqref{opti_org} in case of high-frequency expert opinions. We will denote the arrival times of the expert's views by $t_k=t_k^{(n)}$ to emphasize the dependence on $n$. Then we have for all $n$ that  $0=t_0^{(n)}<t_1^{(n)}<\ldots<t_{n-1}^{(n)}<T$. Again we set $t_n^{(n)} =T$. We also use an additional superscript $n$ and write 	 $(\Mpro^{\HC,\nExperten})_{t\in[0,T]}$ and $(\Qpro^{\HC,\nExperten})_{t\in[0,T]}$ for the
conditional mean and the conditional covariance matrix of the
filter, respectively, in order to emphasize dependence of the filter processes on the number of expert opinions.

We distinguish two different asymptotic regimes. First the expert's variance $\varianceexp$ stays constant, second that variance grows linearly with the number $n$ of expert opinions.
\subsubsection{High-frequency expert opinions with fixed variance}
\label{asymptotic_filter_F}
For an increasing number of expert opinions with fixed variance $\varianceexp$ the investor receives more and more  noisy signals about the
current state of the drift $\drift$ of the same precision. Then it can be expected that in the limit the filter process $\Mpro^Z_t$ constitutes a  perfectly accurate estimate which is equal to  the actual drift $\drift_t$, i.e., the investor has full information about the  drift. This intuitive statistical consistency result has been rigorously proven in  \cite{Sass et al (2017)} under the following 
\begin{assumption}~\label{assumption_determ_Exp_F}~%
	\begin{enumerate}
		\item The expert's covariance matrix $\varianceexp$ is (strictly) positive definite and does not depend on~$n$.
		\item 
		For the mesh size $\delta_n = \max\limits_{k=0,\ldots,n-1}  \vert t_{k+1}^{n)}-t_k^{(n)}\vert  $	~it holds~ $\lim\limits_{n\to \infty} \delta_n=0$.
	\end{enumerate}	
\end{assumption}
\begin{theorem}[Sass et al. \cite{Sass et al (2017)}, Theorem 3.4.]	~
	\label{theorem_Konv_Fn} 
	
	\noindent    Under Assumption \ref{assumption_determ_Exp_F} it holds for 
	all  $t\in(0,T]$ for \\[0.5ex]
	$\begin{array}{rll}
		1. &  \text{the conditional  covariance matrix} &
		\lim\limits_{n\rightarrow\infty}\| \Qpro_t^{\HC,\nExperten}\|=0,
		\\[0.5ex]
		2. &  \text{the conditional mean} & 
		\lim\limits_{n\rightarrow\infty} \E \Big[ \big\|\Mpro_t^{\HC,\nExperten}-\drift_t \big\|^2 \Big]=0.
	\end{array}$
\end{theorem}
\subsubsection{High-frequency expert opinions with linearly growing variance}
\label{asymptotic_filter_D} 
Now we consider another asymptotic regime arising in models in which a higher arrival frequency of expert opinions is only available at the cost of accuracy.  We assume that the expert views arrive at equidistant time points and the variance $\varianceexp$ of the views  $Z_k$ grows linearly with $n$. This reflects that contrary to the above setting with constant variance  now investors are not able to gain  an arbitrary amount of information over a fixed time interval.

We recall the dynamics of the continuous expert opinions  $\contexp=(\contexp_t)_{t\in[0,T]}$ given in \eqref{continuous-expert} by the SDE  ~
$d\contexp_t=\mu_t\; dt+ \volexp\;dW_t^{\HD},~ \contexp_0=0$, and make the following 
\begin{assumption}~%
	\label{assumption_determ_Exp_J}%
	\begin{enumerate}
		\item The expert arrival dates are equidistant , i.e., $t_k=t_k^{(\nExperten)}=k\Delta_{\nExperten}$ for
		$k=0,\ldots,\nExperten-1$ with $\Delta_n=\frac{T}{n}$.
		\item The experts covariance matrix is given by
		$\varianceexp=\varianceexp^{(\nExperten)}=\frac{1}{\Delta_{\nExperten}}\volexp\volexp^{\top}.$
		\item
		The normally distributed random vectors $(\varepsilon_k^{\nExperten})$ in \eqref{Expertenmeinungen_fest} are linked
		with the Brownian motion $W^{\HD}$ from
		\eqref{continuous-expert} via
		$\varepsilon_k^{\nExperten}=\frac{1}{\sqrt{\Delta_{\nExperten}}}\int_{[t_k^{(\nExperten)},	t_{k+1}^{(\nExperten)}]}dW_s^{\HD}$.
	\end{enumerate}   
\end{assumption}
In view of the representation of expert opinions in \eqref{Expertenmeinungen_fest} the third assumption implies that 
\begin{align}
	\label{eq:diff_appr_Z}
	Z_k^{(\nExperten)}=\drift_{t_k^{(\nExperten)}}+\frac{1}{\Delta_{\nExperten}}\volexp\int_{[t_k^{(\nExperten)},	t_{k+1}^{(\nExperten)}]}dW_s^{\HD},\quad k=0,\ldots,\nExperten-1.    
\end{align}

The following Theorem shows that in the present setting the  information obtained from observing the discrete-time expert
opinions is asymptotically the same as that from observing the  diffusion process $\HD$  representing the continuous-time expert and defined in \eqref{continuous-expert}. 
\begin{theorem}[Sass et al. \cite{Sass et al (2022)}, Theorems 3.2 and  3.3]~
	\label{theorem_conv_diffusion_appr}
	
	\noindent       Let $p\in[1,+\infty)$. Under  Assumption \ref{assumption_determ_Exp_J} it holds:
	\begin{itemize}
		\item[1)] There exists a constant $K_{\Qpro}>0$ such that ~
		$\big\|\Qpro^{Z,n}_t-\Qpro_t^{\HD} \big\| \leq K_{\Qpro}\Delta_{\nExperten}~~ \text{for all }~ t\in[0,T].$
		\\[0.5ex]
		In particular, it holds ~~
		$~\lim\limits_{n\rightarrow\infty}\sup\limits_{t\in[0,T]}
		\big\|\Qpro_t^{Z,n}-\Qpro_t^{\HD} \big\|=0.$\\[0.0ex]
		\item[2)] There exists a constant $K_{m,p}>0$ such that
		$\E \left[ \big\|\Mpro_t^{Z,n}-\Mpro_t^{\HD} \big\|^p\right]
		\leq K_{m,p}\Delta_{\nExperten} ~~ \text{for all }~ t\in[0,T].$
		\\[0.5ex]
		In particular, it holds ~~
		$~\lim\limits_{n\rightarrow\infty}\sup\limits_{t\in[0,T]}
		\E \left[ \big\|\Mpro_t^{Z,n}-\Mpro_t^{\HD} \big\|^p
		\right]=0.$
	\end{itemize}
\end{theorem}

\section{Utility Maximization}
\label{Utility_Max}
This section is devoted to the solution of the utility 	maximization problem  \eqref{opti_org}. We briefly review in Subsec.~\ref{LogUtility} the solution for logarithmic utility. For the more demanding case  of power utility we reformulate problem \eqref{opti_org} in Subsec.~\ref{PowerUtility} as an equivalent stochastic optimal control problem which can be solved by dynamic programming techniques. We present the solutions for the fully informed investor ($H=F$) in Subsec.~\ref{FullInfo}.  Results for partially informed investors with diffusion type observations ($H=R,J$) and  for the $Z$-investor observing discrete-time expert opinions will follow  in Secs.~\ref{UtilityMaxDiffusion} and \ref{UtilityMaxDiscreteExperts}.

\subsection{Logarithmic Utility}
\label{LogUtility}
For an investor who wants to maximize expected logarithmic utility of terminal wealth
optimization problem \eqref{opti_org} reads as
\begin{eqnarray}
	\label{opti_org_log} \mathcal
	V_0^H(x_0):=\sup\limits_{\pi\in\mathcal
		A_0^{H}}\E\left[ \log(X_T^{\pi})\mid\mathcal{F}^H_0 \right] .
\end{eqnarray}

This optimization problem has been solved in Gabih et al.~\cite{Gabih et al (2014)} and generalized in Sass et al.~\cite{Sass
	et al (2017)} and Kondakji \cite{Kondkaji (2019)} in the context of the different information regimes
addressed in this paper. In the sequel we state the obtained
results. 
\begin{proposition}
	\label{LogUtilityValue}
	The optimal strategy $(\pi^{H}_t)_{t\in[0,T]})$ for the optimization problem
	\eqref{opti_org_log} is  given in feedback form by $\pi^H_t=\Pi^H(t,\Mpro_t^H)$ where the optimal decision rule is given by
	$$\Pi^{H}(t,m)=\Sigma_{R}^{-1}m \quad\text{for } m\in\R^d,$$
	and the optimal value is
	\begin{align}\nonumber
		\mathcal V_0^H(x_0)&=\log(x_0)+\frac{1}{2}\int\nolimits_0^T \trace\big(\Sigma_{R}^{-1}\E[\Mpro^{H}_t(\Mpro^{H}_t)^{\top}]\big)dt\\
		\label{Log_Utility_Value}
		&=\log(x_0)+\frac{1}{2}\int\nolimits_0^T \trace\big(\Sigma_{R}^{-1}\big(Var[\drift_t]+\E [\drift_t]\E [\drift_t^{\top}]-\Qpro_t^{H}\big)\big)dt.
	\end{align}
\end{proposition}
We assumed in our model for the drift process $\drift$ in \eqref{drift} that the matrix $\revspeed$ is positive definite. Using the closed-form solution of the SDE \eqref{drift} given in \eqref{mu_explicit} it can be deduced that the mean $\E [\drift_t]$ and covariance matrix $Var[\drift_t]$  are bounded. Further, it is known  from Prop.~\ref{properties_filter} that also the conditional covariance matrix $\Qpro_t^{H}$ is bounded. Thus from representation \eqref{Log_Utility_Value}  it can be derived that  the value function $\mathcal V_0^H(x_0) $ is bounded.  As already mentioned in Subsec.~\ref{Subsec:WellPosedness} there exists some constant $\rewardconstant^{H}>0$ such that $\mathcal V_0^H(x_0)\le \rewardconstant^{H}$.

Representation \eqref{Log_Utility_Value} also shows that the value function depends on the information regime $H$ only via  an integral functional of the
conditional covariance  $(\Qpro^{H}_t)_{t\in[0,T]}$. This allows to carry over 
the convergence results for the conditional covariance matrices  $\Qpro^{Z,n}$ for $n\to\infty$ given in Theorems \ref{theorem_conv_diffusion_appr} and \ref{theorem_Konv_Fn} to the  value functions. 
We refer to Sass et al. \cite[Corollary 5.2.]{Sass et al (2017)} for the convergence  $\mathcal V_0^{Z,n}(x_0)\to \mathcal V_0^{F}(x_0)$ 
for the case of a fixed expert's variance $\varianceexp$ and to Sass et al. \cite[Corollary 5.2.]{Sass et al (2022)} for the convergence  $\mathcal V_0^{Z,n}(x_0)\to \mathcal V_0^{J}(x_0)$ for linearly growing variance.

\subsection{Power Utility}
\label{PowerUtility} In this section we focus on  the maximization of expected power utility as  given in \eqref{util_def}.   That problem can be treated as a stochastic optimal control problem and solved using dynamic programming methods. 	We will apply a change of measure technique which was already used
among others by Nagai and  Peng \cite{Nagai and Peng (2002)} and Davis and Lleo \cite{Davis and Lleo (2013_1)}. This allows to study  simplified control problems in which the state variables are reduced  to the (slightly modified) filter processes of conditional mean whereas the wealth process  can be removed from the state.

\paragraph{Performance criterion} 
Recall  equation
\eqref{wealth_phys} for the wealth process $\wealth^{\pi}$ saying that ${dX_t^{\pi}}/{X_t^{\pi}}=
\pi_t^{\top}dR_t$. Rewriting SDE \eqref{ReturnPro} for the return process $R$  in terms  of the
innovations processes $\widetilde{W}^{H}$ given in
Lemmas \ref{Kalmann_Filter_R_lemma}, \ref{Kalmann_Filter_D_lemma}
and \ref{filter_C} we obtain for $H=\HR, \HD, \HC$
the $\mathbb F^{H}$-semimartingale decomposition of $\wealth^{\pi}$
(see also Lakner \cite{Lakner (1998)}, Sass Haussmann
\cite{Sass and Haussmann (2004)})  
\begin{align}
	\label{wealth-innovation}
	\frac{d\wealth_t^{\pi}}{\wealth_t^{\pi}}=\pi^{\top}_t\Mpro_t^{H}+\pi^\top_t\sigma_{\wealth}^{H}\;d\widetilde{W}_t^{H},\quad\wealth_0^{\pi}=x_0,
\end{align}
where $\sigma_{\wealth}^{R}=\sigma_{\wealth}^{\HC}= \Sigma_R^{1/2}$   and
$\sigma_{\wealth}^{\HD}=(\Sigma_R^{1/2},~0_{\nAktien\times\nAktien})$.
From the above wealth equation we obtain that for a given admissible
strategy  $\pi$ the power  utility of terminal wealth $\wealth_T^{\pi}$ is
given by
\begin{align}
	\label{Power_Nutzen_zerlegung}
	\utility_{\theta}(\wealth_T^{\pi})&=\frac{1}{\theta}(\wealth_T^{\pi})^{\theta}
	=\frac{x_0^{\theta}}{\theta} \Radon_T^{H}\exp\Big\{\int\nolimits_0^T
	{b}(\Mpro_s^{H},\pi_s)ds \Big\}
\end{align}
where  for $m,p\in\R^d$
\begin{align}
	\label{b_underline}
	{b}(\zustand,\pointp)&=\theta\Big(\pointp^{\top}\zustand-\frac{1-\theta}{2} \|\pointp^{\top}\sigma_{\wealth}\|^2 \Big) 
	\quad\text{and}\\
	\label{Radon}
	\Radon_T^{H}&=\exp\Big\{ \theta\int\nolimits_0^T \pi_s^{\top}\sigma_{\wealth} \;d\widetilde{W}_s^{H} 
	-\frac{1}{2}\theta^2\int\nolimits_0^T \|\pi_s^{\top}\sigma_{\wealth}\|^2  \;ds \Big\}.
\end{align}
 Since we require that admissible strategies $\pi$ satisfy $\E[\Radon_T^{H}]=1$
we can define an equivalent probability measure $\overline{\P}^{H}$
by $\Radon_T^{H}=\frac{d\overline{\P}^{H}}{d\P~}$ for which Girsanov's
theorem		guarantees that the process
$\overline{W}^{H}=(\overline{W}_t^{H})_{t\in[0,T]}$ with
\begin{align}
	\label{Girsanov_W_prozess}
	\overline{W}_t^{H}=\widetilde{W}_t^{H}-\theta\int_0^t
	\sigma_{\wealth}^{\top}\pi_s \; ds,\quad t\in[0,T],
\end{align}
is a $(\mathbb F^{H},\overline{\P}^{H})-$standard Brownian motion.
This change of measure allows to rewrite the performance criterion $\rewardorigin_0^H(x_0;\pi) =
\E\left[\utility_{\theta}(X_T^{\pi})~\mid~\mathcal{F}^H_0\right]$ of the  utility maximization problem \eqref{opti_org} for $\pi\in\mathcal
A^{H}_0$ as
\begin{align}\nonumber		
	\rewardorigin_0^H(x_0;\pi)&= \frac{x_0^{\theta}}{\theta} 
	\E\Big[  \Radon_T^{H}\exp\Big\{\int\nolimits_0^T {b}(\Mpro_s^{H,\filterinitial,\condcovinitial},\pi_s)ds \Big\}\Big]\\
	\label{ExpUtility_RiskSensitive}
	&=\frac{x_0^{\theta}}{\theta} \Ebar\Big[\exp\Big\{\int\nolimits_0^T {b}(\Mpro_s^{H,\filterinitial,\condcovinitial},\pi_s)ds \Big\}\Big].
\end{align}
Above we used the   notation $\Ebar$ for the expectation under the measure $\overline{\P}^H$. Further, the notation   $\Mpro^{H,\filterinitial,\condcovinitial}$ emphasizes the dependence of the filter processes $\Mpro^H, \Qpro^H$ on the initial values $\filterinitial,\condcovinitial$ at time $t=0$  and reflects the conditioning in $\E\left[\utility_{\theta}(X_T^{\pi})~\mid~\mathcal{F}^H_0\right]$ on the initial information given by the $\sigma$-algebra $\mathcal{F}^H_0$. 
It turns out that the utility maximization problem \eqref{opti_org}
is equivalent to the maximization of 
\begin{align}\label{opti-equivalent-generi}
	\Jpi^H(\filterinitial,\condcovinitial;\pi)=\Ebar \Big[\exp\Big\{\int\nolimits_0^T
	{b}(\Mpro_s^{H,\filterinitial,\condcovinitial},\pi_s)ds \Big\}\Big]
\end{align}
over all admissible strategies for $\theta\in(0, 1)$ while for
$\theta< 0$ the above expectation has to be minimized.  Note that it holds $\rewardorigin_0^H(x_0;\pi)=\frac{x_0^{\theta}}{\theta}\, \Jpi^H(\filterinitial,\condcovinitial;\pi)$.
This allows us to remove the wealth process $X$ from the state of the control problem which we formulate next.

\paragraph{State process}  In view of the performance criterion \eqref{opti-equivalent-generi} the state process of the associated control problem is  	 the conditional mean
	$\Mpro^{H}$ for which we need to express the dynamics under the measure $\overline{\P}^H$. Recall the $\P$-dynamics of 	$\Mpro^{H}$   given in Lemma \ref{Kalmann_Filter_R_lemma} through \ref{filter_C}. Using \eqref{Girsanov_W_prozess} the dynamics under $\overline{\P}^H$ for
$H=\HR, \HC,\HD$ are obtained by expressing $\widetilde{W}^{H}$ in terms of $\overline{W}^{H}$ and leads to the SDE for $\Mpro^{H}=\Mpro^{H,\filterinitial,\condcovinitial}$
\begin{align}
	\label{Filter_M_int_gen}
	d\Mpro_t^{H}& =\alphamm(\Mpro_t^{H},\Qpro_t^{H};\pi_t)\, dt
	+\betam^{H}(\Qpro_{t}^{H})\,d\overline{W}_t^{H},\quad \Mpro_0^{H}=\filterinitial,
\end{align}
where for $m,p\in\R^d$ and $q\in\R^{d\times d}$
\begin{align}
	\label{alpha_beta_MQ_def}
	\alphamm(\zustand,\variance;\pointp)=\revspeed(\revlevel-\zustand)+ \theta q \pointp \quad \text{and}\quad
	\betam^{H}(\variance)= \left\{\begin{array}{cl}
		\variance\Sigma_{\HR}^{-1/2}, & H=\HR,\HC,\\[1ex]
		\variance(\Sigma_{\HR}^{-1/2},\Sigma_{J}^{-1/2}),& H=\HD.
	\end{array}\right.
\end{align}
Note that for $H=\HC$ the above SDE describes the dynamics only between two arrival dates
$t_{k-1}$ and $t_k, k=1,\ldots,n,$ whereas  at the arrival dates $t_k$ according to the
updating formula \eqref{updateformel_Filter_CN}   there are jumps of size $\Mpro_{t_k}^{\HC}-\Mpro_{t_k-}^{\HC} =(I_d-\rho_k)(Z_k-\Mpro_{t_k-}^{\HC})$.
Further, note that the drift coefficient $\alphamm$ in the SDE \eqref{Filter_M_int_gen} for $\Mpro^H$ now depends also on the conditional variance $\Qpro^H$ as well as on the strategey $\pi$.
Since the  conditional covariance $\Qpro^H$  is deterministic it is not affected by the change of measure.

The case of full information can formally be incorporated in
our model with the settings $\Mstate^F=\mu$ and a state equation
\eqref{Filter_M_int_gen}   with the coefficients
$\alphamm(\zustand,\variance,\pointp)=\revspeed(\revlevel-\zustand)$,
$\betam^F(q) =\voldrift$  which are independent of $q$.

\medskip
\paragraph{Markov Strategies}
To apply the dynamic programming approach to the
optimization problem \eqref{opti-equivalent-generi} the state
process $\Mstate^{H}$ needs to be a Markov process adapted to
$\mathbb F^{H}$. To allow for this situation we restrict the set of
admissible strategies to those of Markov type which are defined  in terms  
of time and the state process $\Mstate^{H}$ according to a given specified
decision rule $\Pi$, i.e, $\pi_t=\Pi(t,\Mstate^{H}_t)$ for a some
given measurable function $\Pi : [0,T]\times\R^d\rightarrow
\mathbb{R}^{\nAktien}$. Below  we will need some technical
conditions on $\Pi$ which we collect in the following
\begin{assumption}\label{admi_stra_rule}~
	\begin{itemize}
		\item[1.]\textbf{Lipschitz condition}\\
		There exists a constant $C_L>0$ such that for all $m_1,
		m_2\in\R^d$ and all $t\in[0,T]$ it holds
		\begin{align}\label{Lipschitz_rule}
			\|\Pi(t,m_1)-\Pi(t,m_2)\|\leq C_L\|m_1-m_2\|.
		\end{align}
		\item[2.]\textbf{Linear growth condition}\\
		There exists a constant $C_G>0$ such that for all $m\in\R^d$
		and all $t\in[0,T]$ it holds
		\begin{align}\label{Linear-growth-rule}
			\|\Pi(t,m)\|\leq C_G(1+\|m\|).
		\end{align}
		
			\item[3.]\textbf{Integrability condition } \\
			For  the information regimes $H=\HR,\HD,\HC$ the strategy processes $\pi$ defined by $\pi_t=\Pi(t,\Mstate^H_t)$ on $[0,T]$  are such that the  process $\Radon$ defined by \eqref{Radon} satisfies 	$\E[\Radon^H_T]=1$.		
	\end{itemize}
\end{assumption}
\noindent We denote by
\begin{align}
	\label{set_admiss_Markoc} \mathcal{A}^H:=\Big\{\Pi:[0,T]\times
	\R^d\to\R^d: ~\Pi \text{ is a measurable function satisfying
		Assumption \ref{admi_stra_rule}}\Big\}
\end{align}%
the \textit{set of admissible decision rules}. 

\begin{remark}
	\label{remark on decision rule}  The integrability  condition 
		guarantees that the Radon-Nikodym density 	process $\Radon^{H}$ given in \eqref{b_underline} is 
		a martingale, hence	the equivalent measure	$\overline{\P}^H$ is well-defined.
	A Markov strategy $\pi=(\pi_t)_{t\in[0,T]}$ with
	$\pi_t=\Pi(t,\Mstate_t^{H})$ defined by an admissible decision rule
	$\Pi$ is contained in the set of admissible strategies $\mathcal
	{A}^H_0$ given in \eqref{set_admiss_0} since by construction it is
	$\mathbb{F}^{H}$-adapted, the positivity of the wealth process
	$\wealth^{\pi}$ follows from \eqref{Power_Nutzen_zerlegung}.   The integrability condition implies the
	square-integrability of $\pi$.  Finally, the Lipschitz and linear growth condition ensure that SDE \eqref{Filter_M_int_gen} for the dynamics for the controlled state process  enjoys for all admissible strategies a unique strong solution.
\end{remark}

\paragraph{Control problem} We are now ready to formulate the
stochastic optimal control problem with the state process $\Mstate^H
$ and a Markov control defined by the decision rule $\Pi$. The
dynamics of the  state process  $\Mstate^H$  are given in     \eqref{Filter_M_int_gen}.
We write $\Mstate^{H,\Pi,t,m}_s$  for the  state process at time $s\in[t,T]$  controlled by the decision rule $\Pi$ and  starting at   time $t\in[0,T]$  with initial value $\Mstate^H_0=m\in\R^d$.      
Note that $\Mstate^{H,\Pi,t,m}_s$  depends on the conditional covariance $\Qpro^H$ which is deterministic and can be computed offline.  Therefore $\Qpro^H$ needs not to be included as state process of the control problem. Further, we remove the initial value $\condcovinitial=\Qpro^H_t$ from the notation but keep in mind the dependence of $\Mstate^{H,\Pi,t,m}_s$ on $\Qpro^H$.

 To solve the control problem \eqref{opti-equivalent-generi}  we
will apply the  dynamic programming approach which requires the
introduction of the following reward and value functions. For all
$t\in[0,T]$ and $\Pi\in\mathcal{A}^H$ the associated reward function  or performance criterion
of the control problem \eqref{opti-equivalent-generi}  reads as
\begin{align}
	\label{Zielfkt_H} \reward^{H}(t,m;\Pi)&:= 
	\Ebar\Big[  \exp\Big\{ \int\nolimits_t^Tb(\Mstate^{H,\Pi,t,m}_s,\Pi(s,\Mstate^{H,\Pi,t,m}_s))ds \Big\}\Big],
	~~\text{for}~ \Pi\in\mathcal{A}^H,
\end{align}
while the associated value function reads 
\begin{align}
	\label{Wertfkt_H}
	\valuefkt^{H}(t,m)&:=\begin{cases}\sup\limits_{\Pi\in\mathcal{A}^H}\reward^{H}(t,m;\Pi),&\quad \theta\in(0,1), \\[2ex] \inf\limits_{\Pi\in\mathcal{A}^H}\reward^{H}(t,m;\Pi),&\quad \theta\in(-\infty,0), \end{cases}
\end{align}
and it holds $\valuefkt^{H}(T,m)=\reward^{H}(T,m;\Pi)=1$. In the
sequel we will concentrate on the case $0 < \theta < 1$, the
necessary changes for $\theta < 0$ will be indicated where
appropriate.

In view of relation \eqref{ExpUtility_RiskSensitive}  and the above
transformations the value function of the original utility
maximization problem \eqref{opti_org} can be obtained  from the
value function of control problem \eqref{Wertfkt_H} by
\begin{align}
	\label{Value_Orig_RiskSens}
	\mathcal{V}_0^H(x_0)=\frac{x_0^\theta}{\theta} \valuefkt^{H}(0,m_0).
\end{align}

\subsection{Full Information}
\label{FullInfo} The utility maximization problem for the case of
full information $H=F$  is already investigated in Kim und Omberg
\cite{Kim and Omberg (1996)}, Brendle \cite{Brendle2006} and  Davis and Lleo \cite[Chapter 2]{Davis and Lleo (2014)}.  In our analysis it
will serve as a reference case for the comparison with results for
partial information.  Recall that for $H=F$  the state process
is set to be  the drift, i.e.,  $\Mstate^F=\mu$ whereas for the partially informed investors the state $\Mstate^H$ is the conditional mean process under the measure $\overline{\P}^H$. In the  state equation
\eqref{Filter_M_int_gen}   for $H=\HF$  the coefficients read as
$\alphamm(\zustand,\variance,\pointp)=\revspeed(\revlevel-\zustand)$,
$\betam^F(q) =\voldrift$.
Below we only present the  results   for the associated control
problem \eqref{Wertfkt_H}  which will serve as a reference when we investigate the other
information regimes of partial information. For details we refer to  Kondakji \cite[Sec.~4.1, 5.1]{Kondkaji (2019)}.

\paragraph{Well posedness} We assume that the model parameters $\volR,\voldrift,\revspeed,T,\theta$ are such that the following  sufficient condition for the well posedness of the control problem under full information is satisfied. It is derived in our paper  \cite[Corollary 3.4]{Gabih et al (2022) Nirvana} and requires that   the terminal value problem for the Riccati equation
	\begin{align}
		\label{wellposed_full}
		\frac{d A_\cpsi (t)}{dt}=-2 A_\cpsi (t)\Sigma_{\drift}A_\cpsi (t)+\revspeed^{\top} A_\cpsi (t)+A_\cpsi (t)\revspeed-\cpsi\Sigma_R^{-1},\quad A_\cpsi(T)=0_{\nAktien\times\nAktien}
	\end{align}
	has for  $ \cpsi=\cpsi_0=\frac{\theta}{2(1-\theta)^2} $ a bounded solution on [0,T]. That condition implies restrictions to the choice of model parameters for $\theta\in(0,1)$, that is for investors which are less risk-averse than log-utility maximizing investor. For $\theta<0$ it is always fulfilled.  This follows from Theorem 5.2 in Fleming and Rishel \cite{Fleming and Rishel (1975)}. For $\theta<0$  we have $\cpsi_0<0$ and $\cpsi_0\Sigma_R^{-1}$ is negative semindefinite. Then the above theorem says that  the solution $A_{\cpsi_0}$ to the  Riccati ODE \eqref{A_bound} does not explode on $(-\infty,T)$ and is thus always bounded on $[0,T]$.
\begin{theorem}\label{kandidate_full}
	For the control problem \eqref{Wertfkt_H} under full information ($H=\HF$) the optimal decision rule is  for $t\in[0,T), m\in \R^d$ given by
	\begin{align}
		\label{optimal_str_F} \Pi^{\HF}= \Pi^{\HF}(t,m)=\frac{1}{(1-\theta)}
		\Sigma_R^{-1} m.
	\end{align}
	The value function  for $t\in[0,T], m\in \R^d$ is given by
	\begin{align}
		\label{Ansatz_DPE_F} \valuefkt^{\HF}(t,m)=\exp\Big\{ m^{\top}
		A^{\HF}(t) m+m^{\top}B^{\HF}(t)+C^{\HF}(t)  \Big\},
	\end{align}
	where  $A^{\HF}$, $B^{\HF}$ and $C^{\HF}$ are  functions
	on
	$[0,T]$ taking values in $\mathbb R^{{\nAktien}\times
		{\nAktien}}$, $\mathbb R^{\nAktien}$ and $\mathbb R$,  respectively,  satisfying a terminal value problem for the following  system of ODEs  
	\begin{align}
		\label{A_bound}
		\frac{d \Abound (t)}{dt}&=-2 \Abound (t)\Sigma_{\drift}\Abound (t)+\revspeed^{\top} \Abound (t)+\Abound (t)\revspeed-\frac{1}{2}\frac{\theta}{1-\theta}\Sigma_R^{-1},&& \Abound(T)=0_{\nAktien\times\nAktien},\\[2ex]
		\label{B_bound}
		\frac{d\Bbound(t)}{dt}&=-2\Abound(t)\revspeed\overline{\drift}+\big[\revspeed^{\top}-2\Abound(t)\Sigma_{\drift}\big] \Bbound(t),&&
		\Bbound(T)=0_{\nAktien}, \\[2ex]
		\label{C_bound}
		\frac{d\Cbound(t)}{dt}&=-\frac{1}{2}(\Bbound(t))^{\top} \Sigma_{\drift} 
		\Bbound(t)-(\Bbound(t))^{\top}\revspeed\overline{\drift}-\trace\{\Sigma_{\drift} \Abound(t)\},&&
		\Cbound(T)=0.
	\end{align}			
\end{theorem}
A proof and a  detailed verification analysis can be found in Davis and Lleo \cite[Chapter 2]{Davis and Lleo (2014)}.

\paragraph{Boundedness of  $\mathbf{A^F,B^F,C^F}$}
The differential equations for $A^F,B^F,C^F$ are coupled. The ODE for $A^F$ is an autonomous matrix Riccati ODE which can be solved first. Given the solution to $A^F$, one can solve the linear  ODE for $B^F$, and then one can find $C^F$ by integrating the r.h.s.~of the last ODE.	
Therefore, a bounded solution $A^F$ implies boundedness of $B^F$ and $C^F$ on $[0,T]$. 
	Since the Riccati equation for $\Abound$ is a special case of \eqref{wellposed_full}   it is sufficient to require the condition
	\begin{align}
		\label{A_bounded_full}
		\text{the solution $ A_{\cpsi}$ to \eqref{wellposed_full} is for    $ \cpsi=\cpsi^F=\frac{\theta}{2(1-\theta)}  $  bounded on [0,T], }
	\end{align}
	to be satisfied. For $\theta<0$  we have $\cpsi^F<0$ and as above for $A_{\cpsi_0}$, from  Fleming and Rishel \cite[Theorem 5.2]{Fleming and Rishel (1975)} it follows that the solution $\Abound=A_{\cpsi^F}$ to \eqref{wellposed_full} does not explode on $(-\infty,T)$ and is thus always bounded on $[0,T]$.
	
	For $\theta\in(0,1)$ and  the scalar case, that is $d=1$, condition \eqref{A_bounded_full} follows immediately from  the  well posedness condition in  \eqref{wellposed_full}.		
	To see this, we introduce for  $G\in\R$ the notation $h_\cpsi(G)$ for the r.h.s. of \eqref{wellposed_full} such that this ODE  reads $\frac{d}{dt} A_{\cpsi_0}=h_{\cpsi_0}(A_{\cpsi_0})$ and the Riccati equation for $A^F=A_{\cpsi^F}$ as $\frac{d}{dt} A_{\cpsi^F}=h_{\cpsi^F}(A_{\cpsi^F})$.   Since for $\theta\in(0,1)$ it holds $\cpsi_0>\cpsi^F>0$  
	and $\Sigma_R^{-1}$ is positive  it holds for the r.h.s. of the above ODEs $h_{\cpsi_0}(G)< h_{\cpsi^F}(G)$. 
	Since the terminal conditions $ A_{\cpsi_0}(T)=A_{\cpsi^F}(T)= 0$ are the same for both equations   it can be deduced that the solutions satisfy $	A_{\cpsi_0}(t) \ge A_{\cpsi^F}(t)=\Abound(t)$ for  $t\in[0,T)$. 
	From Roduner \cite[Theorem 1.2]{Roduner (1994)} it follows that for $\gamma>0$ the solutions of the Riccati equations are nonnegative on $[0,T]$. Thus, we have  $A_{\cpsi_0}(t) \ge \Abound(t) \ge  0$ on $[0,T)$ and the  boundedness of $A_{\cpsi_0}$ implies that  $\Abound$ is bounded.

\smallskip
\paragraph{Monetary value of information} In order to quantify  the
monetary value of    information  we have introduced in
Subsection \ref{MonetaryValue}  the quantity $x_0^{H/F}$ for which
we have to evaluate an expectation  given in
\eqref{monetaer_Wert}.     
The latter   can be given in terms of the functions $A^{\HF}$,
$B^{\HF}$, $C^{\HF}$ appearing in the above theorem. Using \eqref{Value_Orig_RiskSens} we find  for $H=\HR,\HD,\HC$		
\begin{align}
	\E \Big[\mathcal V_0^{\HF}\big(x_0^{H/ \HF}\big)\mid \mathcal
	F_0^{H}\Big] =\E \bigg[\frac{\big(x_0^{H/
			\HF}\big)^{\theta}}{\theta} V^{\HF}(0,\drift_0)\mid \mathcal
	F_0^{H}\bigg] =\frac{\big(x_0^{H/
			\HF}\big)^{\theta}}{\theta}\E \Big[
	V^{\HF}(0,\drift_0)\mid\Mpro_0^{H},\Qpro_0^{H}\Big].\nonumber
\end{align}
For $\Mpro_0^{H}=\zustand$ and $\Qpro_0^{H}=\variance$ the
conditional distribution of  $\drift_0$ given $\mathcal F_0^{H}$ is
the normal distribution $\mathcal N(\zustand,\variance)$, i.e.
$\drift_0=\zustand+\variance^{\frac{1}{2}} \varepsilon$ with
$\varepsilon\sim\mathcal N(0,I_{\nAktien})$. Hence we have
\begin{align}
	\E \Big[
	V^{\HF}(0,\drift_0)\mid\Mpro_0^{H}=\zustand,\Qpro_0^{H}=\variance\Big] =
	\E \Big[ V^{\HF}\big(0, \zustand+\variance^{1/2}
	\varepsilon\big)\Big].\nonumber
\end{align}
This approach can be extended to the case of arbitrary  time points  
$t\in[0,T]$ with $ \Mstate_t^H=\zustand$ and $\Qpro_t^H=\variance$ by	means of  the following function
\begin{align}
	\overline{V}^{\HF}(t,\zustand,\variance):=\E
	\Big[V^{\HF}\big(t,\zustand+\variance^{\frac{1}{2}} \varepsilon\big)
	\Big],\qquad \forall t\in[0,T]. \label{mitteled_VF}
\end{align}
The following Lemma gives an explicit form for computing
$\overline{V}^{\HF}(t,\zustand,\variance)$.
\begin{lemma}[Kondakji \cite{Kondkaji (2019)}, Lemma 5.1.3] \ \\
	\label{value_F_average}%
	Under the assumption that the eigenvalues of the matrix  $I_{\nAktien}-2A^{\HF}(t)\variance$ are positive,
	it holds
	\begin{align}
		\label{F_mittel}
		\overline{V}^{\HF}(t,\zustand,\variance)=\exp\Big\{ \zustand^{\top}\,\overline{A}^{\HF}(t,\variance)\, \zustand+\zustand^{\top}\overline{B}^{\HF}(t,\variance)\zustand+\overline{C}^{\HF}(t,\variance)\Big\},
	\end{align}
	for $t\in[0,T], \zustand\in \R^d, \variance\in \R^{d\times d}$, where
	\begin{align}
		\overline{A}^{\HF}(t,\variance) &=(I_d-2A^{\HF}(t) \variance)^{-1} A^{\HF}(t),\nonumber\\
		\overline{B}^{\HF}(t,\variance) &=(I_d-2A^{\HF}(t) \variance)^{-1} B^{\HF}(t),\nonumber\\
		\overline{C}^{\HF}(t,\variance) &=C^{\HF}(t)+\frac{1}{2}\big(B^{\HF}(t)\big)^{\top} (I_d- 2{A^{\HF}(t) \variance  })^{-1}\variance B^{\HF}(t)-\frac{1}{2}\log \rm{det}(I_d-2{A^{\HF}(t) \variance  })\nonumber
	\end{align}		
	with $A^{\HF}$, $B^{\HF}$ and $C^{\HF}$  given in Theorem \ref{kandidate_full}.	
\end{lemma}

\section{Partially Informed Investors Observing Diffusion Processes}
\label{UtilityMaxDiffusion} 
In this section we start to solve to the control problem \eqref{Wertfkt_H} for partially informed investors. We consider investors  observing the diffusion processes $R$ and/or $J$, i.e., the information regimes $H=R,\HD$.  The case of the information regime $H=\HC$ with discrete-time expert opinions follows in Sec.~\ref{UtilityMaxDiscreteExperts}.

	\paragraph{Well posedness} In   \cite[Corollary 3.7]{Gabih et al (2022) Nirvana}  we show for $\theta <0$ the control problem \eqref{Wertfkt_H} for $H=R,\HD,\HC$ is always well posed. However, for $\theta\in(0,1)$ in addition to the condition 	
	\eqref{wellposed_full}   ensuring well posedness of the control problem for the fully informed investor one now  also has to impose the condition 
	\begin{align}
		\label{wellposed_partial}
		\text{the eigenvalues of } ~I_{\nAktien}-2  A_{\gamma_0}(t)\Qpro^H_t  \text{ ~~are positive  on $[0,T]$ for } \cpsi_0=\frac{\theta}{2(1-\theta)^2} .  
	\end{align}
	Condition \eqref{wellposed_partial} says that the conditional covariance $\Qpro^H$   must not be ``too big''  such that  the eigenvalues of $I_{\nAktien}-2  A_{\gamma_0}(t)\Qpro^H_t$  are  positive. Recall, $\Qpro^H$  is deterministic and can be computed offline,  and from Proposition \ref{properties_filter} we know that  $\Qpro^H$ is bounded.

\smallskip	
For solving control problem \eqref{Wertfkt_H} we apply dynamic programming techniques. Starting point is the dynamic programming principle given in the following lemma. For the proof we refer to Frey et al.~\cite[Prop.~6.2]{Gabih et al (2014)} and  Pham~\cite[Prop.~3.1]{Pham (1998)}.		
\begin{lemma}(Dynamic Programming Principle).
	\label{lemma_DPP} For every $t\in[0,T]$, $\zustand\in\R^{\nAktien}$ and for
	every stopping time $\tau$ with values in $[t,T]$ it
	holds		
	\begin{align}
		\valuefkt^{H}(t,\zustand)= 
		\sup\limits_{\Pi\in\mathcal{A}^H} \Ebar  \bigg[\exp\Big\{\int_t^{\tau} b(\Mstate_s^{H,\Pi,t,\zustand},\Pi(s,\Mstate_s^{H,\Pi,t,\zustand}))ds\Big\}
		\valuefkt^{H}(\tau,\Mstate_{\tau}^{H,\Pi,t,\zustand}) \bigg].~~ \label{DPP_Y}
	\end{align}
\end{lemma}
From the dynamic programming principle the dynamic programming equation (DPE) for the value function presented in Theorem \ref{DPE_general} can be deduced. That equation  constitutes a necessary optimality condition  and allows to derive the optimal decision rule.  We recall that we focus on  the solution for $\theta\in(0,1)$, the case  $\theta <0$ follows analogously by changing  $\sup$  into $\inf$ in  \eqref{DPP_Y}. 
For convenience we introduce the shorthand notation 
\begin{align}
	\label{Sigma_overline}
	\overline \Sigma_{H}^{-1}=\begin{cases}
		\Sigma_{\HR}^{-1}, & H=R,\HC\\
		\Sigma_{\HR}^{-1}+\Sigma_{\HD}^{-1}, & H=\HD.
	\end{cases}
\end{align}

\begin{theorem}[Dynamic programming equation]%
	\label{DPE_general}~
	\begin{enumerate} 
		\item In the case of diffusion type
		observations. i.e., $H=\HR, \HD$, the value function $\valuefkt^H$ satisfies  for  $t\in[0,T)$ and  $m\in \R^d$ the PDE 
		\begin{align}
			0&=\frac{\partial}{\partial t} \valuefkt^H(t,\zustand)+\diffop_{\zustand}^{\top} \valuefkt^H(t,\zustand)\Big(\revspeed(\revlevel-\zustand)+\frac{\theta}{(1-\theta)} \Qpro_t^{H} \Sigma_{R}^{-1} \zustand \Big)\nonumber\\
			&\quad+\frac{1}{2}\trace\Big\{\diffop_{\zustand\zustand}\valuefkt^{H}(t,\zustand)\Big( \Qpro_t^{H} \overline \Sigma_{H}^{-1}\Qpro_t^{H}\Big)\Big\}+\frac{\theta}{2(1-\theta)}\zustand^{\top}\Sigma_{R}^{-1} \zustand \valuefkt^H(t,\zustand)\nonumber\\
			&\quad+\frac{\theta}{2(1-\theta)}\frac{1}{\valuefkt^H(t,\zustand)}
			\diffop_\zustand^{\top}\valuefkt^H(t,\zustand)\Big(\Qpro_t^{H}\Sigma_{R}^{-1}\Qpro_t^{H}\Big)\diffop_{\zustand}\valuefkt^H(t,\zustand),
			\label{HJB_Power_general}
		\end{align}
		with the terminal condition $\valuefkt^H(T,\zustand)=1$ 				 
		and $\diffop_{\zustand} V^H$, $\diffop_{\zustand\zustand} V^H$ denoting the gradient and Hessian matrix of $V^H$, respectively.
		\item
		The  candidate optimal decision rule   is for $t\in [0,T)$ and  $m\in \R^d$ given by 				
		\begin{align}
			\label{opti_stra_general}
			\Pi^{H}=\Pi^{H}(t,\zustand)&
			=\frac{1}{1-\theta}\Sigma_{R}^{-1}\left(\zustand+\frac{1}{\valuefkt^H(t,\zustand)} \Qpro_t^{H} \diffop_{\zustand}\valuefkt^H(t,\zustand) \right). 
		\end{align}
	\end{enumerate}
\end{theorem}
\begin{proof}
	Let  $t,\tau \in [0,T)$ with   $\tau>t$  for some  fixed time point $t$. Then   the dynamic programming
	principle \eqref{DPP_Y} and the continuity of $\Mstate_s^{H,\Pi,t,\zustand}$ on  $[0,T)$ imply
	\begin{align}
		\nonumber \hspace*{-0.5em}
		\valuefkt^H(t,\zustand) & =\lim\limits_{\tau\searrow	t} \valuefkt^H(\tau,\zustand)\\
		&=\lim\limits_{\tau\searrow    t}\;\sup\limits_{\Pi\in\mathcal{A}^H} 
		\Ebar \Big[\exp\Big\{\int_t^{\tau} b(\Mstate_s^{H,\Pi,t,\zustand},\Pi(s,\Mstate_s^{H,\Pi,t,\zustand}))ds\Big\}
		\valuefkt^{H}(\tau,\Mstate_{\tau}^{H,\Pi,t,\zustand}) \Big].
		\label{DPP_Y_Grenzwert}
	\end{align}
	For the state process $\Mstate=\Mstate^{H,\Pi,t,\zustand}$ given in \eqref{Filter_M_int_gen} the associated generator
	$\generator=\generator^{p}$ applied to a function $g\in
	C^2(\R^{\nAktien})$ for fixed $p=\Pi(t,\zustand)$ reads
	\begin{align}
		\mathcal{L}g(\zustand)&=\diffop_m^{\top}g(\zustand)\alphamm(\zustand,q,\pointp)
		+\frac{1}{2}\trace\Big\{\diffop_{mm}
		g(\zustand)\betam^{H}(q)\big(\betam^{H}(q)\big)^{\top}\Big\}.
		\nonumber
	\end{align}
	From \eqref{DPP_Y_Grenzwert}  we obtain using Dynkin's formula for $\valuefkt^{H}(\tau,\Mstate_{\tau}^{H,\Pi,t,\zustand})$ and standard arguments of the dynamic programming approach the following PDE
	\begin{align}
		0=&\frac{\partial}{\partial t} V_k(t,\zustand)+\diffop_{\zustand}^{\top} \valuefkt^H(t,\zustand)\revspeed(\revlevel-\zustand)+\frac{1}{2}\trace\Big\{\diffop_{\zustand\zustand}\valuefkt^H(t,\zustand) \Qpro_t^{H}\overline \Sigma_{H}^{-1}\Qpro_t^{H}\Big\}\nonumber\\
		&+\sup\limits_{\pointp\in\mathbb
			R^{\nAktien}}\Big\{\diffop_{\zustand}^{\top} \valuefkt^H(t,\zustand)
		\theta\Qpro_t^{H}\pointp+\theta\Big(\pointp^{\top}\zustand
		-\frac{1-\theta}{2}\pointp^{\top}\Sigma_{R}\pointp\Big)
		\valuefkt^H(t,\zustand)\Big\}. \label{Bell_equa_power_C_deter}
	\end{align}		
	The maximizer for the supremum appearing
	in \eqref{Bell_equa_power_C_deter} yields the 	optimal decision rule $\Pi^{H}=\Pi^{H}(t,\zustand)$ which is given in \eqref{opti_stra_power_C_det}. 
	Plugging the expression for  the maximizer $\Pi^{\HC}$ back into the DPE
	\eqref{Bell_equa_power_C_deter}  we obtain the  PDE \eqref{HJB_Power_general}. 			
\end{proof}
%
The above dynamic programming equation can be solved using an ansatz as in \eqref{value_power_RJ} below and leads to closed-form expressions for the value function $V^H$ and the optimal decision rule $\Pi^H$ in terms of solutions of some ODEs. The results are given in the next theorem and are known for $H=R$ already from Brendle \cite{Brendle2006}. For the case  $H=\HD$ and details of the proof we refer  Kondakji
\cite[Sec.~5.2]{Kondkaji (2019)}.

\begin{theorem}[Solution of DPE and optimal decision rule]~
	\label{Kandidate_diffusion} 
	\begin{enumerate} 
		\item In the case of diffusion type	observations,  that is  $H=\HR, \HD$, the solution to the dynamic programming equation  equation \eqref{HJB_Power_general} is given for $t\in[0,T], m\in \R^d$  by
			\begin{align}
				\label{value_power_RJ}
				\valuefkt^{H}(t,m)=\exp\left\{m^{\top} A^{H}(t) m+m^{\top}B^{H}(t)
				+C^{H}(t)\right\}.
			\end{align}  
			The functions $A^{H}(t)$, $B^{H}(t)$ and $C^{H}(t)$ staisfy the system of ODEs
			\begin{align}
				\label{DGL_A_diff}
				\frac{dA^H(t)}{dt}
				&=-2A^H(t)\Qpro_t^{H}  \Big(\frac{\theta }{1-\theta} \Sigma_{\HR}^{-1} +\overline\Sigma_{H}^{-1}\Big)  \Qpro_t^{H} A^H(t) +\Big(\revspeed^{\top} - \frac{\theta }{1-\theta} \Sigma_{\HR}^{-1}\Qpro_t^{H} \Big) A^H(t)\\[-2ex]
				& ~~~~+A^H(t)\Big(\revspeed - \frac{\theta }{1-\theta} \Qpro_t^{H}\Sigma_{\HR}^{-1} \Big)
				-\frac{\theta}{2(1-\theta)}\Sigma_R^{-1},\\[1ex]
				\label{DGL_B_diff}
				\frac{dB^H(t)}{dt}&=\left[ \revspeed^{\top}\!\!-2A^H(t)\Qpro_t^{H} \Big(\frac{\theta }{1-\theta} \Sigma_{\HR}^{-1} +\overline\Sigma_{H}^{-1}\Big) \Qpro_t^{H}\!-\frac{\theta}{1-\theta}\Sigma_{R}^{-1}\Qpro_t^{H}\right] B^H(t) \\
				&~~~~-2A^H(t)\revspeed\overline{\drift},\\[1ex]
				\label{DGL_C_diff}
				\frac{dC^H(t)}{dt}&=-(B^H(t))^{\top}\Big[\revspeed\revlevel 
				+\frac{1}{2}\Qpro_t^{H}\Big(\frac{\theta }{1-\theta} \Sigma_{\HR}^{-1} +\overline\Sigma_{H}^{-1}\Big) \Qpro_t^{H} B^H(t)\Big] \\
				&~~~~-\trace\big\{ \Qpro_t^{H}\Sigma_{R}^{-1}\Qpro_t^{H} A^H(t) \big\},
			\end{align}
			with terminal values $A^H(T)=0_{d\times d}, B^H=0_d, C^H=0$.			
		\item
		The candidate optimal decision rule  is for for $t\in[0,T), m\in \R^d$, is given by
		\begin{align}
			\label{opti_stra_power}
			\Pi^{H}(t,m)& = \frac{1}{1-\theta}\Sigma_{R}^{-1} \,\big(m+ \Qpro^H_t (2A^{H}(t)m+B^{H}(t))   \big).
		\end{align}			
	\end{enumerate}
\end{theorem}
The next proposition shows that the functions $A^H, B^H,C^H$  solving  the system of ODEs in Theorem \ref{Kandidate_diffusion}  can be expressed by the solutions $A^F,B^F,C^F$ to the ODEs for the full information problem given in Theorem \ref{kandidate_full}  via the functions $\overline{A}^{\HF}, \overline{B}^{\HF}, \overline{C}^{\HF}$  given in  Lemma \ref{value_F_average}. The proof is given in \cite[Lemma 5.2.1]{Kondkaji (2019)}. The result will facilitate the proof of boundedness of $A^H, B^H,C^H$.   
\begin{proposition}
	\label{prop_AH}			
	 If the eigenvalues of  $I_{\nAktien}-2A^{\HF}(t)\Qpro_t^H$ are positive for all $t\in[0,T]$ then it holds 
	for the functions $A^{H}$, $B^{H}$ and $C^{H}$  on  $[0,T]$
	\begin{align}
		\label{ABCH}
		A^{H}(t)=\overline{A}^{\HF}\big(t,\Qpro_t^{H}\big),~
		B^{H}(t)=\overline{B}^{\HF}\big(t,\Qpro_t^{H}\big), ~
		C^{H}(t)=\overline{C}^{\HF}\big(t,\Qpro_t^{H}\big)-\theta\Delta_X^{H}(t),
	\end{align}
	where the functions $\overline{A}^{\HF}, \overline{B}^{\HF}$ and $\overline{C}^{\HF}$ are given in  Lemma \ref{value_F_average} and			
	\begin{align}
		\Delta_X^{H}(t):=\frac{1}{2}\log \frac{det(I_{\nAktien}-2{\xi^{H}(t)\Qpro_t^{H}})}{det(I_{\nAktien}-2{A^{\HF}(t)\Qpro_t^{H}})}+ \underline{K}^{H}(t)-\overline{K}^{H}(t),\quad t\in[0,T].
		\label{Delta_X_RD}
	\end{align}	
	The function  $\xi^{H}(t)$ satisfies  on $[0,T]$ the Riccati equation
	\begin{align}\label{riccati_xi}
		\frac{d\xi^{H}(t)}{dt}=-2\xi^{H}(t)\Sigma_{\drift}\xi^{H}(t)+\revspeed^{\top} \xi^{H}(t)+\xi^{H}(t)\revspeed
		+\frac{1}{2}\overline \Sigma_{H}, ~\xi^{H}(T)=0,
	\end{align}
	where $\overline \Sigma_{H}$ is given in \eqref{Sigma_overline}.
	The functions $\underline{K}^{H}(t), \overline{K}^{H}(t)$ are given by
	\begin{align*}
		\underline{K}^{H}(t)&=\int\nolimits_t^T \trace\{\Sigma_{\drift}\big(A^{\HF}(u)-\xi^{H}(u)\big)\}du,\\
		\overline{K}^{\HD}(t) &=\frac{1}{2}\int\nolimits_t^T \trace\{\Qpro_u^{\HD}\Sigma_{\HD}^{-1}(I_{\nAktien}-2\Qpro_u^{\HD}
		A^{\HF}(u))\}du\quad\text{and}\quad\overline{K}^{R}(t)=0.
	\end{align*}
	
\end{proposition}

	\paragraph{Boundedness of $\mathbf{A^H,B^H,C^H}$} In view of Proposition \ref{prop_AH}	this property holds under  condition \eqref{A_bounded_full} saying that $\Abound$ is bounded,  and if   the eigenvalues of  $I_{\nAktien}-2A^{\HF}(t)\Qpro_t^H$ are positive on $[0,T]$. Note that if $\Abound$ is bounded then $\Bbound$ and $\Cbound$ are also bounded.
	Thus, the expressions $\overline{A}^{\HF}\big(t,\Qpro_t^{H}\big), \overline{B}^{\HF}\big(t,\Qpro_t^{H}\big), \overline{B}^{\HF}\big(t,\Qpro_t^{H}\big)$ in \eqref{ABCH} with  $\overline{A}^{\HF}, \overline{B}^{\HF}$ given in  Lemma \ref{value_F_average} are  bounded on $[0,T]$. This proves the boundedness of $A^H$ and $B^H$. Finally, the boundedness of $C^H$ follows from integrating the r.h.s.~of the ODE for $C^H$ given in Theorem \ref{Kandidate_diffusion}, which is bounded. 
	
	Recall, for $\theta<0$ we have that $\Abound$ is always bounded  and negative semindefinite. Analogously to the approach in our paper \cite[Section 3]{Gabih et al (2022) Nirvana} one can show that  the eigenvalues of  $K(t)=I_{\nAktien}-2\Abound(t)\Qpro_t^H$ are bounded below by $1$ and thus positive. However, for $\theta\in(0,1)$ the  eigenvalues of $K(t)$ are positive only  if the conditional covariance $\Qpro^H_t$ is on $[0,T]$ not ``too large''.  Since if $\lambda$ is an eigenvalue of $\Abound(t)\Qpro_t^H$ then $1-2\lambda$ is an eigenvalue of  $K(t)$  one has to require that 	 	$\lambda_{\max}(\Abound(t)\Qpro^H_t)<1/2$ 
	for all $t\in[0,T]$. Here, $\lambda_{\max}(G)$  denotes the largest eigenvalue of a generic matrix $G$.

	\smallskip
	\paragraph{Verification}
	We derived the above  candidate solution to the control problem  \eqref{Wertfkt_H} using the classical stochastic control approach.  To ensure that the solution to the DPE is indeed the value function $V$, the candidate optimal decision rule $\Pi ^H$ indeed satisfies $ \valuefkt^{H}(t,m)=\reward^{H}(t,m;\Pi^H)$, and that $\Pi^H$ defines an optimal strategy process via $\pi^H_t=\Pi^H(t,\Mstate^H_t)$ that is admissible,  one needs to prove  a suitable verification theorem. Such a verification theorem is given in Hata and Sheu \cite[Theorem 4.1]{Hata Sheu (2018)} for a slightly general setting. That paper considers a combined consumption and investment problem in which the investor is also allowed to consume portfolio wealth and aims to maximize the expectation of the sum of the power utility of terminal wealth and aggregated running power utility of consumption. Further,  \cite{Hata Sheu (2018)} allows for correlation between return and drift process, non-negative interest rate for the risk-free asset, and  for discounting the utility. The findings can be directly adopted to the control problem \eqref{Wertfkt_H} for $H=R,\HD$ if the utility for consumption formally is set to zero, leading to zero optimal consumption, and removing the consumption from the strategy process. 
	
	In \cite{Hata Sheu (2018)} the authors take advantage of the consideration of a logarithmic  transformation of the performance criterion and study the associated DPE for $\log V^H$. One of the key assumptions for the above mentioned verifications results is the boundedness of the  functions $A^H,B^H,C^H$ on $[0,T]$. Further, the linearity of the optimal decision rule $\Pi^H$ w.r.t. the state variable $m$, see \eqref{opti_stra_power}, is exploited to prove, that the optimal strategy process $\pi^H$ generated by $\Pi^H$, is admissible. In particular, using a result of Bensoussan \cite[Lemma 4.1.1]{Bensoussan (1992)} which is also given in Nagai \cite[Lemma 5.1]{Nagai (2015)}, it can be deduced that  the associated density process $\Radon^H$  defining the change of measure and given in \eqref{Radon} satisfies $\E[\Radon_T^{H}]=1$.

	\begin{remark}\label{rem:myopic}
		The optimal decision rules given in \eqref{opti_stra_general} and \eqref{opti_stra_power} can be rewritten as 
		\begin{align} \label{strat_myopic}				
			\Pi^{H}(t,\zustand)	
			&= \Pi^{F}(t,\zustand)+\frac{1}{(1-\theta)\valuefkt^H(t,\zustand)}\Sigma_{R}^{-1}\Qpro_t^{H}\diffop_{y}\valuefkt^H(t,\zustand),\\
			&=  \Pi^{F}(t,\zustand)+ \frac{1}{1-\theta}\Sigma_{R}^{-1} \,\Qpro^H_t (2A^{H}(t)m+B^{H}(t))  .
		\end{align}
		  Thus, $\Pi^{H}$ can be decomposed into two parts. The first part  $\Pi^{F}$ is the optimal decision rule of the fully informed investor given in \eqref{optimal_str_F}. To obtain the value of the strategy process $\pi_t$ at time $t$ the fully informed investor plugs in for $m$  the current value of the drift $\drift_t$,  whereas  the partially informed $H$-investor plugs in the filter estimate $\Mpro^H_t$. 	
			In the literature  $\Pi^{F}$ is also known as  {\em myopic decision rule}. The second part is  the ``correction term''  $\Pi^{H}- \Pi^{F}$ which is known as  {\em drift risk} of the partially informed $H$-investor since it accounts for the investor's uncertainty about the current value of the non-observable drift, see Rieder and Bäuerle \cite[Remark 1]{Rieder_Baeuerle2005} and Frey et al.~\cite[Remark 5.2]{Frey et al. (2012)}.  
			The decomposition shows that, in contrast to the case of log utility (see Sass et al. (2017)), the so-called  {\em certainty equivalence principle} does not apply to power utility. 		 	It states  that the optimal strategy under partial information
			is obtained by replacing the unknown drift $\drift_t$ by the filter estimate $\Mpro^H_t$ 	in the formula for the optimal strategy under full information. 	 	
	\end{remark}
	
	\begin{remark}
		\label{rem_log_power} It is well-known that log-utility can be embedded in the family of power utilities using the relation $\utility_\theta(x)-1/\theta= (x^\theta -1)/\theta \to \log x$ for $\theta\to 0$.  Replacing $\utility_\theta(x)$ by $\utility_\theta(x)-1/\theta$ in the utility maximization problem \eqref{opti_org}  will only lead to an additive shift of the value function while the optimal strategy remains unchanged. 		
		Hence, it can be expected that the optimal decision rule given \eqref{opti_stra_power}  converges for $\theta\to 0$ to the optimal decision rule for log-utility $\Pi^{H}_{\log}(t,m)=\Sigma_{R}^{-1}m$ given in   Proposition \ref{LogUtilityValue}. An in fact, this is the case, since for $\theta\to 0$ the solutions to the ODEs for $A^H,B^H$ converge to zero. First, the r.h.s.~of the Riccati ODE for $A^H$ is a  quadratic expression of $A^H$ with the absolute term $-\frac{\theta}{2(1-\theta)}\Sigma_R^{-1}$. The latter   vanishes for $\theta=0$. Since the terminal value at time $t=T$ is zero, the solution $A^H$ is zero on $[0,T]$. Thus, the linear ODE for $B^H$ is homogeneous with zero terminal value, and the solution $B^H$ is zero on $[0,T]$. Finally, the correction term or drift risk $\Pi^{H}- \Pi^{F}$  containing $2A^{H}(t)m+B^{H}(t)$ vanishes for $\theta\to 0$.
\end{remark}   
By comparing the optimal decision rules for the $\HR, \HD$ and
$\HF$-investor for identical values $m$ of the conditional mean and
$q$ for the conditional covariance, some interesting properties
result, which are are formulated in the following lemma. For the
proof we refer to \cite[ Lemma 5.2.4 and  5.2.5]{Kondkaji
	(2019)}.

\begin{lemma}~
	\label{CoincidenceDecRule}
	\begin{enumerate}
		\item If  ~$\Mstate_t^{\HR}=\Mstate_t^{\HD}=m$~ and
		$\Qpro_t^{\HR}=\Qpro_t^{\HD}=q$~ then it holds
		$
		\Pi^{\HR}(t,m)=\Pi^{\HD}(t,m).
		$
		
		\item If ~$\Mstate_t^{H}=\,\drift_t=m$ and $\Qpro_t^{H}=0$~  then it holds
		$
		\Pi^{\HF}(t,m)=\Pi^{H}(t,m)$ for $H=\HR,\HD$.		
	\end{enumerate}
\end{lemma}
\paragraph{Monetary value of information} Theorem \ref{Kandidate_diffusion} allows  to derive  explicit
expressions for the initial investment $x_0^{\HR/H}$ and
$x_0^{H/\HF}$ introduced in Subsection \ref{MonetaryValue} and
needed to evaluate the efficiency of the $\HR$- and $\HD$-investor
and the  monetary value of expert opinions.

\begin{lemma}[{Kondkaji (2019)}, Lemma 5.3.1]
	\label{monetary_value_RJ}
	For the initial capital $x_0^{H/\HF}$  and the efficiency
	$\varepsilon^{H}$   defined
	in  \eqref{monetaer_Wert} and \eqref{effekt_formel}, respectively,   it holds for $H=\HR, \HD$
	$$
	x_0^{H/\HF}=x_0^H\exp\{-\Delta_X^H(0)\}\quad  \text{and}\quad
	\varepsilon^{H}=\exp\{-\Delta_X^H(0)\}.
	$$
	For the initial capital $x_0^{\HR/\HD}$  defined in \eqref{initial_cap_H} and the monetary value of expert opinions $P_{Exp}^\HD$ it holds
	$$
	P_{Exp}^\HD=x_0^R-x_0^{\HR/\HD}~~~
	\text{with}~~~x_0^{\HR/\HD}=\exp\{-\Delta_X^{\HD}(0)+\Delta_X^{\HR}(0)\},
	$$
	where $\Delta_X^{\HD}(0)$ and $\Delta_X^{\HR}(0)$ are given in
	\eqref{Delta_X_RD}.
\end{lemma}

\section{
	Partially  Informed Investors  Observing Discrete-Time Expert Opinions}
\label{UtilityMaxDiscreteExperts} 
After solving to the control problem \eqref{Wertfkt_H} for partially informed investors observing the diffusion processes this section presents the solution for investors observing returns and discrete-time expert opinions, i.e., the information regime $H=\HC$.  As in Sec.~\ref{UtilityMaxDiffusion} for $H=R,\HD$ we impose the pair of well posedness conditions \eqref{wellposed_full} and \eqref{wellposed_partial}.
We again apply the dynamic programming principle to derive the DPE for the value function and introduce the notation 
$$V_k(t,\zustand)= \left\{
\begin{array}{cl}
	V^{\HC}(t,\zustand) & \text{for }t\in[t_{k-1},t_{k})\\
	\vterm_k(y)=V^{\HC}(t_k-,\zustand)=\lim\limits_{t\nearrow t_k} V^{\HC}(t,\zustand) & \text{for }t=t_{k}
\end{array}\right. $$	
for    $k=1,\ldots,n$, i.e., $V_k:[t_{k-1},t_{k}]\to \R$ denotes the value function on the $k$-th interval between two subsequent information dates and  $\vterm_k(y)$ its left-hand limit at $t_k$. Note that for $t_n=T$ we have $\vterm_n(\zustand)=V(T,\zustand)=1$.
\begin{theorem}[Dynamic programming equation]%
	\label{backward_recursion_I}~
	\begin{enumerate} 
		\item   Let the value function be defined piecewise  for  $t\in [t_{k-1},t_{k}),~k=1,\ldots n,$ and $m\in \R^d$  by  $V^{\HC}(t,\zustand)=V_k(t,\zustand)$. Then  the functions $V_k$ satisfy  on $(t_0,t_1)$ and $[t_{k-1},t_{k}),~k=2,\ldots,n$, the PDE \eqref{HJB_Power_general}  with the terminal conditions
		\begin{align}
			\valuefkt_k(t_k,\zustand)=\vterm_k(\zustand)&=\EbarZ
			\Big[V_{k+1}\big(t_{k},\Mstate_{t_{k}}^{\Pi,t_{k}-,\zustand}\big)\Big], \quad  k=1,\ldots,n-1.
			\label{dynamic_P_P}
		\end{align}
		For $t=t_n=T$ it holds  $\valuefkt_n(T,\zustand)=1$ and for $t=t_0=0$ 
		\begin{align}
			V_1(0,\zustand)=\vterm_0(\zustand)&=\EbarZ
			\Big[V_{1}\big(0,\Mstate_{0+}^{\Pi,0,m}\big)\Big], \quad  m\in \R^d.
			\label{dynamic_P_P_0}
		\end{align}
		\item
		The  candidate optimal decision rule   is for $t\in [t_{k-1},t_{k}),~k=1,\ldots n$, given by 
		\begin{align}
			\label{opti_stra_power_C_det}
			\Pi^{\HC}=\Pi^{\HC}(t,\zustand)				
			&= \Pi^{F}(t,\zustand)+\frac{1}{1-\theta}\Sigma_{R}^{-1}\Qpro_t^{\HC}\frac{\diffop_{\zustand}\valuefkt_k(t,\zustand)}{\valuefkt_k(t,\zustand)}, 
		\end{align}
		where $\Pi^F$ is given in \eqref{optimal_str_F}.
		
	\end{enumerate}
\end{theorem}
\begin{proof}
	Let   $t,\tau  \in(0,t_1)$  or  $t,\tau \in [t_{k-1},t_{k}),~k=2,\ldots n,$ with   $\tau>t$  for a fixed time point $t$. Analogously to the proof of Theorem \ref{DPE_general}  the dynamic programming principle \eqref{DPP_Y} and the continuity of $\Mstate_s^{H,\Pi,t,\zustand}$ on   $(0,t_1)$ and  $[t_{k-1},t_{k}),~k=2,\ldots n$, imply  that $V_k $ satisfies the PDE \eqref{HJB_Power_general} and that the optimal decision rule is given  as in \eqref{opti_stra_power_C_det}. 			
	
	It remains to prove the terminal conditions in \eqref{dynamic_P_P}   and relation \eqref{dynamic_P_P_0} for the initial time $t=0$. We fix an
	information date $t_k$, $k=1,\ldots,n-1$ and apply again 
	the dynamic programming principle \eqref{DPP_Y} where  we set
	$\tau=t_k$ and  consider the following limit
	\begin{align}
		\vterm_{k}(\zustand)&=\valuefkt_k(t_k-,\zustand)=\valuefkt^{\HC}(t_k-,\zustand)\\
		&=\lim\limits_{t\nearrow t_k}\;\sup\limits_{\Pi\in\mathcal{A}^\HC} \EbarZ  \bigg[\exp\Big\{\int_t^{t_k} 
		b(\Mstate_s^{\HC,\Pi,t,\zustand},\Pi(s,\Mstate_s^{\HC,\Pi,t,\zustand}))ds\Big\}
		\valuefkt^\HC(t_k,\Mstate_{t_k}^{\HC,\Pi,t,\zustand}) \bigg]\nonumber\\
		&=\sup\limits_{\Pi\in\mathcal{A}^\HC} \EbarZ
		\Big[V_{k+1}\big(t_{k},\Mstate_{t_{k}}^{\HC,\Pi,t_{k}-,\zustand}\big)\Big].\nonumber
	\end{align}
	 The above expectation depends only on the distribution of the jump size $\Mstate_{t_k}-\Mstate_{t_k-}$ of the state process which  is independent of the  decision rule $\Pi$.   Thus, we can omit the 		supremum in the last equation and  get
	$\vterm_{k}(\zustand)=\EbarZ
	\Big[V_{k+1}\big(t_{k},\Mstate_{t_{k}}^{\HC,\Pi,t_{k}-,\zustand}\big)\Big].$
	
	For time $t=0$ we have $\mathcal{F}_0^\HC=\mathcal{F}_0^I$ which represents the  prior information on the initial drift $\mu_0$ but does not yet contain the first expert opinion $Z_0$. As above the dynamic programming principle yields for $t=0$ and $ \tau\in(0,t_1)$
	\begin{align}
		\valuefkt^{\HC}(0,\zustand)&=\valuefkt_1(0,\zustand)\\
		&=\lim\limits_{\tau\searrow 0}\;\sup\limits_{\Pi\in\mathcal{A}^\HC} \EbarZ  \bigg[\exp\Big\{\int_0^{\tau}    b(\Mstate_s^{\HC,\Pi,0,\zustand},\Pi(s,\Mstate_s^{\HC,\Pi,0,\zustand}))ds\Big\}
		\valuefkt^\HC(\tau,\Mstate_{\tau}^{\HC,\Pi,0,\zustand}) \bigg]\nonumber\\
		&=\sup\limits_{\Pi\in\mathcal{A}^\HC} \EbarZ
		\Big[V_{1}\big(0+,\Mstate_{0+}^{\HC,\Pi,0,\zustand}\big)\Big] = v_0(\zustand).\nonumber 
	\end{align}						
\end{proof}

\begin{remark}
	\label{rem:jump_V}
	The above theorem shows that at the information dates $t_k, k=1,\ldots,n-1$, the value function exhibits  jumps of size $V^{\HC}(t_k,\zustand)-V^{\HC}(t_k-,\zustand)=V_{k+1}(t_k,\zustand)-\vterm_k(\zustand)$.     Note that we excluded  the information of the first expert opinion $Z_0$ from  the initial $\sigma$-algebra $\mathcal{F}_0^\HC$. Therefore $V$ exhibits  at time $t=0$  a jump of size $V^{\HC}(0+,\zustand)-V^{\HC}(0,\zustand)=V_{1}(0,\zustand)-v_0(m)$. 
\end{remark}

For $H=Z$ the DPE  appears as a system of coupled terminal value problems for the PDE \eqref{HJB_Power_general} for $V_k$ which are tied together by the terminal conditions \eqref{dynamic_P_P}. The latter appear as pasting conditions for the value function described by  $V_k$ and $V_{k+1}$ on two subsequent intervals divided by the information date $t_k$. Therefore  that system can  be solved recursively for $k=n,\ldots,1$ starting with $V_n(t_n,m)=V_n(T,m)=1$. From Sec.~\ref{UtilityMaxDiffusion} it is already known
that the DPEs  for the control problems for the information regimes $H=R,\HD$ can be solved explicitly using an exponential ansatz leading to the  results given in Theorem  \ref{Kandidate_diffusion}.  We apply this idea to our problem and  make for $t\in [t_{k-1},t_{k}),~k=1,\ldots n$, the ansatz 
\begin{align}
	V_k(t,\zustand)=\exp\{\zustand^{\top}A_k(t)\zustand+B_k^{\top}(t)\zustand+C_k(t)\}
	\label{Value_C}
\end{align}
where   $A_k$ is a function on $[t_{k-1},t_k]$ taking values in the set of real  symmetric $\nAktien\times\nAktien$ matrices, whereas $B_k$  and $C_k$ are  some  functions on $[t_{k-1},t_k]$ with values in $\mathbb R^{\nAktien}$ and $\mathbb R$, respectively, which have to determined.

\begin{theorem}[Solution of dynamic programming equation and optimal decision rule]
	\label{backward_recursion_II}~
	\begin{enumerate} 
		\item  The solution to the dynamic programming equation given in \eqref{dynamic_P_P} and \eqref{dynamic_P_P_0}  is for $[t_{k-1},t_{k}),~k=1,\ldots n$ and $m\in \R^d$ given by $$V_k(t,\zustand)=\exp\{\zustand^{\top}A_k(t)\zustand+B_k^{\top}(t)\zustand+C_k(t)\}$$  where the functions $A_k, B_k$  and $C_k$ satisfy  on  $t\in (t_0,t_1)$ and $t\in [t_{k-1},t_{k}),~k=2,\ldots n$,   the system of ODEs given in Theorem \ref{Kandidate_diffusion} for $H=R$
		with terminal values for $t=t_k$ 
		\begin{align}					
			A_{k}(t_k)&=\Update_k\, A_{k+1}(t_k),\label{end_An}\\
			B_{k}(t_k)&=\Update_k\,  B_{k+1}(t_k),\label{end_Bn}\\
			C_{k}(t_k)&=C_{k+1}(t_k) \plusSigmak \frac{1}{2}\, B_{k+1}^{\top}(t_k)\, \Update_k \, \Sigmak B_{k+1}(t_k)
			+\frac{1}{2}\log \rm{det}\, \Update_k ,\label{end_Cn}
		\end{align} 
		where $ \Sigmak$ is the increment of the conditional variance at $t_k$ and given in \eqref{updateformel_Variance_CN_bedingt}. Further
		\begin{align}\label{Lambda_update}
			\Update_k:= \big({I}_{\nAktien}\minusSigmak 2\, A_{k+1}(t_k) \Sigmak \big)^{-1} \quad \text{for }k=1,\ldots,n-1.
		\end{align}
		
		For $k=n$ it holds  $A_{\nExperten}(t_{\nExperten})=0_{\nAktien\times\nAktien}$,
		$B_{\nExperten}(t_{\nExperten})=0_{\nAktien\times1}$, 
		$C_{\nExperten}(t_{\nExperten})=0$.  \\
		For $t_0=0$ the values of $A_1,B_1, C_1$ are obtained from the formulas \eqref{end_An}, \eqref{end_Bn}, \eqref{end_Cn} replacing  $A_{k+1}(t_k)$, $B_{k+1}(t_k)$, $C_{k+1}(t_k)$ by $A_{1}(0+), B_{1}(0+), C_{1}(0+)$, respectively.
		
		\item
		The  candidate optimal decision rule  is for $t\in [t_{k-1},t_{k}),~k=1,\ldots n$, $m\in \R^d$, given by 
		\begin{align}
			\label{opti_stra_power_C_det_II}
			\Pi^{\HC}=\Pi^{\HC}(t,\zustand)&
			= \Pi^{F}(t,\zustand) +  \frac{1}{1-\theta}\Sigma_{R}^{-1}  \Qpro^\HC_t \big(2\,A_k(t)m+B_k(t)\big), 					
		\end{align}
		where $\Pi^F$ is given in \eqref{optimal_str_F}.
		
	\end{enumerate}
\end{theorem}
\begin{proof}
	Plugging ansatz \eqref{Value_C} for $V_k$ into PDE
	\eqref{HJB_Power_general} with $\overline \Sigma^\HC=\Sigma_R$   and equate coefficients in front of $\zustand$ yields the ODEs
	 given in Theorem \ref{Kandidate_diffusion} for $H=R$. 
	The terminal value $V(T,\zustand)=V_n(t_n,\zustand)=1$ implies the given terminal values for $A_n, B_n,C_n$			
	at the terminal time $t_{\nExperten}=T$.
	The other terminal values follow from the evaluation of the expectation  on the right side of
	\eqref{dynamic_P_P}. Using the shorthand notation 
	$	\Ashort=A_{k+1}(t_{k}),  \Bshort=B_{k+1}(t_{k}),  \Cshort=C_{k+1}(t_{k})$ and ansatz \eqref{Value_C} for $V_{k+1}(t_k)$ we obtain
	\begin{align}
		\vterm_{k}(\zustand)&= \EbarZ \Big[  V_{k+1}\big(t_{k},\Mstate_{t_{k}}^{\Pi,t_{k}-,\zustand}\big) \Big]
		= \EbarZ \Big[ \exp\Big\{ \big(\Mstate_{t_{k}}^{\Pi,t_{k}-,\zustand}\big)^{\top} \Ashort\, \Mstate_{t_{k}}^{\Pi,t_{k}-,\zustand} +\Bshort^{\top} \Mstate_{t_{k}}^{\Pi,t_{k}-,\zustand} +\Cshort \Big\} \Big].
		\nonumber
	\end{align}
	Completing the square with respect to $\Mstate_{t_{k}}^{\Pi,t_{k}-,\zustand}$ 
	yields
	\begin{align}
		v_{k}(\zustand)&=\EbarZ \Big[ \exp \Big\{ \big( \Mstate_{t_{k}}^{\Pi,t_{k}-,\zustand}+\frac{1}{2}\Ashort^{-1} \Bshort\big)^{\top} \Ashort\,
		\big( \Mstate_{t_{k}}^{\Pi,t_{k}-,\zustand}+\frac{1}{2}\Ashort^{-1} \Bshort\big) \Big\}
		\cdot\exp\Big\{-\frac{1}{4}\Bshort^\top \Ashort^{-1} \Bshort+\Cshort \Big\}  \Big]. \nonumber
	\end{align}
	The above expectation can be  evaluated using \cite[Lemma 3.4]{Gabih et al (2022) Nirvana} saying that for 
	a $\nAktien$-dimensional Gaussian random vector $Y\sim\mathcal N(\mu_Y,\Sigma_Y)$, $b\in \mathbb R^\nAktien$ and   a symmetric and invertible matrix   ${U}\in \mathbb R^{d\times d}$  such that  all eigenvalues of  ${I}_{\nAktien}-2{U}\Sigma_Y$  are positive,  it holds
	\begin{align}
		\nonumber
		\EbarZ \big[\exp\{(Y+b)^{\top} {U}(Y+b)\}\big]=&\big(\rm{det}({I}_{\nAktien}-2\,{U} \Sigma_Y)\big)^{-1/2} \times\\
		&\exp\big\{(\mu_Y+b)^{\top}\;   ({I}_{\nAktien}-2\,{U}\Sigma_Y)^{-1} {U} \,(\mu_Y+b) \big\}.
		\nonumber
	\end{align}
	We set  $U=\Ashort$ and $b=\frac{1}{2}\Ashort^{-1} \Bshort $
	and $Y=\Mstate_{t_{k}}^{\Pi,t_{k}-,\zustand}\sim\mathcal N (\mu_Y,\Sigma_Y)$. For mean and covariance of $Y$ update formula \eqref{updateformel_Filter_CN_bedingt} and relation \eqref{updateformel_Variance_CN_bedingt} imply $\mu_Y=\zustand$ and $\Sigma_Y=   
	\Qpro_{t_k-}^{\HC}(\varianceexp+\Qpro_{t_k-}^{\HC})^{-1}  \Qpro_{t_k-}^{\HC} =- \Sigmak$. 
	We obtain 
	\begin{align}
		\vterm_{k}(\zustand)&=\exp\Big\{\big(\zustand+ \frac{1}{2}\Ashort^{-1} \Bshort\big)^{\top}\;  \Update_k \Ashort \big(\zustand+\frac{1}{2}\Ashort^{-1} \Bshort\big) - \frac{1}{4}\Bshort^\top \Ashort^{-1}\Bshort+\Cshort +\frac{1}{2}\log \det \Update_k \Big\}.
		\nonumber
	\end{align}
	and rearranging terms yields
	\begin{align}
		\vterm_{k}(\zustand)&=\exp\Big\{\zustand^{\top}\Update_k \Ashort\, \zustand~+\Bshort^{\top}\Update_k\, \zustand ~+\Cshort \plusSigmak \frac{1}{2} \Bshort^{\top} \Update_k\, \Sigmak \Bshort +\frac{1}{2}\log \rm{det} \Update_k\Big\},
		\nonumber
	\end{align}
	On the other hand we have from ansatz \eqref{Value_C}
	$
	\vterm_{k}(\zustand)= V_{k}(t_k,\zustand)
	=\exp\big\{\zustand^{\top}A_k(t_k) \zustand+B_k^{\top}(t_k)\zustand+C_k(t_k) \big\}.
	$ 
	By comparing the coefficients  
	in front of $\zustand$ in the last two expressions for $\vterm_k$ and substituting back the expressions for $\Ashort, \Bshort, \Cshort$ we obtain the desired result.
	
	The proof for $t=t_0=0$ is analogous.			
	Finally, the expression for the optimal decision rule follows if ansatz \eqref{Value_C} for the value function $V^\HC$ is substituted into \eqref{opti_stra_power_C_det}.			
\end{proof} 
\begin{remark} \label{rem:update_unreliable_experts}
	Analysing the update formulas \eqref{end_An} through 
	\eqref{end_Cn} for expert opinions which become less and less reliable in the sense that $\norm{\varianceexp^{-1}}\to 0$ 
	it can be seen that 	$ \Sigmak$ tend to $0_d$ and   the update-factors $\Update_k$ given in \eqref{Lambda_update} tend to  $I_d$. 
	As a consequence the functions $A_k$, $B_k$ and
	$C_k$  define a smooth value function $V^\HC$ on $[0,T]$ which equals the value function $V^\HR$ of the $\HR$-investor. This is as expected since  in the considered limiting case for the $\HC$-investors expert opinions do not provide any additional information about the drift.
	
	 In Remark \ref{rem_log_power} we have seen that for $H=\HR,\HD$ and $\theta\to 0$ the optimal decision rule $\Pi^{H}(t,m)$ converges to   $\Pi^{H}_{\log}(t,m)=\Sigma_{R}^{-1}m$ which is optimal for log-utility, see Proposition \ref{LogUtilityValue}.  The same holds for $H=\HC$ since the functions $A^k,B^k$ and therefore the correction term $\Pi^Z-\Pi^F$ vanish  for $\theta=0$. This can be seen from the vanishing solutions of the ODEs for $A_k,B_k$ between the information dates. This implies that the updates \eqref{end_An},\eqref{end_Bn} at the informations dates also vanish.
\end{remark}	

	\paragraph{Boundedness of $\mathbf{A_k,B_k,C_k}$} The proof of this property is based on the relations $A_k(t)=\overline{A}^{\HF}\big(t,\Qpro_t^{\HC}\big)$ and $B_k(t)=\overline{B}^{\HF}\big(t,\Qpro_t^{\HC}\big)$ on $[t_{k-1},t_k], ~k=1,\ldots,n$.  They can be verified by some calculations showing  that between the information dates the left and right hand sides satisfy the same ODEs  as in the  corresponding proof for $A^H,B^H$ for $H=R,\HD$ in 	\cite[Lemma 5.2.1]{Kondkaji (2019)}. There are also the same terminal conditions at the information dates which  are given by the update formulas \eqref{end_An} and \eqref{end_An}. In view of the definition of $\overline{A}^F,\overline{B}^F$ in Lemma \ref{value_F_average}, the boundedness of $A_k,B_k$ follows from  condition \eqref{A_bounded_full} saying that $\Abound$ is bounded. Then also $\Bbound$ and $\Cbound$ are bounded.
	
	Further,   the eigenvalues of $I_{\nAktien}-2A^{\HC}(t)\Qpro_t^\HC$ have to be positive  on $[0,T]$. As already explained for $H=\HR,\HD$ this always holds true for $\theta<0$ while for $\theta\in(0,1)$ one has to require $\lambda_{\max}(\Abound(t)\Qpro^\HC_t)<1/2$. From Prop.~\ref{properties_filter} it is known that  $\Qpro^{\HC}_t \preceq \Qpro^R_t$ from which one can deduce $\lambda_{\max}(\Abound(t)\Qpro^\HC_t)\le \lambda_{\max}(\Abound(t)\Qpro^\HR_t)<1/2$. Thus, it is enough to check the condition for $H=\HR$. 
	
	Finally,    The boundedness of $C_k$ follows from integrating the r.h.s.~of the ODE for $C_k$ which is bounded and the updates  given in  \eqref{end_An} which are also bounded.
	
	\smallskip
	\paragraph{Verification}  The key idea for the verification is that between the information dates the functions $V_k$ satisfy the same DPE as the value function $V^R$ of the $R$-investor given in   \eqref{HJB_Power_general}. This allows to rely on the verification results for $H=R$ and to iterate them backward in time. Starting point is a control problem on $[t_{n-1},T]$ with the modified initial time $t=t_{n-1}$ instead of $t=0$ and initial value $m_0=\Mstate^Z_{t_{n-1}}$ and $q_0=\Qpro^Z_{t_{tn-1}}$. This is a control problem for the $R$-investor on $[t_{n-1},T]$ for which $V_n(t,m)$  and $\Pi^Z(t,m)$ for all $(t,m)$ are verfified to be  the value function and  the optimal decision rule, respectively. Here,  $V_n$ satisfies the terminal condition $V(T,m)=1$, as for $H=R$.
	
	Next, we consider the control problem on  $[t_{n-2},t_{n-1}]$ with initial time $t=t_{n-2}$. At the terminal time $t_{n-1}$ the terminal condition is obtained from the dynamic programming principle \eqref{DPP_Y} leading to the expression for $V_{n-1}(t_{n-1},m)$ given in  \eqref{dynamic_P_P}. Again we can apply the verification results for $H=R$ on this time interval. Note, that these results also work for nonzero terminal conditions. They only require the boundedness of the solutions of the ODEs for $A_k,B_k,C_k$ which already checked above. Continuing this iteration completes the verification.

\paragraph{Monetary value of information} \label{Eff_partiell_Cn}
For the initial investments $x_0^{\HR/\HC}$ and
$x_0^{\HC/\HF}$ introduced in Subsection \ref{MonetaryValue} and
needed to evaluate the efficiency of the $\HC$-investor
and the  monetary value of expert opinions one can derive the following  expressions.

\begin{lemma}[{Kondkaji (2019)}, Lemma 6.3.1]
	\label{monetary_value_Z}	
	For the initial capital $x_0^{Z/\HF}$  and the efficiency
	$\varepsilon^{Z}$   defined
	in  \eqref{monetaer_Wert} and \eqref{effekt_formel}, respectively,   it holds 
	$$
	x_0^{Z/\HF}=x_0^Z\exp\{-\Delta_X^Z(0)\}\quad  \text{and}\quad
	\varepsilon^{Z}=\exp\{-\Delta_X^Z(0)\}.
	$$
	For the initial capital $x_0^{\HR/\HC}$  defined in \eqref{initial_cap_H} and the monetary value of expert opinions $P_{Exp}^\HC$ it holds
	$$
	P_{Exp}^\HC=x_0^R-x_0^{\HR/\HC}~
	\text{with}~x_0^{\HR/\HC}=\exp\{-\Delta_X^{\HC}(0)+\Delta_X^{\HR}(0)\},
	$$
	where $\Delta_X^{\HC}(0)$ and $\Delta_X^{\HR}(0)$ are given in
	{\eqref{Delta_X_RD}} where for $H=\HC$ the settings $\Qpro^H=\Qpro^Z, \overline \Sigma^\HC=\Sigma_R$ and $\overline{K}^{\HC}(t)=0$ apply.
\end{lemma}

\section{Numerical Results}
\label{numeric_result} In this section we illustrate the theoretical
findings of the previous sections  with results of some numerical
experiments.
\subsection{Model Parameters}
\label{model_parameters} Our experiments are based on a stock
market model where the unobservable drift $(\drift_t)_{t\in[0,T]}$
follows an Ornstein-Uhlenbeck process as given in \eqref{drift} whereas the volatility is known and constant.
For simplicity, we assume that there is only one risky asset in the
market, i.e. $\nAktien=1$.

If not stated otherwise our numerical experiments are based on model parameters as given 
in Table \ref{parameter}.   The parameter  of the utility function is chosen to be negative,  $\theta=-0.3$. This corresponds to an investor which is more risk averse than a log-utility maximizing investor trading the Kelly portfolio. This is the relevant case for most investors. Note that for a negative $\theta$  the optimization problem is always  well-posed, see  our  paper \cite{Gabih et al (2022) Nirvana}. 
\begin{table}[ht]
	\footnotesize
	\begin{tabular}{|l@{\hspace*{+0.2em}}ll|r||@{\hspace*{-0.5em}}lll|r|}
		\hline
		\rule{0mm}{2.5ex}%
		Drift & mean reversion level& $\revlevel$ & $0.05$      &   & Investment horizon & $T$ & $1~$ year
		\\ \hline \rule{0mm}{2.5ex}%
		& mean reversion speed& $\revspeed$ & $3$ &   & Power utility parameter  & $\theta$ & $-0.3$
		\\ \hline
		& volatility  &$\sigma_{\mu}$ \rule[-1.0ex]{0pt}{3.0ex} &$1$ &    &  Volatility of stock & $\volR$ & $0.25$
		\\\hline \rule{0mm}{2.5ex}%
		& mean of $\drift_0$  &   $\driftinitial$  & $\revlevel=0.05$&     &  Volatility of cont.-time experts & $\volexp$ & $0.2$
		\\ \hline
		&variance of $\drift_0$  &   $\covinitial$  & \rule[-1.3ex]{0pt}{4.5ex}$\frac{\sigma_{\drift}^2}{2\revspeed}=0.1\overline{6}$ &     & Variance of discr.-time experts &$\varianceexp$ &$0.4$
		\\ \hline
		Filter &initial value $\Mstate^H_0$ & $\filterinitial$ \rule[-1.1ex]{0pt}{3.5ex} & $\driftinitial=0$  &     &   Number of expert opinions &$\nExperten$ &$10$
		\\ \hline
		&initial value  $\Qpro_0^H$ & $ \condcovinitial$ \rule[-1.1ex]{0pt}{3.5ex}& $\covinitial=0.1\overline{6}$ &     & &    &     \\ \hline
	\end{tabular}
	\\[1ex]
	\centering \caption{\label{parameter}
		Model parameters for the numerical experiments
	}
\end{table}

The mean and variance of the initial value $\drift_0$ of the drift are the parameters of the stationary distribution of the drift process which is known to be Gaussian with mean $\revlevel$ and variance  ${\sigma_{\drift}^2}/{(2\revspeed)}$. Hence the drift process $\drift$ is stationary on $[0,\infty)$ and for the chosen parameters there is a $90\%$ probability for values in the interval   $(-0.75,0.85)$ centered at $\revlevel=0.05$, and a probability around 2/3 for values in $(0.35,0.45)$. The filter processes $\Mpro^H$ and $\Qpro^H$ are also initialized with the stationary mean and variance modeling partially informed investors which only know the model parameters but have no additional prior knowledge about the drift. The volatility $\volexp$ of the continuous-time expert opinions $\contexp$ is chosen as $0.2$ and  slightly smaller than the volatility  $\volR=0.25$ of the return process $R$. Hence the observations of $\contexp$ are more informative than those of the returns $R$.

\subsection{Filter}
\label{Num_Filter}
In this subsection we want to illustrate our theoretical findings on filtering based on different information regimes by results on a simulation study.	
Figure~\ref{fig_num_filter} shows in the top panel simulated paths of the two diffusion type observation processes which are the return process $R$ and the continuous-time expert opinion process $J$  associated to the drift process $\mu$.  The drift process which is not observable by the investors is plotted in the middle panel.  The top panel also presents for comparison the cumulated drift process representing the drift component in the dynamics of both $R$ and $J$, see \eqref{ReturnPro} and \eqref{continuous-expert}. Expert's views $Z_k$ which are noisy observations of the drift process $\drift$ at the information dates and forming the additional information of the $\HC$-investor  are shown as red crosses in the middle panel. From the observed quantities  the filter of the $R$-, $\HD$- and $\HC$-investor are computed in terms of the pair $(\Mpro^H,\Qpro^H)$. For $H=R,\HC,\HD$ the conditional expectations $\Mpro^H$ are plotted against time  in the middle panel while the conditional variances $\Qpro^H$ are shown in the bottom panel. 

\begin{figure}[th]
	\hspace*{-0.05\textwidth}
	\includegraphics[width=1\textwidth]{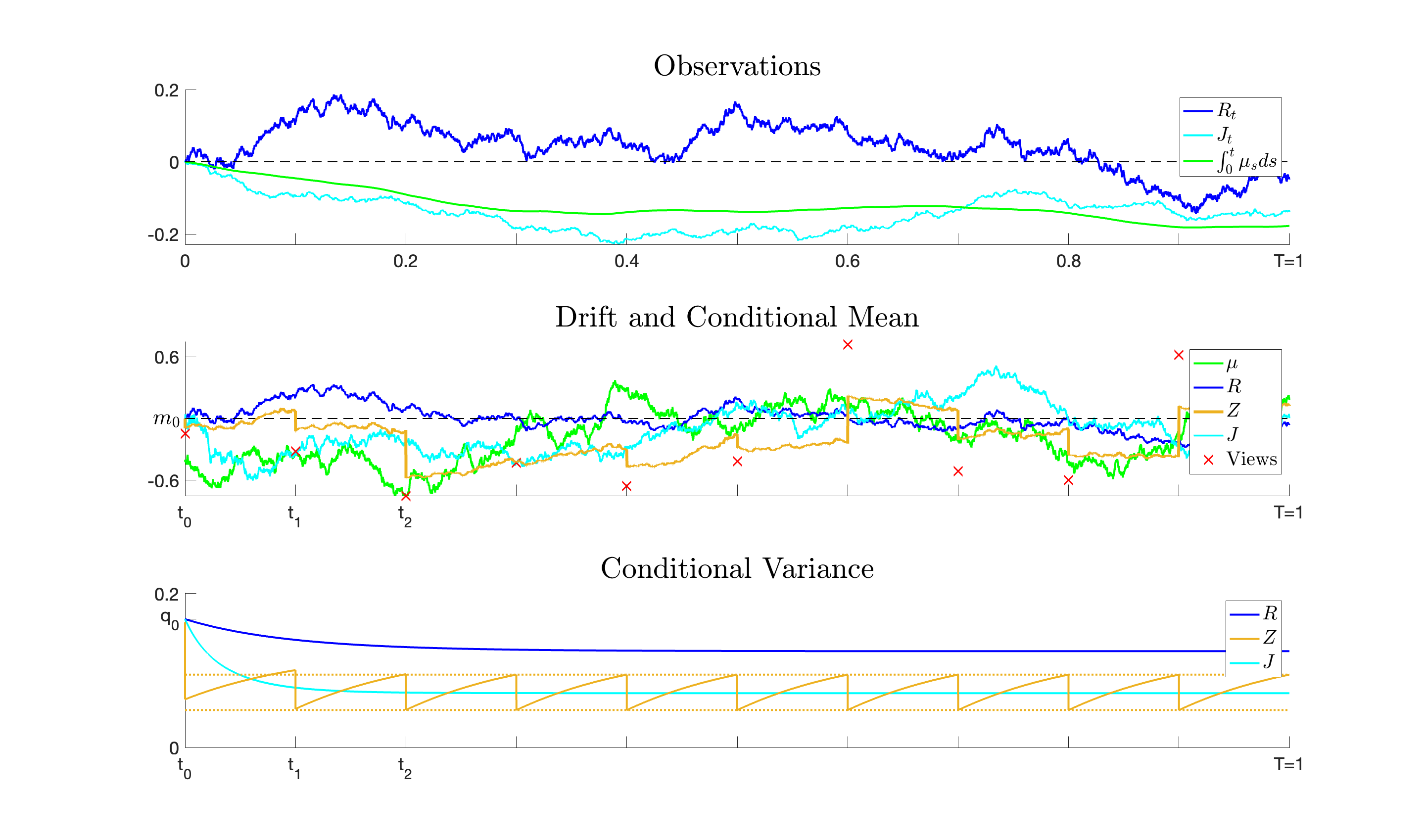} 
	\centering \caption{\label{fig_num_filter}
		Observation and filter processes. ~~ \newline 
		\begin{tabular}[t]{ll}
			Top: & Diffusion-type observation processes $\R$ and $\HD$. \\
			Middle: & Drift $\drift$, expert views $Z_k$ and conditional mean $\Mpro^H$ for $H=R,\HC,\HD$. \\
			Bottom: &Conditional variance $\Qpro^H$. 
		\end{tabular}		
	}
\end{figure}

Recall that $\Qpro^R$ and $\Qpro^\HD$ as well as $\Qpro^\HC$ for any $n\in\N$ are deterministic. In the bottom panel one sees that for any fixed $t\in[0,T]$, the value of $\Qpro^\HD_t$ as well as the value of $\Qpro^\HC_t$ is less or equal than the value of $\Qpro^R_t$. This shows that additional information by the expert opinions improves the accuracy of the filter estimate. It confirms the underlying theoretical result on the partial ordering of the conditional covariance matrices as stated in Proposition \ref{properties_filter}. 
The updates in the filter of the $\HC$-investor at the information dates of the expert's views    decrease
the conditional variance and lead to a jump of the conditional mean $\Mpro^{\HC}$.
These are typically  jumps towards  the hidden drift $\mu$,  of course this
depends on the actual value of the expert's view. Note that the
drift estimate $\Mpro^R$ of the $R$-investor is quite poor and
mostly fluctuates just around the mean-reversion level $\revlevel$.
However, the expert opinions  visibly improve the drift estimate.

For increasing $t$ the conditional variances $\Qpro_t^R$ and $\Qpro_t^\HD$ approach a finite value. An associated convergence result for $t\to \infty$ has been proven in Proposition~4.6 of Gabih et al.~\cite{Gabih et al (2014)} for markets with $d=1$ stock and generalized in Theorem~4.1 of Sass et al.~\cite{Sass et al (2017)} for markets with an arbitrary number of stocks. For the $\HC$-investor  we observe an almost periodic behavior of the conditional variance $(\Qpro^{\HC}_t)_{t\geq 0}$. The asymptotic behavior  for $t\to \infty$ and the derivation of  upper and lower bounds have been studied in detail in Gabih et al.~\cite[Prop.~4.6]{Gabih et al (2014)} for $d=1$ and in  Sass et al.~\cite[Sec.~4.2]{Sass et al (2017)} for the general case.  These bounds are shown as dashed lines in the bottom panel.

\medskip
\begin{figure}[h!]
	\hspace*{-0.05\textwidth}
	\includegraphics[width=0.56\textwidth,height=0.5\textwidth]{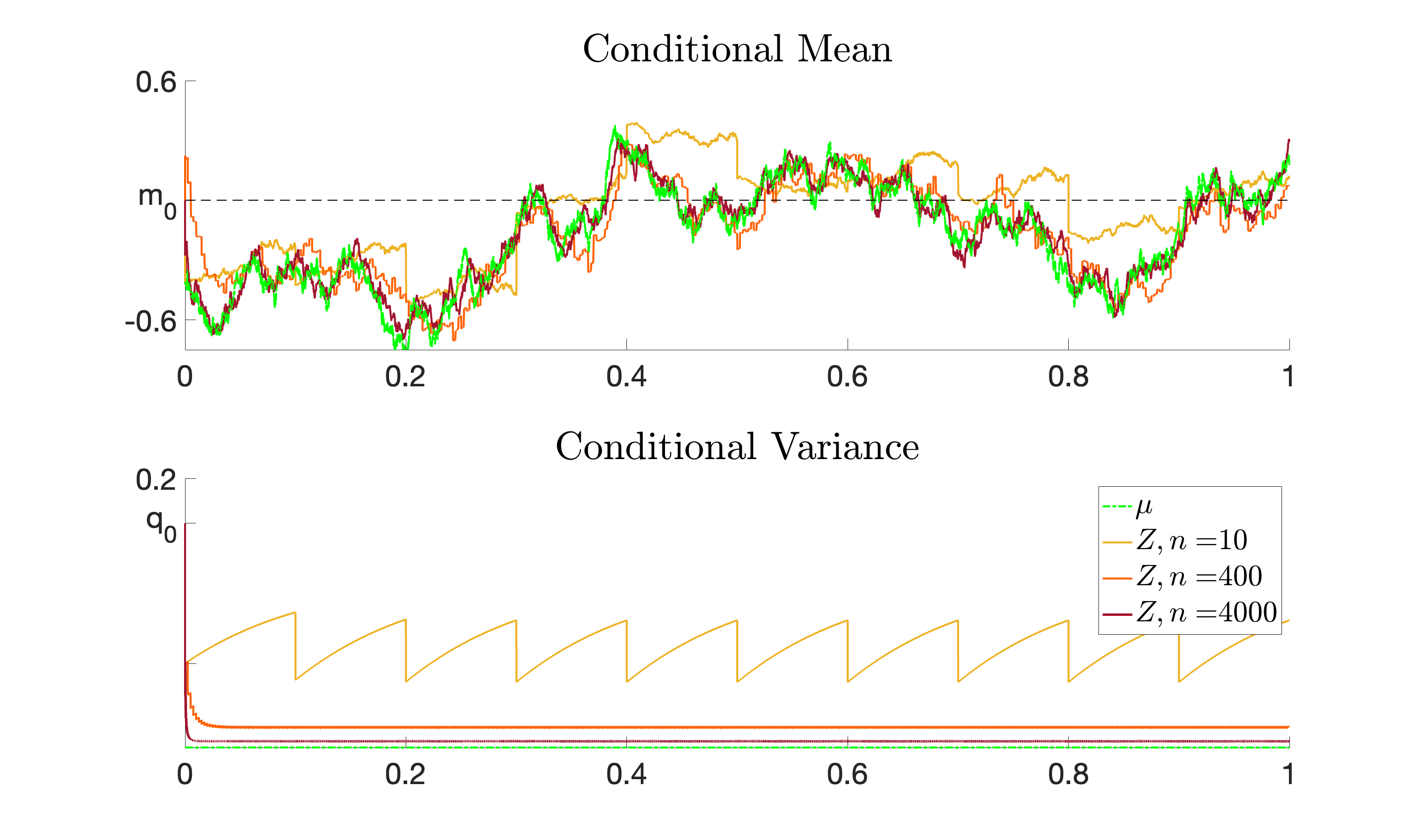} %
	\hspace*{-0.06\textwidth}
	\includegraphics[width=0.56\textwidth,height=0.5\textwidth]{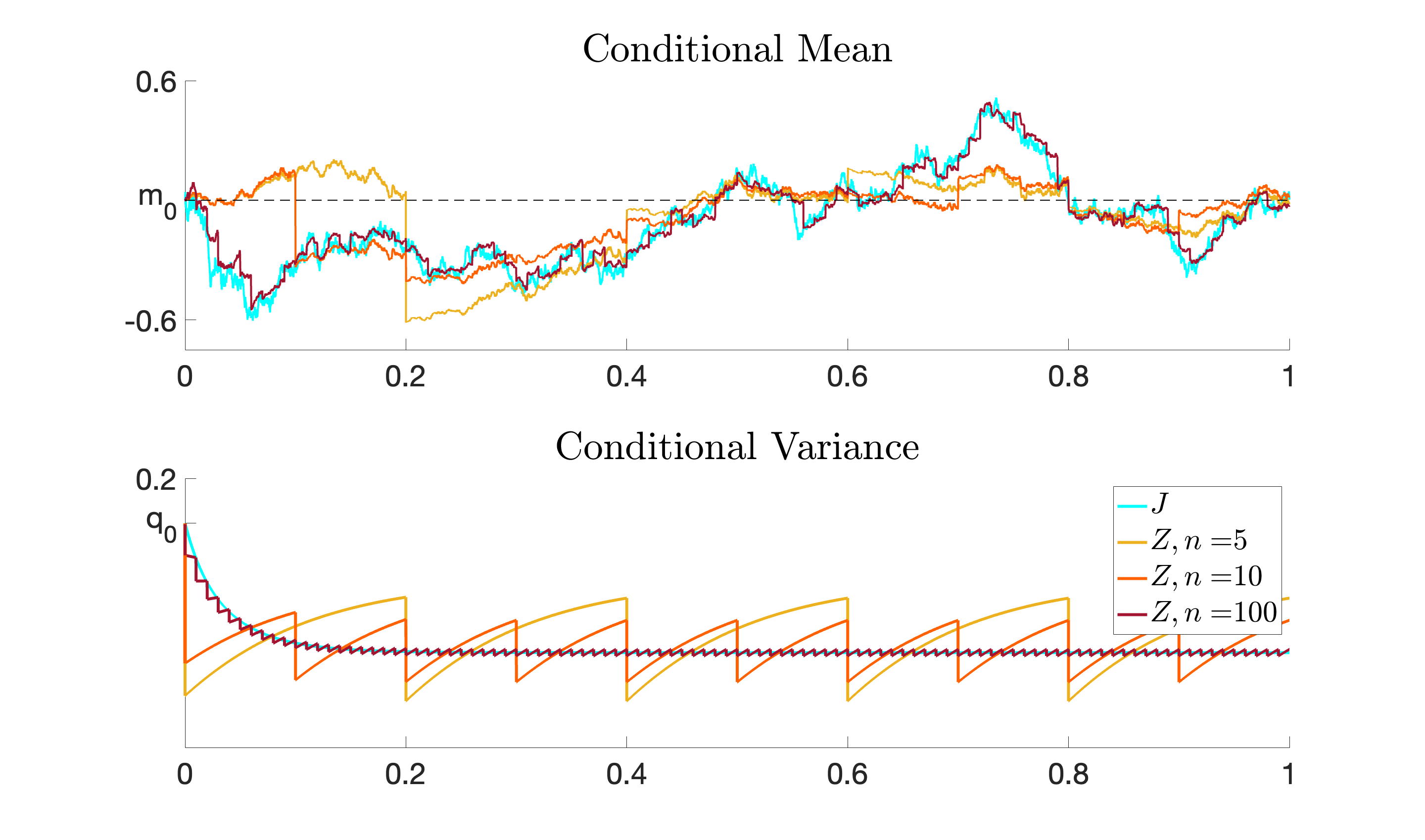} 
	
	\centering \caption{\label{fig_num_filter_to_F}
		Asymptotic behavior of the filter for $n\to \infty$  \newline
		Left: ~ Constant expert's variance $\varianceexp=0.4$, convergence to full information    \newline  
		Right: Linearly growing expert's variance $ \varianceexp^{(n)}=\frac{n}{T}\volexp^2$, diffusion approximation   \newline  
	}
\end{figure}
Next we perform some experiments illustrating the theoretical results from Subsec.~\ref{propertiesfiltersection} on the  asymptotic filter behavior for increasing number $n$ of expert opinions. We distinguish two cases. First, the expert's variance $\varianceexp$ stays constant leading to convergence to full information, i.e.,  mean square convergence of $\Mpro^{Z,n}$ to  $\mu$ and $\Qpro^{Z,n}\to 0$ on $(0,T]$, see Theorem \ref{theorem_Konv_Fn}. Second, that variance grows linearly with $n$ leading to convergence to the filter processes of the $Z$-invester to those of the $J$-investor, see Theorem \ref{theorem_conv_diffusion_appr}. 
For that experiment the expert's views are generated as in \eqref{eq:diff_appr_Z}, i.e., the Gaussian random variables $\varepsilon_k$ in \eqref{Expertenmeinungen_fest}  are linked	with the Brownian motion $W^{\HD}$ from \eqref{continuous-expert} driving the continuous-time expert opinion process $J$.

Fig.~\ref{fig_num_filter_to_F} shows on the left side for the experiment with constant variance $\varianceexp$ the conditional mean $\Mpro^{Z,n}$ and the drift $\mu$ (top) and the conditional variance $\Qpro^{Z,n}$ (bottom) for $n=10,400,4000$. It can be nicely seen that for increasing $n$ the conditional variance tends to zero while the conditional mean approaches the  drift process  for any $t\in(0,T]$. In the limit for $n\to \infty$ the $Z$-investor has full information about the drift process.
The panels on the right side show the results for the experiment with linearly growing variance for which we consider the cases $n=5,10,100$.
It can be seen that the both filter processes $\Mpro^{Z,n}$ and  $\Qpro^{Z,n}$ approach the corresponding processes of the $J$-investor for any $t\in[0,T]$ in accordance with Theorem \ref{theorem_conv_diffusion_appr}. Contrary to the first experiment we observe that this convergence is much faster.

Note that for the chosen parameters from Table \ref{parameter} we have for $n=10$ expert opinions that $\varianceexp=\varianceexp^{(n)}=\frac{n}{T}\volexp^2=0.4$, i.e., the same expert's variances for both experiments. This yields for $n=10$ identical conditional variances as it can be seen  in the two bottom panels. However, the paths of the conditional mean  are different since the expert's views $Z_k$ in the experiment with linearly growing expert's variance  are linked to the Brownian motion $W^{\HD}$, see \eqref{eq:diff_appr_Z}, whereas in the left panels they are not.

\subsection{Value function}
\label{Power_numeric_Kapitel}
In this subsection we present for the case of power utility solutions to the control problem \eqref{Wertfkt_H}. In particular we analyze the value functions $V^H(t,m)$ and the associated optimal decision rules $\Pi^H(t,m)$  for the different information regimes $H$. They are given in Theorems  \ref{kandidate_full},  \ref{Kandidate_diffusion} and  \ref{backward_recursion_II} for $H=F$, $R,J$ and $Z$, respectively. We recall relation  \eqref{Value_Orig_RiskSens} saying that  the solution of the original problem of maximizing expected power  utility  can be obtained from $V^H$ by the relation 
$\mathcal{V}_0^H(x_0)=\frac{x_0^\theta}{\theta} \valuefkt^{H}(0,m_0)$. For $H=F$ the relation also holds true if the initial value $m_0$ of the conditional mean is replaced by the initial value $\drift_0$ of the drift.

\begin{figure}[h]
	\hspace*{-0.1\textwidth}
	\includegraphics[width=1.1\textwidth,height=0.6\textwidth]{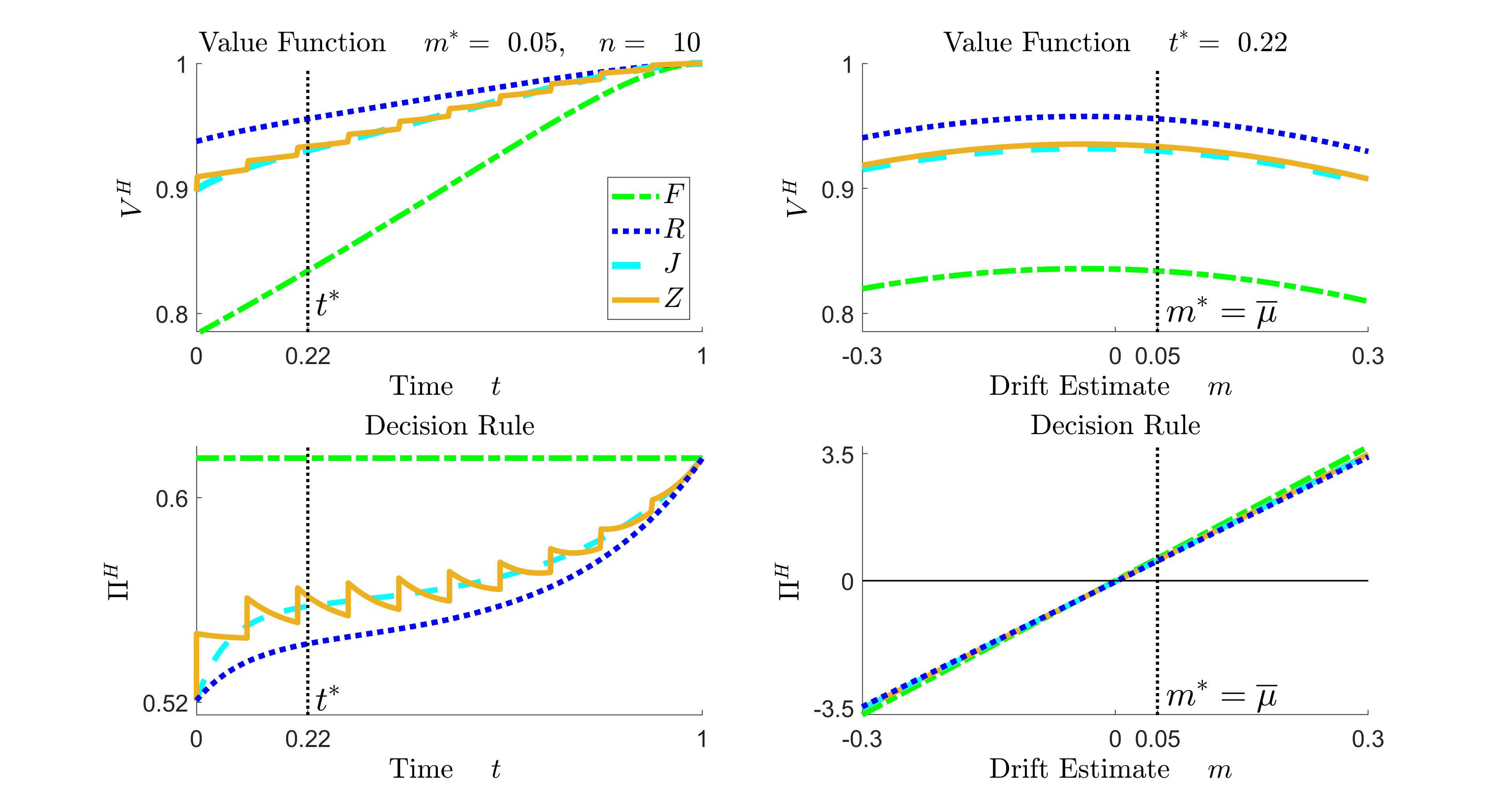}
	
	\centering \caption{\label{value_bild}				
		Value functions and optimal decision rules.
		\newline
		Top: ~~~~ Value functions  $\valuefkt^{H}(t,\zustand)$ for $H=F,\HR,\HD,\HC~ (\nExperten=10)$ depending on $t/\zustand$ (left/right).
		\newline
		Bottom: Optimal decision rule ${\Pi}^H(t,m)$ depending on $t/\zustand$ (left/right).
	}
\end{figure}

Figure \ref{value_bild} shows in the upper part the value functions $V^H(t,m)$ for $H=R,Z,J,F$ plotted against time $t$ (left) for  a fixed value for the  drift estimate  $m= m^{\ast}=0.05$ and plotted against  $m$ (right) for fixed time $t=t^\ast=0.22$. The lower panels show the corresponding decision rules $\Pi^H(t,m)$ for the partially informed investors.

For the value functions we can observe that they  increase with  time $t$ and reach the value $1$ at terminal time $T$  which follows from the definition of the performance criterion in \eqref{Zielfkt_H}. The value function of the $Z$-investor exhibits jumps at the information dates. The upper right panel illustrates that the value functions are exponentials  of a quadratic function of the  drift estimate $m$.
For almost all $(t,m)$ the value function of the fully informed investor  $V^F(t,m)$  is smaller than  those of the partially observed investors.  We recall relation \eqref{Value_Orig_RiskSens} saying that 	$\mathcal{V}_0^H(x_0)=\theta^{-1}x_0^\theta \valuefkt^{H}(0,m_0)$ and that  we work with a negative parameter  of the utility function ($\theta=-0.3$). Hence, order relations for the maximized expected utilities $\E\left[\utility_{\theta}(X_T^{\pi})\mid \mathcal{F}^H_0\right]$ represented by  $\mathcal{V}_0^H$ are reversed to relations for value functions $V^H$. Further,  the value functions of the $Z$- and $J$-investor with access to additional information from the expert opinions   are smaller than the value function of  $R$-investor observing returns only.  We note,  that these relations do not hold in general, except for $t=0$. 

The lower plots show the optimal decision rules $\Pi^H(t,m)$ which are given in \eqref{optimal_str_F}, \eqref{opti_stra_power} and \eqref{opti_stra_power_C_det}. They are all of the form,  see  Remark \ref{rem:myopic} and \eqref{opti_stra_general},
$$ \Pi^{H}(t,\zustand)=\Pi^{\HF}(t,\zustand)+
\frac{1}{1-\theta}\Sigma_{R}^{-1}\Qpro_t^{H}\big(2\,A^H(t)m+B^H(t)\big) ~~ \text{with} ~~
\Pi^{F}(m)=\frac{1}{1-\theta}\Sigma_{\HR}^{-1}\zustand.$$
Here,  $\Pi^F$ constitutes the   myopic decision rule  whereas the correction term $\Pi^H - \Pi^F$ describes  the   drift risk of the partially informed $H$-investor.
All decision rules  are linear in the  drift estimate $m$ as it can be seen from the lower right panel. 
For  drift estimates $m$   much larger (smaller) than the mean reversion level of the drift $\revlevel$ the investors holds a long (short) position in the stock which are smaller (in absolute terms) for the fully informed investor than for partially informed ones.
The figures further show that the drift risk   decreases over time (in absolute terms) and vanishes at terminal time $T$. For $H=Z,J$  it  is smaller than for $H=R$ indicating that more information about the hidden drift leads to  decision rules closer to the  myopic decision rule. This effect is also supported by the observation that the decision rule of the $Z$-investor exhibits jumps at the information dates towards the  myopic decision rule. The arrival of an additional information improves the filter estimate of the hidden drift and decreases the correction term.

We refer to Kondakji
\cite[Sec.~8.3]{Kondkaji (2019)} for results for a positive parameter $\theta$, i.e., a relative risk aversion $1-\theta$ larger than for log-utility. Contrary to   the present case with $\theta=-0.3$ in which  the control problem  \eqref{Wertfkt_H} is a  minimization problem we face a maximization problem for   $\theta>0 $. There are similar results but the monotonicity properties w.r.t.~time $t$ and the ordering of the value function and optimal decision rules for the different information regimes $H$ are reversed.

\subsection{High-Frequency Experts}
\label{Numerik_Konvergenz_Wertfunktion}		
In this subsection we want to study the asymptotic behavior of the value functions and optimal decision rules of the $Z$-investor for growing number $n$ of expert opinions. For the case of log-utility it is known 
that the convergence results for the filters as given in  Theorems \ref{theorem_Konv_Fn} and \ref{theorem_conv_diffusion_appr}  carry over directly to the convergence of value functions. The proof is straightforward and relies on representations of the value function as in \eqref{Log_Utility_Value} in terms of an integral functional of the (deterministic) conditional variance $\Qpro^Z$.

However, in case of power utility that approach can no longer be adopted since
the performance criterion in \eqref{Zielfkt_H}  and consequently the value functions $V^H(t,m)$ are now  given in terms of  expectations of the exponential of  quite involved integral functionals of the  filter processes $\Mstate^H$ 
	under the equivalent measure $\overline{\P}^H$ introduced in Subsec.~\ref{PowerUtility}. Hence, the value functions $V^H(t,m)$ depend on the complete filter distribution, not only on its second-order moments. Further, for power utility the optimal strategies do not depend only on the current drift estimate but contain correction terms depending on the distribution of future drift estimates.
A formal and rigorous proof of convergence of value functions is ongoing work and deferred to a forthcoming publication. It is based  on the $\mathrm{L}^p$-convergence of conditional mean processes for arbitrary $p\ge 2$ as  it can be deduced from  Theorems \ref{theorem_Konv_Fn} and \ref{theorem_conv_diffusion_appr}.

Our numerical results presented below provide a strong support of the convergence of value functions also for power utility. 
As in Subsec.~\ref{propertiesfiltersection} we consider two different asymptotic regimes which are obtained if the expert's variance $\varianceexp$ is  either fixed or grows linear in $n$. In order to emphasize the dependence of the value function and optimal decision rule of the $Z$-investor on $n$ we use the notation ${\valuefkt}^{\HC,n}$  and  ${\Pi}^{\HC,n}$.
\begin{figure}[h]
	\hspace*{-0.1\textwidth}
	\includegraphics[width=1.15\textwidth]{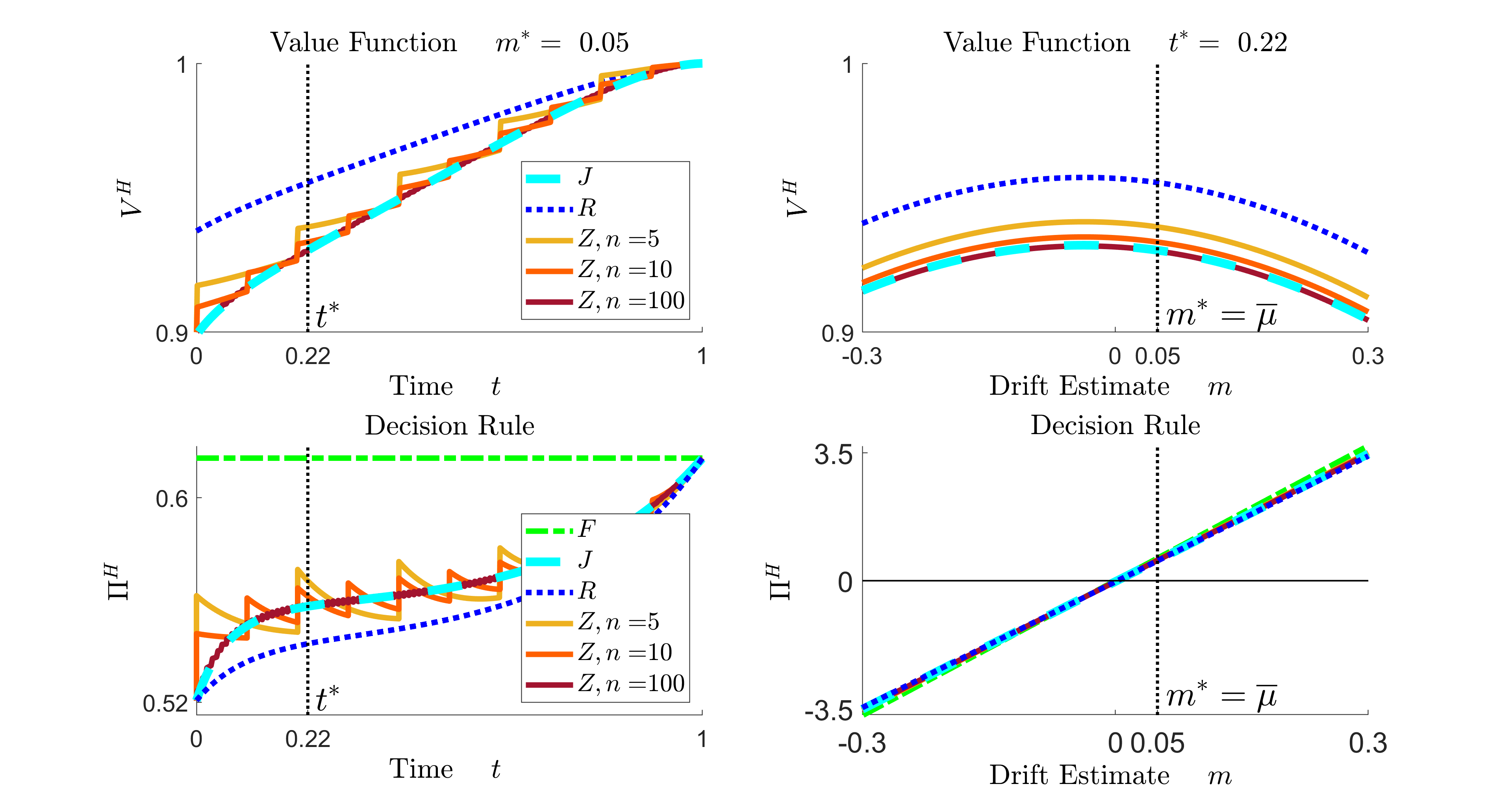}%
	
	\centering \caption{\label{positive_Konvergenz_CnD}
		Asymptotic behavior of value functions and optimal decisions rules for growing $\nExperten$  and $~ \varianceexp^{(n)}=\frac{\nExperten}{T}\volexp^2$.
		\newline
		Top:~~~~~ Value functions $\valuefkt^{H}(t,\zustand)$ depending on $t/\zustand$ (left/right) for $H=\HR,\HD,\HC$
		\newline
		Bottom: Optimal decision rules ${\Pi}^H(t,m)$ for $H=F,R,J,Z$ depending on $t/\zustand$ (left/right).
	}
\end{figure}

Figure \ref{positive_Konvergenz_CnD} presents results of experiments for linearly growing variance $ \varianceexp^{(n)}=\frac{\nExperten}{T}\volexp^2$  for which we have convergence of the filter processes $\Mpro^{Z,n}, \Qpro^{Z,n}$ to the diffusion limit given by filter processes of the $J$-investor observing a continuous-time expert opinion process. The top panels show  the value function ${\valuefkt}^{\HC,n}(t,\zustand)$  while the bottom panels present  the optimal decision rule ${\Pi}^{\HC,n}(t,\zustand)$ of the $Z$-investor observing  $\nExperten=5,10,100$ expert  opinions. For comparison we also show the results for the $R$- and  $J$-investor. For increasing $n$ both ${\valuefkt}^{\HC,n}$ and  ${\Pi}^{\HC,n}$ quickly approach the corresponding quantities of the $J$-investor. This shows that for the chosen parameters quite accurate diffusion approximations of solutions to the control problem for the $Z$-investor  are available already for moderate numbers $n$ of expert opinions. Since the latter require less computational effort this is very helpful for deriving approximations not only for the value functions but also for related quantities such as efficiencies and prices of expert opinions  introduced in Subsec.~\ref{MonetaryValue} and considered in the next subsection. 

\begin{figure}[h!]
	\hspace*{-0.1\textwidth}
	\includegraphics[width=1.15\textwidth]{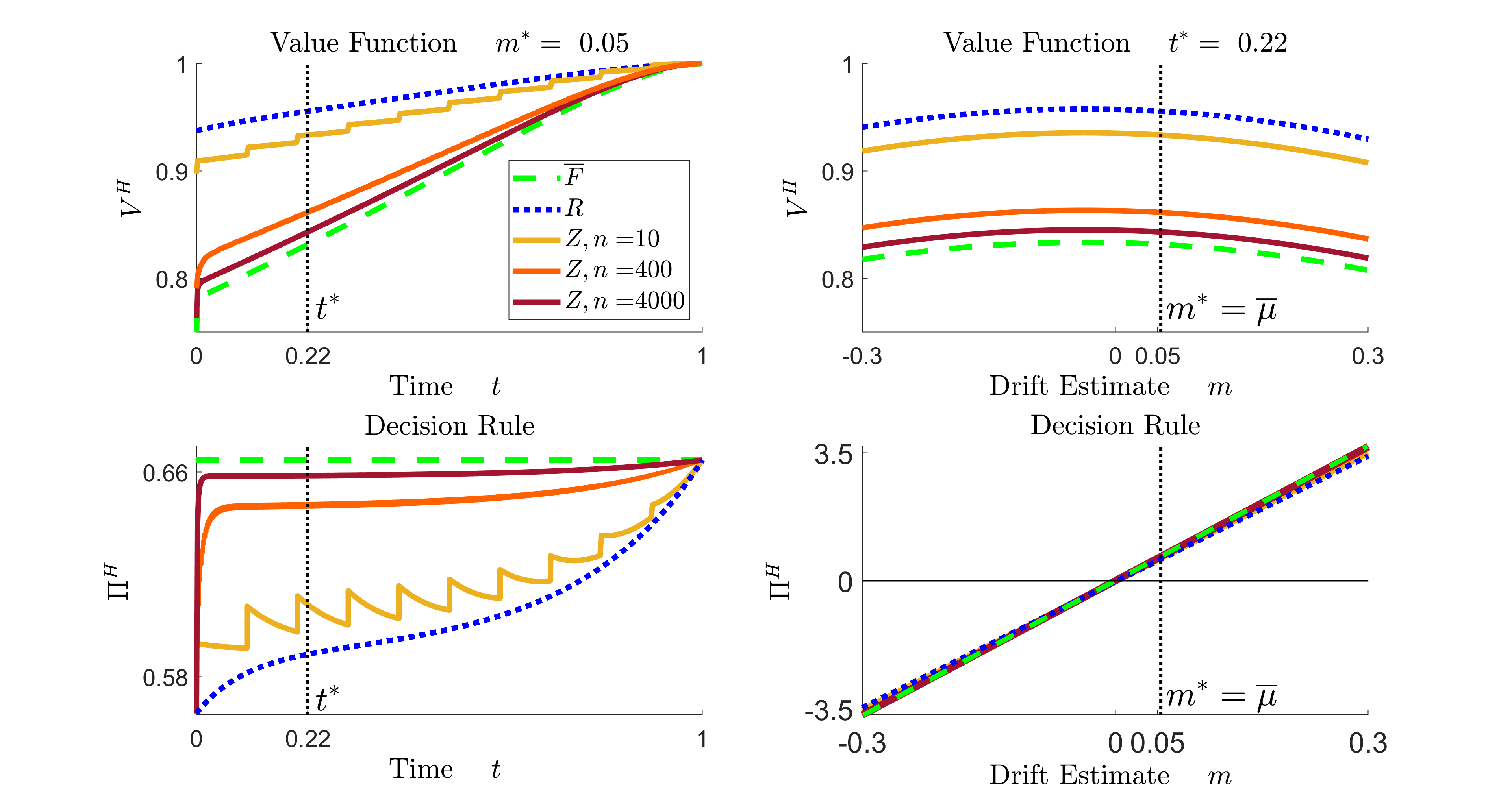}
	
	\centering \caption{\label{positive_Konvergenz_CnF}
		Asymptotic behavior of value functions and optimal decision rules for growing $\nExperten$ and  $~\Gamma=0.4$ (fixed).
		\newline
		Top:~~~~ Value function $\overline{\valuefkt}^{\HF}(t,m,\Qpro_t^{{\HC},{4000}})$ and ${\valuefkt}^H(t,m)$ depending on $t/ \zustand$ (left/right) for  $H=\HR,\HC$.
		\newline
		Bottom:  Optimal decision rules ${\Pi}^H(t,m)$ for $H=F,R,Z$ depending on $t/ \zustand$ (left/rights).
	}
\end{figure}

Fig.~\ref{positive_Konvergenz_CnF} shows results of the experiment with fixed  variance $\varianceexp=0.4$ for which we have convergence to full information, i.e.,  {mean-square convergence of $\Mpro^{Z,n}$ to $\drift$} and  $\Qpro^{Z,n}\to 0$ on $(0,T]$.  As in Fig.~\ref{positive_Konvergenz_CnD} we plot ${\valuefkt}^{Z,n}$ and ${\Pi}^{Z,n}$ against time $t$ and  drift estimate $m$,  but now for  $\nExperten=10,400,4000$.  We expect that 
$\vert {\valuefkt}^{Z,n}(t,\zustand)-\overline{\valuefkt}^{\HF}(t,\zustand,\Qpro_t^{Z,n})\vert$ converges to zero where $\overline{\valuefkt}^{\HF}$ is the conditional expectation of the  value function  $V^F(t,\mu_t)$  of the fully informed investor given $\mathcal{F}_t^{Z,n}$.  That function is introduced in  \eqref{mitteled_VF} and  Lemma \ref{F_mittel} provides a closed-form expression.  The upper panels show for comparison $\overline{\valuefkt}^{\HF}(t,\zustand,\Qpro_t^{Z,4000})$  and also the value function of the $R$-investor while the bottom panels also show the decision rules of the $R$- and $F$-investor.  The latter is independent of time $t$ and defines the  myopic decision rule. We observe that for increasing $n$ the value function and the optimal decision rule  of the $Z$-investor  approach $\overline{\valuefkt}^{\HF}$ and the  myopic decision rule, respectively. However, compared to the case of linearly growing expert's variance (see Fig.~\ref{positive_Konvergenz_CnD}) the convergence is much slower.		
This was already observed in Subsec.~\ref{Num_Filter} for the convergence of filter processes.

We note again that for the chosen parameters  we have for $n=10$ expert opinions that $\varianceexp=\varianceexp^{(n)}=\frac{n}{T}\volexp^2=0.4$.  This yields that for  $n=10$ the  value function and decision rule  for the experiment with linear growing expert's variance $\varianceexp^{(n)}$ coincide with those for the experiment with constant variance $\varianceexp$.

\subsection{Monetary Value of Information}
\label{Monetary-Value_Num}
We conclude this section with some results of experiments illustrating the concepts of  efficiency and price of expert opinions introduced in Subsec.~\ref{MonetaryValue} for the description of the monetary value of information. 

\paragraph{Efficency}
Recall that we followed an utility indifference approach and considered  the initial capital $x_0^{H/\HF}$ which the fully informed $\HF$-investor needs to
obtain the same maximized expected utility at time $T$ as the
partially informed $H$-investor who started at time $0$ with wealth $x_0^H>0$. That wealth is given in Eq.~\eqref{monetaer_Wert} as the solution of the equation  
$\mathcal V_0^{H}(x_0^H)=\E \big[ \mathcal V_0^{\HF}\big(\text{$x$}_0^{H/ \HF}\big)\mid \mathcal F_0^{H}\big]$ for $H=\HR,\HC, \HD$.
The difference $ x_0^H-x_0^{H/ \HF}>0$ describes the loss of information for the partially  informed $H$-investor relative to the $\HF-$investor measured in
monetary units. The ratio $\varepsilon^{H}={x_0^{H/ \HF}}/{ x_0^H} \in (0,1]$
introduced in \eqref{effekt_formel} is a measure for the efficiency of the $H$-investor. We refer to Lemma  \ref{monetary_value_RJ} and \ref{monetary_value_Z} where we give explicit expressions for the above quantities for $H=\HR,\HD$ and $H=\HC$, respectively. 

In Fig.~\ref{fig:efficiency} we compare the efficiencies of the $Z$-investor for increasing $n$ and   parameter of the utility function $\theta=\pm0.3$. In the left panel the expert's variance is kept constant and equal to $\varianceexp=0.4$. Then the $\HC$-investor asymptotically for $n\to \infty$ has full information about the hidden drift. The figure shows that the $\HC$-investor's efficiency increases with $n$ starting with the efficiency of the $R$-investor (blue) and approaching $1$ which is the efficiency of the fully informed investor (green). Note that the investment horizon is $T=1$ year such that arrival of the expert opinions once per year,  month,  week,  day, hour or
minute corresponds to $\nExperten=1, 12, 52, 365, 8.760$ or  $525.600$, respectively. Comparing the efficiencies for different parameters $\theta$ it can be seen that an investor with the positive parameter $\theta=0.3$, i.e.,   less risk averse than the log-utility investor ($\theta=0$),  achieves smaller efficiencies than an investor with the negative parameter $\theta=-0.3$. Note that the latter is   more  risk averse than the log-utility investor. Additional experiments have shown that the efficiency increases with increasing risk aversion $1-\theta$.

\begin{figure}[h!]
	\hspace*{-0.05\textwidth}
	\includegraphics[width=0.56\textwidth,height=0.35\textwidth]{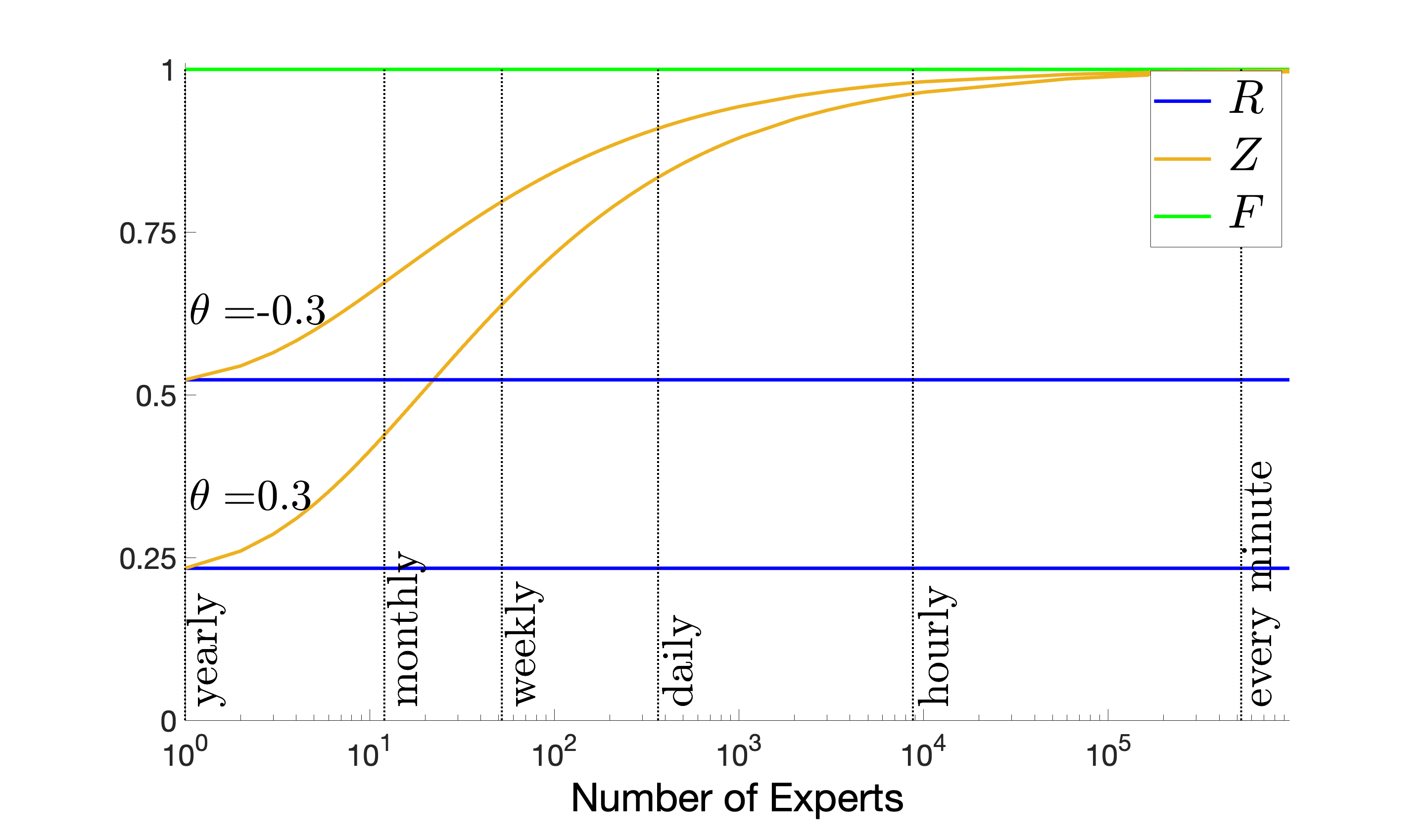}
	\hspace*{-0.06\textwidth}
	\includegraphics[width=0.56\textwidth,height=0.35\textwidth]{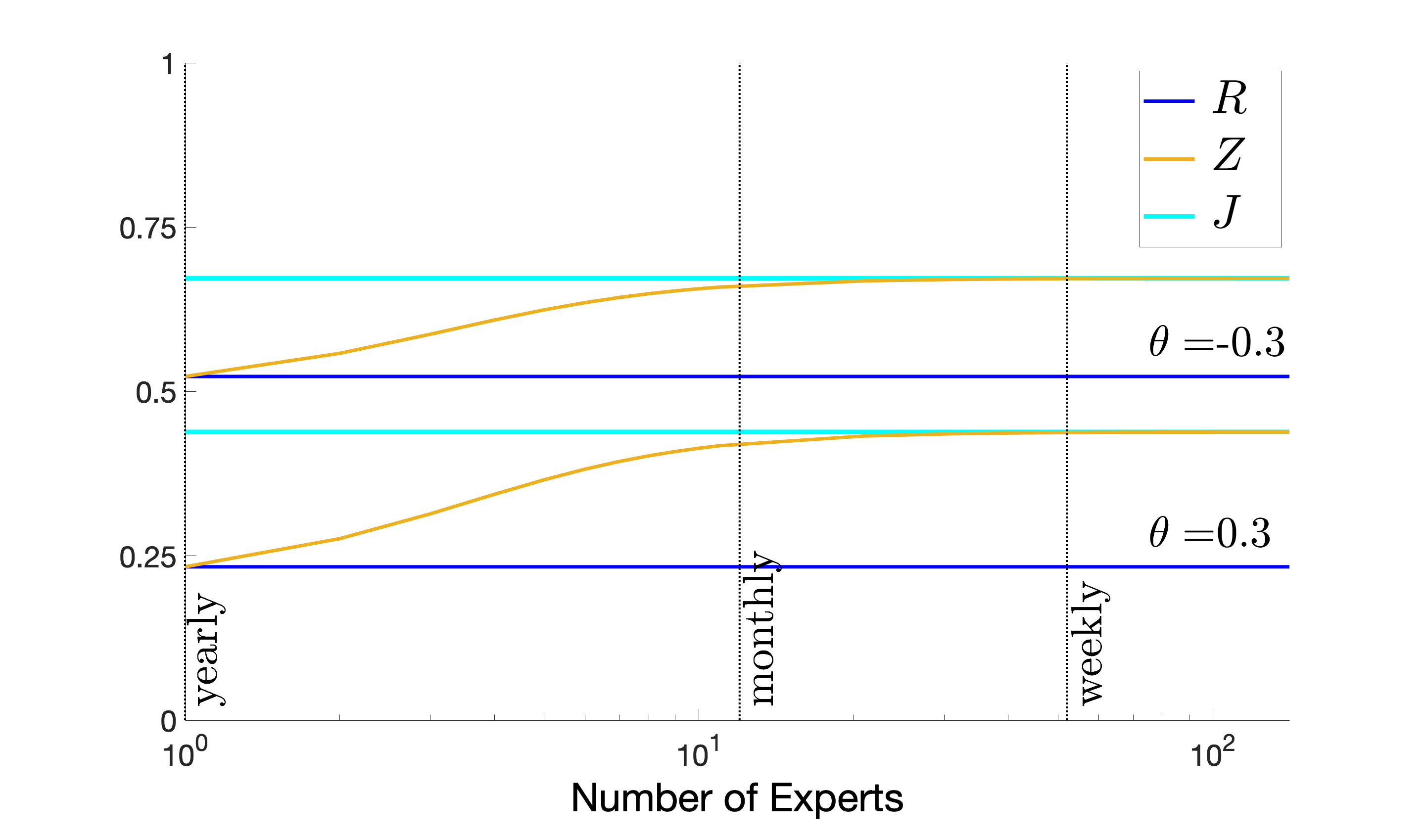}

	\centering \caption{\label{fig:efficiency}
		Efficiency of the $Z$-investor for increasing $n$ and   power utility function with $\theta=\pm0.3$:
		\newline
		Left: ~ Expert's variance $\Gamma= 0.4$ fixed;
		\newline
		Right: Expert's variance $\Gamma^{(n)}=\frac{\nExperten}{T}\volexp^2$ linearly growing.		
	}
\end{figure} 

In the right panel in Fig.~\ref{fig:efficiency} we show results of experiments in which the  expert's variance $\varianceexp$ grows linearly with $n$. In that setting we  expect convergence to the diffusion limit represented by the $\HD$-investor. Here, the $\HC$-investor's efficiency again increases with $n$ starting with the efficiency of the $R$-investor (blue) but now approaches the efficiency  $\varepsilon^{\HD}$ of the $\HD$-investor (light blue) which is less than $1$. As already observed for the value functions in Subsec.~\ref{Numerik_Konvergenz_Wertfunktion} that convergence is much faster than the convergence to full information for fixed $\varianceexp$. The diffusion limit  $\varepsilon^{\HD}$ provides quite accurate approximations for  $\varepsilon^{\HC,n}$ already for $n\approx 50$, i.e., weekly expert's views.

\smallskip

\paragraph{Price of the experts}
\begin{figure}[h]
	\hspace*{-0.05\textwidth}
	
	\includegraphics[width=0.56\textwidth,height=0.35\textwidth]{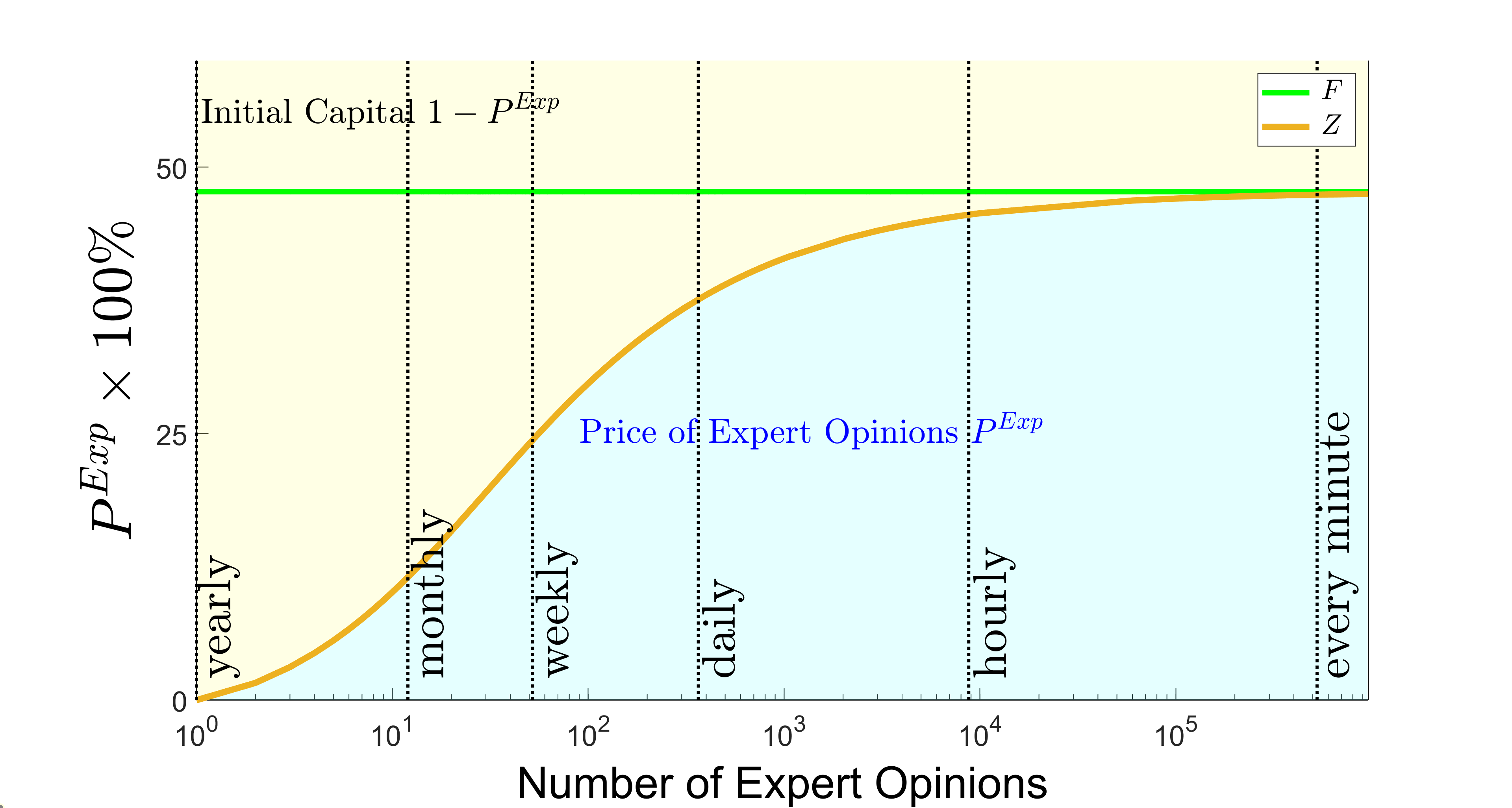}
	\hspace*{-0.06\textwidth}
	\includegraphics[width=0.56\textwidth,height=0.35\textwidth]{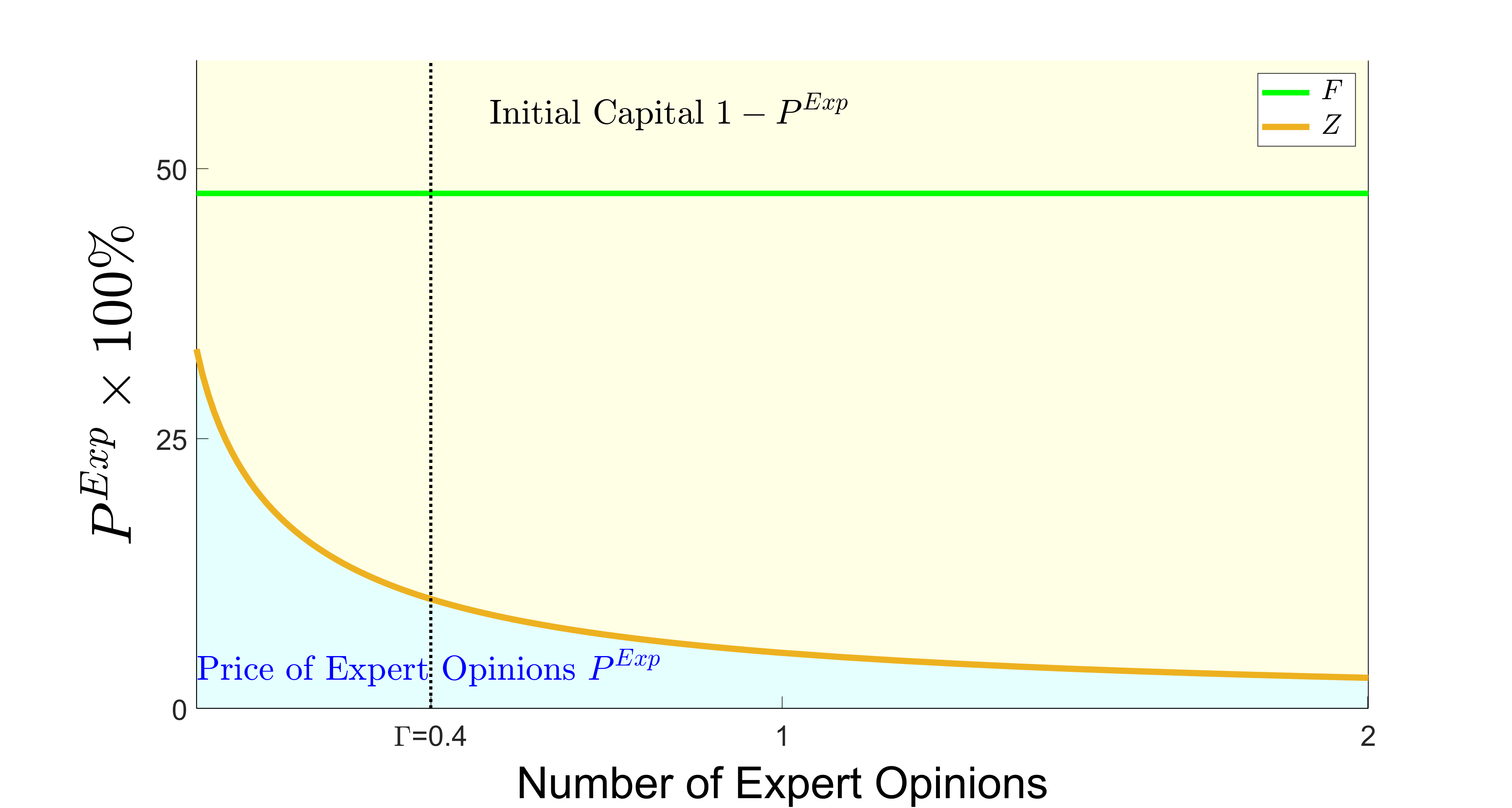}
	
	\centering \caption{\label{fig:price_experts}
		Price of the experts for  power utility with $\theta=-0.3$: 
		\newline
		Left: ~ Increasing number $n$  of expert opinions and expert's variance $\Gamma=0.4$ fixed; 
		\newline
		Right: Increasing expert's variance $\Gamma$ and $n=10$ fixed.				
	}
\end{figure}
In Subsec.~\ref{MonetaryValue} we also used the utility indifference approach to derive a measure for the 
monetary value of the additional information delivered by the experts. The idea was to equate 
the maximum expected utilities  of an  $\HR$-investor who only observes returns
of the $H$-investor for $H=\HC,\HD$. The latter  combines
return observations with information from the experts.  Given
the $\HR$-investor is equipped with initial capital
$x_0^{\HR}>0$ one computes  the initial capital $x_0^{\HR/H}\le x_0^{\HR}$
for the $H$-investor which leads to the same maximum expected
utility, we refer to Eq.~\eqref{initial_cap_H}.	
Then the  $R$-investor could put aside from its initial capital $x_0^{\HR}$  the amount $P_{Exp}^{H}:= x_0^{\HR}-x_0^{\HR/H}\ge 0$ to buy the
information from the expert. The remaining capital  $x_0^{\HR/H}$
is invested in  an $H-$optimal portfolio and providing  the same
expected utility of terminal wealth as the $\HR$-optimal portfolio
starting with initial capital $ x_0^{\HR}$. We refer to Lemma  \ref{monetary_value_RJ} and \ref{monetary_value_Z} where we give explicit expressions for $P_{Exp}^{H}$ for $H=\HD$ and $H=\HC$, respectively.

Fig.~\ref{fig:price_experts} shows the above decomposition of the initial capital of the $R$-investor for $x_0^{\HR}=1$. In the left panel we fix the expert's variance $\varianceexp=0.4$ and plot $P_{Exp}^{\HC,n}$ against $n$. As expected that price increases with the number of expert opinions but for $n\to \infty$, i.e., in the full information limit, the price reaches  a saturation level which is given by $x_0^{\HR}-x_0^{\HR/F}=1-\varepsilon^{\HR}$. 

The right panel shows results for fixed  $n=10$ but growing variance $\varianceexp$. Then the expert's views provide less and less information about the hidden drift leading to a decreasing price $P_{Exp}^{\HC}$ approaching zero for $\varianceexp\to \infty$, i.e., for fully non-informative expert's views. On the other hand in the limiting case for $\varianceexp\to 0 $ at each of the $n=10$ information dates the $\HC$-investor has full information about the drift process.  Note that full information is not available for all for all $t\in(0,T]$ but only at finitely many information dates  and thus  $P_{Exp}^{\HC}$ is for $\Gamma\to 0$ moving towards but not reaching the full information limit $1-\varepsilon^{\HR}$.

\bmhead{Acknowledgments}{ \small The authors thank Dorothee Westphal and Jörn Sass  (TU Kaiserslautern) and  Benjamin Auer (Friedrich Schiller University Jena) for valuable discussions that improved this paper. 



\let\oldbibliography\thebibliography
\renewcommand{\thebibliography}[1]{%
	\oldbibliography{#1}%
	\setlength{\itemsep}{.15ex plus .05ex}%
}
\bibliographystyle{amsplain}

\begin{thebibliography}{99}
	\footnotesize
	
	
	\bibitem{Angoshtari2013}
	Angoshtari, B.:
	\newblock  {Stochastic Modeling and Methods for Portfolio Managment in  Cointegrated Markets}.
	\newblock{\it PhD thesis, University of Oxford}, 
	\url{https://ora.ox.ac.uk/objects/uuid:1ae9236c-4bf0-4d9b-a694-f08e1b8713c0} (2014).
	
	\bibitem{Angoshtari2016}
	Angoshtari, B. (2016): Portfolio
	On the Market-Neutrality of Optimal Pairs-Trading Strategies.  arXiv:1608.08268 [q-fin.PM].    
	
	\bibitem{Battauz et al (2017)} 
	Battauz, A., De Donno, M. and Sbuelz, A. (2017): Reaching Nirvana With a Defaultable Asset?. {\it Decisions in Economics and Finance} 40, 31–52. 
	
	
	\bibitem{Bensoussan (1992)}		 
	Bensoussan, A.: {\it Stochastic Control of Partially Observable Systems}, Cambridge University Press, Cambridge (1992).
	
	
	\bibitem
	{Bielecki_Pilska (1999)}
	Bielecki, T.R,  and Pliska, S.R: Risk-sensitive dynamic asset management.
	{\it Applied Mathematics and Optimization} 39 , 337–360 (1999).
	
	\bibitem
	{Bjoerk et al (2010)} Bj\"ork, T., Davis, M. H. A. and Land\'en, C.: Optimal
	investment with partial information, {\it Mathematical Methods of
		Operations Research} {\bf 71} , 371--399 (2010).
	
	\bibitem
	{Black_Litterman (1992)}
	Black, F. and Litterman, R.: Global portfolio optimization.
	{\it Financial Analysts Journal} 48(5), 28-43 (1992).
	
	
	
	
	\bibitem
	{Brendle2006}
	Brendle, S.: Portfolio selection  under incomplete
	information. {\it Stochastic Processes and their Applications},
	116(5), 701-723 (2006).
	
	
	\bibitem
	{Brennan et al (1997)}
	Brennan, M.J., Schwartz, E.S., Lagnado, R.: Strategic asset allocation. {\it Journal of Economic Dynamics and Control} 21 (8), 1377–1403 (1997). 
	
	\bibitem
	{Broadie (1993)}
	Broadie, M.: Computing efficient frontiers using estimated parameters. {\it Annals of Operations	Research} 45, 21–58 (1993). 
	
	\bibitem
	{Chen Wong (2022)}
	Chen, K. and Wong, Y.:
	Duality in optimal consumption--investment problems with alternative data.  
	arXiv:2210.08422 [q-fin.MF] (2022).
	
	\bibitem
	{Colaneri et al (2021)}
	Colaneri, K., Herzel, S. and Nicolosi, M.: The value of
	knowing the market price of risk. {\it Annals of Operations	Research} 299, 101–131 (2021). 
	
	
	
	\bibitem
	{Davis and Lleo (2013_1)} Davis, M. and Lleo, S.:
	Black-Litterman in continous time: the case for filtering. {\it
		Quantitative Finance Letters}, 1,30-35 (2013).
	
	
	\bibitem
	{Davis and Lleo (2014)} Davis, M. and Lleo, S.:
	{\it Risk-Sensitive Investment Management,} World Scientific, Singapore (2013).
	
	
	
	
	
	\bibitem{Davis and Lleo (2020)} 
	Davis, M. and Lleo, S.:
	Debiased expert opinions in continuous-time asset allocation.
	{\it Journal of Banking \& Finance,} Vol. 113, 105759  (2020).
	
	\bibitem{Davis and Lleo (2022)} 		
	Davis, M. and Lleo, S.:
	Jump-diffusion risk-sensitive benchmarked asset
	management with traditional and alternative data.
	{\it Annals of Operations	Research}  (2022). \url{DOI: 10.1007/s10479-022-05130-3}
	
	
	
	\bibitem
	{Fleming and Rishel (1975)} Fleming, W. and Rishel, R.:
	{\it Deterministic and Stochastic Optimal Control ,} Springer New York (1975).
	
	
	
	
	
	\bibitem
	{Fouque et al. (2015)}
	Fouque, J. P, Papanicolaou, A. and Sircar, R.: Filtering and
	portfolio optimization with stochastic unobserved drift in asset
	returns. {\it SSRN Electronic Journal},13(4):935-953, (2015).
	
	\bibitem
	{Frey et al. (2012)}
	Frey, R., Gabih, A. and Wunderlich, R.: Portfolio
	optimization under partial information with expert opinions.  {\it
		International Journal of Theoretical and Applied Finance}, 15,
	No.~1, (2012).
	
	\bibitem
	{Frey-Wunderlich-2014}
	Frey, R., Gabih, A. and Wunderlich, R.:  Portfolio
	optimization under partial information with expert opinions: Dynamic
	programming approach. {\it Communications on Stochastic Analysis}
	Vol. 8, No. 1, 49-71 (2014).
	
	
	\bibitem
	{Gabih et al (2014)}
	Gabih, A., Kondakji, H. Sass, J. and Wunderlich, R.: Expert
	opinions and logarithmic utility maximization in a market with
	Gaussian drift. {\it Communications on Stochastic Analysis} Vol. 8,
	No. 1,27-47 (2014).
	
	
	
	\bibitem
	{Gabih et al (2019) FullInfo}
	Gabih, A., Kondakji, H. and Wunderlich, R.: Asymptotic filter
	behavior for high-frequency expert opinions in a market with
	Gaussian drift. {\it Stochastic Models} 36(4), 519-547 (2020).
	
	\bibitem
	{Gabih et al (2022) Nirvana}
	Gabih, A., Kondakji, H. and Wunderlich, R.: 
	Well posedness of utility maximization problems under partial information in a market with Gaussian drift. submitted,  
	arXiv:2205.08614 [q-fin.PM] (2022).
	
	
	\bibitem
	{Gabih et al (2022) PowerRandom}
	Gabih, A. and Wunderlich, R.: 	Portfolio optimization in a market with hidden Gaussian drift and randomly arriving expert opinions: Modeling and theoretical results.   
	arXiv:2308.02049 [q-fin.PM]  (2023).
	
	
	
	\bibitem
	{Hata Sheu (2018)}
	Hata, H. and Sheu, S.J.: An optimal consumption and investment problem with partial information. {\it Advances in Applied Probability} 50, 131-153 (2018).
	
	
	
	\bibitem
	{Kim and Omberg (1996)}
	Kim, T. S. and Omberg, E.: Dynamic non myopic portfolio behavior,
	{\it The Review of Financial Studies} 9(1), 141-16 (1996).
	
	
	\bibitem
	{Kondkaji (2019)}
	Kondakji, H.: Optimal Portfolios for Partially Informed
	Investors in a Financial Market with Gaussian Drift and Expert
	Opinions (in German). PhD Thesis BTU Cottbus-Senftenberg. Available
	at
	\url{https://opus4.kobv.de/opus4-btu/frontdoor/deliver/index/docId/4736/file/Kondakji_Hakam.pdf},  (2019).
	
	\bibitem{Korn and Kraft (2004)}
	Korn, R. and Kraft, H. (2004): On the Stability of Continuous-Time Portfolio Problems with Stochastic Opportunity Set,
	{\it Mathematical Finance} 14(3), 403-414.
	
	
	
	\bibitem
	{Lakner (1998)} Lakner, P.: Optimal trading strategy for an investor: the case of
	partial information. {\it Stochastic Processes and their
		Applications} 76, 77-97 (1998).
	
	
	\bibitem
	{Lee Papanicolaou (2016)} 
	Lee, S. and  Papanicolaou A.:
	Pairs trading of two assets with uncertainty in co-integration's level of mean reversion.	
	{\it International Journal of Theoretical and Applied Finance} 19 (08), 1650054 (2016)
	
	
	\bibitem
	{Liptser-Shiryaev} 
	Liptser, R.S. and Shiryaev  A.N.: {\em Statistics of Random Processes:
		General Theory}, 2nd edn, Springer, New York (2001).
	
	\bibitem
	{Merton (1971)} Merton, R.: Optimal consumption and portfolio rules in a continuous time model. {\it Journal of Economic Theory} 3, 373-413 (1971).		
	
	
	\bibitem
	{Nagai (2015)}
	Nagai, H.: H–J–B equations of optimal consumption-investment and verification theorems. {\it Applied Mathematics \& Optimization} 71, 279–311 (2015).
	
	
	
	\bibitem
	{Nagai and Peng	(2002)} 
	Nagai, H. and Peng, S.: Risk-sensitive dynamic	portfolio optimization with partial information on infinite time
	horizon. {\it The Annals of Applied Probability} 12, 173-195 (2002).
	
	
	\bibitem{Pham  (1998)} Pham, H.: Optimal stopping of controlled jump diffusion processes: A
	viscosity solution approach. {\it Journal of Mathematical Systems,
		estimations, and Control} Vol.8, No.1, 1-27 (1998).
	
	\bibitem
	{Putschoegl and Sass (2008)}
	Putsch\"ogl, W. and Sass, J.: Optimal consumption and investment
	under partial information, {\it Decisions in Economics and Finance} (31), 131--170 (2008).
	
	\bibitem
	{Rieder_Baeuerle2005} Rieder, U. and B\"auerle, N.:
	Portfolio optimization with unobservable Markov-modulated drift process. {\it Journal of Applied Probability} { 43}, 362-378 (2005).
	
	
	\bibitem
	{Roduner (1994)}		
	Roduner, Ch.: Die Riccati-Gleichung, {\it IMRT-Report} Nr. 26, ETH Zurich (1994). 
	
	
	\bibitem
	{Rogers (2013)}
	Rogers, L.C.G.: {\it Optimal Investment.} Briefs in Quantitative Finance. Springer, Berlin-Heidelberg, (2013).
	
	
	\bibitem
	{Sass and Haussmann (2004)}
	Sass, J. and Haussmann, U.G.: Optimizing the terminal wealth
	under partial information: The drift process as a continuous time
	Markov chain. {\it Finance and Stochastics} (8), 553-577 (2004).
	
	\bibitem
	{Sass et al (2017)}
	Sass, J.,  Westphal, D.  and Wunderlich, R.: Expert opinions
	and logarithmic utility maximization for multivariate stock returns
	with Gaussian drift. {\it International Journal of Theoretical and Applied Finance} Vol. 20, No.~4, 1-41 (2017).
	
	
	
	\bibitem
	{Sass et al (2021)}
	Sass, J.,  Westphal, D.  and Wunderlich, R.: Diffusion
	approximations for randomly arriving expert opinions in a financial
	market with Gaussian drift. {\it Journal of Applied Probability}, 58(1), 197 - 216 (2021).
	
	\bibitem
	{Sass et al (2022)}
	Sass, J.,  Westphal, D.  and Wunderlich, R.:
	Diffusion approximations for periodically arriving expert opinions in
	a financial market with Gaussian drift. Stochastic Models (2022), \url{DOI: 10.1080/15326349.2022.2100423}
	
	
	
	\bibitem
	{Schoettle et al. (2010)}
	Sch\"ottle, K., Werner, R. and Zagst, R.: Comparison and
	robustification of Bayes and Black-Litterman models. {\it
		Mathematical Methods of Operations Research}  71, 453--475 (2010).
	
	\bibitem{Xia (2001)}
	Xia, Y.: Learning about Predictability:
	The Effects of parameter uncertainty
	on dynamic asset allocation. {\it Journal of  Finance,} Vol. 56, No. 1, 205-246.
	
	
\end{thebibliography}

\end{document}